\documentclass[12 pt]{report}
\usepackage[margin=2.5cm]{geometry}
\usepackage{multicol}

\usepackage{graphicx}
\usepackage{dcolumn}
\usepackage[hidelinks]{hyperref}
\usepackage[english]{babel}
\usepackage[autostyle, english = american]{csquotes}
\MakeOuterQuote{"}
\usepackage{amsthm}
\usepackage{eufrak}
\usepackage[mathscr]{euscript}
\usepackage{amssymb}
\usepackage{bm}
\usepackage{physics}
\usepackage{comment}

\usepackage{lipsum}
\usepackage{setspace}

\usepackage[backend = bibtex]{biblatex}
\newtheorem{theorem}{Theorem}[section]
\theoremstyle{definition}

\newtheorem{claim}{Claim}[section]

\theoremstyle{remark}

\numberwithin{equation}{section}

\usepackage[p,osf, theoremfont]{cochineal}
\usepackage[varqu,varl,var0]{inconsolata}
\usepackage[cochineal,vvarbb]{newtxmath}
\usepackage[scr = euler, cal = stixplain]{mathalfa}

\DeclareMathOperator\supp{supp}

\begin{document}

\onehalfspacing
\pagenumbering{roman} 

\begin{titlepage}

\onehalfspacing

\begin{center}
\begin{figure}[h]
\centering
\includegraphics[scale=0.5]{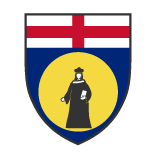}
\end{figure}

\Large{Università di Genova\\
Scuola di Scienze Matematiche, Fisiche e Naturali\\}

\vspace{5 mm}

\large{Master Degree Course in Physics}

\vspace{10mm}
\hrule

\vspace{15 mm}
%\Large{Master Degree Thesis}
%
% \fontfamily{Roboto Slab}\selectfont \Large{Master Degree Thesis}
%
%\LARGE{Università degli Studi di Genova}

\vspace{1\baselineskip}

\Large{ \textbf{Relative Entropy from Coherent States in Black Hole Thermodynamics and Cosmology}}\\

\end{center}

\vspace{10mm}

\par

\noindent

\begin{minipage}[t]
{\textwidth}
{ \large{{Supervisors:}}}\\
Prof. Nicola Pinamonti \\
Prof. Pierantonio Zanghì

\end{minipage}

\vspace{5 mm}
\begin{minipage}[t]
{\textwidth}
{\large{Co-Supervisor:}}\\
Prof. Camillo Imbimbo

\end{minipage}

\vspace{10 mm}

\begin{minipage}[t]
{\textwidth} \raggedleft
{\large{Candidate:}}\\
Edoardo D'Angelo
\end{minipage}

\vspace{10 mm}

\begin{center}
\large{\scshape ACADEMIC YEAR 2019-2020}
\end{center}

\end{titlepage}
\newpage

\newpage

\chapter*{Abstract}

The aim of this work is to study the role of relative entropy in the thermodynamics of black holes and cosmological horizons. Since the seminal paper of Bombelli et al. (1986) \cite{Bombelli86}, many attempts have been made to characterise the Bekenstein-Hawking entropy as the entanglement entropy of quantum degrees of freedom separated by the event horizon. Here we adapted some recent results by Ciolli et al. (2019) \cite{CiolliLongoRuzzi19}, and Casini et al. (2019) \cite{CasiniGrilloPontiello19} for the relative entropy of coherent excitations of the vacuum computed with the Tomita-Takesaki modular theory, to find the variation of generalised entropy of static and dynamical black holes and for cosmological horizons.  The main result we use is the explicit formula of the entanglement entropy in terms of the symplectic form of the classical field theory.

We review the argument for static black holes by Hollands and Ishibashi (2019) \cite{HollandsIshibashi19} with simple modifications, in particular to compute the entanglement entropy of coherent states with respect to the Unruh state, the most physical state for a black hole in formation. The entanglement entropy is computed at past infinity, where asymptotic flatness guarantees that we can apply the results found in flat spacetime. We link the variation of relative entropy to the growth of the horizon using a conservation law for the stress-energy tensor, and we recover the formulas of Hollands and Ishibashi. We then study the application of the same framework to the case of apparent horizons, which are local structure first introduced in the study of dynamical black holes. We study in detail the case of Vaidya spacetime, being the simplest generalisation of a Schwarzschild black hole to the dynamical case, and we find that a notion of black hole entropy naturally emerges, equals to one-fourth of the area of the apparent horizon. In the case of dynamical black holes we find that the variation of generalised entropy (defined as the matter relative entropy plus one-fourth of the horizon area) equals a flux term, plus an additional term, which is not present in the static case, and that can be interpreted as a work term done on the black hole. We finally show in a simple case that it is possible to follow the same scheme to assign an entropy to the horizons emerging in cosmological scenarios.
 
The work is organised as follows. We sum up the relevant tools from General Relativity, i.e., the Raychaudhuri equation for null geodesic congruences and the Gauss-Stokes Theorem. We review the algebraic approach to the free, neutral, scalar field propagating in globally hyperbolic spacetimes, with emphasis on the connections to the usual mode-decomposition approach, and the quasifree and Hadamard conditions for the physical states. We describe the KMS characterization of thermal states and its connection to the Tomita-Takesaki modular theory, which is used to define the relative entropy for quantum field theories via the Araki formula. We see that the relative entropy for coherent perturbations can be explicitly computed using the classical structure only. We then apply the formula for the relative entropy in the specific examples of Schwarzschild, Vaidya, and cosmological horizons.

\tableofcontents

\chapter*{Ringraziamenti}

Scrivere una tesi non è un lavoro individuale: assomiglia di più alla costruzione delle Piramidi, alle imprese che richiedono secoli e migliaia di persone. Per il dispiegamento di tutti i mezzi materiali ed emotivi necessari a farmi arrivare in fondo, anche nei momenti di peggiore autocommiserazione, a tutte le persone che mi sono state vicine, che hanno scritto questa tesi con me, va questa prima pagina.

Innanzitutto, ringrazio i miei relatori, i professori Nicola Pinamonti e Pierantonio Zanghì, per le infinite chiamate Skype, le discussioni, le indicazioni e i suggerimenti con cui abbiamo costruito questa tesi a distanza. Mi hanno insegnato nel modo migliore che fare scienza è, innanzitutto, una conversazione fatta d'idee, errori e immaginazione: senza il loro continuo supporto, scientifico, e, a volte soprattutto, psicologico, non avrei scritto una sola parola di questa tesi. Ringrazio anche il mio correlatore, il professore Camillo Imbimbo, per il prezioso aiuto e la pazienza con cui ha ascoltato la versione preliminare di questa tesi. Grazie per esservi dedicati a questo progetto perfino alle 9 di domenica mattina, praticamente a metà della notte.

Ringrazio poi tutti i miei compagni di corso, il nostro sensei Giulio, i compagni Martino, Giacomo, Alessandro e Dorwal, e tutti gli amici del \textit{Bit di Fisica in più}, perché non era mai chiaro quando finivano la Fisica, i corsi, gli esami e lo stress e cominciava l'amicizia. Grazie per questi anni condivisi insieme in Piccionaia, divorando panini nelle pause fra le lezioni e i pomeriggi di laboratorio, che per questo sono stati molto di più che anni di studio.

Ringrazio la mia famiglia, Mamma, Papà, Beatrice, Lucrezia, Anna, Luca, i miei nonni e i miei zii, e, mi sento di aggiungere, Matteo e Francesco, che durante i tre mesi d'isolamento hanno scoperto i tic e le nevrastenie di uno studente universitario agli sgoccioli (in tutti i sensi), e non hanno smesso di incoraggiarmi anche quando mi hanno visto camminare per casa borbottando mezze frasi sugli spazi quadridimensionali. Grazie per l'entusiasmo con cui avete accolto gli aggiornamenti quotidiani sul numero di pagine, puntuali come il bollettino meteo, e per avermi fatto prendere una boccata d'aria quando rischiavo di annegare fra i conti sbagliati.

Grazie, infine, ad Elena. Grazie per scoprire sempre la parte migliore di me, per essere il mio sostegno e la mia luce. Grazie per essere la migliore compagna con cui avventurarsi in questo Universo, e di aver scelto di condividere l'unica vita che abbiamo.

Ho avuto il privilegio di studiare e fare Fisica mentre il mondo era in fiamme: di potere uscire dall'isolamento per esplorare liberamente i misteri dei buchi neri, stupendomi per la meraviglia del cosmo. Credo che la scienza e i suoi valori possano aiutare a spegnere l'incendio che le crisi scoppiate in questa primavera (crisi sanitaria, climatica, della giustizia e della democrazia) fanno sembrare inarrestabile, e per questo ringrazio tutti le persone che hanno reso possibile questo lavoro.

\newpage
\
\newpage

\chapter*{Introduction. Bekenstein's Quantum Cup of Tea} \label{sec:intro}
\addcontentsline{toc}{chapter}{\nameref{sec:intro}}

\pagenumbering{arabic}
\setcounter{page}{9}

Let's start with one of our favourite equations, 
\begin{equation*}
S_{BH} = \frac{c^3 k_B}{G \hbar} \frac{A}{4}
\end{equation*}
written, for the first and last time, with the full set of constants it carries in SI units: the speed of light $c$, the Boltzmann constant $k_B$, the Newton's constant $G$, and the reduced Planck's constant $\hbar$. Even without knowing what $S_{BH}$ or $A$ are, the above equation seems puzzling: it seems to encompass all the branches of Physics, each one represented by the associated constant. There is thermodynamics, with the Boltzmann constant, and Quantum Mechanics, because of the Planck's constant. This can happen, in studying quantum materials at finite temperature. However, the presence of the speed of light and of Newton's constant denounce that the equation includes gravity and relativity too, and this is where a physicist becomes really interested in the equation: usually, Quantum Mechanics and General Relativity describe different worlds, the former being the theory from the tiniest constituents of matter and their interactions to the structure of materials, the latter being the theory of the motions of stars and galaxies, and of the structure of the Universe at its largest scales. Quantum Mechanics and General Relativity belong to different realms of Nature, and are rarely used together: because when the quantum effects are important, gravity can usually be neglected, and vice versa, because in the study of the largest structure of the Universe, the strange behaviour of matter at its microscopic scales is usually not relevant.

If the curious physicist would continue with a dimensional analysis, she could discover that the combination $c^3 /G \hbar $ is the inverse of a squared length, the so-called Planck length, roughly $l_P = \sqrt{G \hbar / c^3} \simeq 10^{-35} \ m$. On the other hand, the Boltzmann's constant has units of an entropy, $J K^{-1}$; then, maybe guided by the choice of symbols for $S_{BH}$ and $A$, she could come with a well-educated guess (suspiciously too good, we are afraid to admit), that the above equation assigns an entropy, $S_{BH}$, proportional to one-fourth of the area of some system, $A$. But, far from being resolved, the mystery thickens: at the very least, since the entropy is a measure of the microstates of a system, if we assign a bit of information to each small cell of the system, we would expect the entropy to be proportional to its volume, not to the area. Besides, entropy is a thermodynamic quantity, which has to do with energy and temperature, not with geometric quantities like the area. Still, it is far from clear the role that Quantum Mechanics and gravity should play in such a system.

As it turns out, the above equation was one of the first discoveries suggesting that gravity and Quantum Mechanics can conjure up in surprising ways, giving rise to phenomena which can be explained only by a combination of quantum and gravitational effects, and cannot be described using one of the two theories only. Physicists are discovering more and more that, in order to explain the phenomena we see, we need to simultaneously take into account both General Relativity and Quantum Mechanics: to study the extreme phases of the Universe in its first moments, for example, or to investigate the mysterious interactions between elementary particles at the largest scales of energy, well beyond the reach of current accelerators. The above equation was one of the first of such examples, perhaps the most important. The story of its discovery has to do with a couple of brilliant students and their professors, black holes, and a cup of tea.

Once upon a time there were a venerable professor and an audacious student talking of black holes in front of a cup of tea \cite{Oppenheim15}. They were discussing about the recent proposal they made, that black holes can be described only by a handful of parameters: their mass, their angular momentum, and their electric charge. Nothing else was needed to completely characterise a black hole state.

But what happens, the venerable professor argued, if I throw this cup of tea into a black hole? The tea carries an entropy, which contributes to the entropy of the Universe. If it falls into a black hole the entropy of the cup of tea would suddenly disappear, and the total amount of entropy in the Universe would decrease, violating the second law of thermodynamics. Was it possible, the professor asked, that black holes transcend such an important law in such a simple manner?

If the student were a bit lazy, maybe busy with his next exam, he could have made up excuses saying that one cannot look into a black hole, but that the cup of tea would still be there, with its entropy forever contributing to the entropy balance of the Universe, even without being measured ever again. But the venerable professor was John A. Wheeler, and the student Jacob Bekenstein.

Bekenstein started thinking about the problem, and in 1972 he published a seminal paper \cite{Bekenstein72}, in which he proposed that a black hole must indeed carry an entropy. On page 2, after exposing the puzzle, Bekenstein wrote: "We state the second law as follows: Common entropy plus black-hole entropy never decreases". He then went on explaining that common entropy is the entropy carried by matter outside the black hole, and that the black hole entropy should have been proportional to the event horizon area, with a proportionality coefficient of order unity, writing a formula very similar (indeed, the same, apart from an arbitrary coefficient) to our first equation:
\begin{equation*}
S_{BH} = \eta k_B \frac{A}{l_P^2} \ .
\end{equation*}

The choice of the area of the black hole as its entropy was motivated by a result by a student in Cambridge, Stephen Hawking, who discovered that, just as the entropy of a system, the area of black holes never decreases \cite{HawkingEllis73}. In a subsequent paper, \cite{Bekenstein73} Bekenstein went further, arguing that the proportionality coefficient should be $\eta = \frac{1}{2} \ln 2 \simeq 0.347 $, and stating that the black hole entropy should be understood as a statistical measure of the black hole microstates, and that "it would be somewhat pretentious to attempt to calculate the precise value of the constant $\eta \hbar^{-1}$ without a full understanding of the quantum reality which underlies a \textit{classical} black hole".

Bekenstein's proposal was outrageous. Hawking's area theorem was a statement in differential geometry, while the entropy was a statistical measure, emerging from the counting of the microstates of a system. How could they be related? Moreover, a finite entropy implies a finite temperature, because of the first law of thermodynamics, and that would mean that black holes could emit a thermal radiation in vacuum, in striking contrast with their fundamental property in General Relativity, the very fact that they are \textit{black}, and do not emit anything. To settle the question, Hawking himself computed the effects that the formation of a black hole would cause on the modes of a free quantum field, to definitely confute Bekenstein. The result was published in \textit{Nature} \cite{Hawking74BHExplosions}: Hawking discovered that black holes emit a black-body radiation at temperature
\begin{equation*}
T = \frac{\hbar c^3}{8 \pi G k_B} \frac{1}{M} \ .
\end{equation*} 
Hawking computation in turn fixed the proportionality coefficient between the area and the entropy of a black hole to $1/4$, showing that not only was Bekenstein right, but also that he came incredibly close to the correct value for the proportionality coefficient.

Hawking's paper put black hole thermodynamics on firm grounds. The next, natural step was to answer the question Bekenstein posed in his first paper, on the statistical interpretation of the black hole entropy. What are the microstates that cause the black hole to have a finite entropy? The first attempt in this sense was made by Bombelli et al. in 1986 \cite{Bombelli86}. They considered a quantum field propagating on a black hole background, and noted that, although the black hole interior was causally disconnected from the outside, still the degrees of freedom inside and outside would be entangled. They considered, then, an initially pure quantum state, described by a state vector $\ket \Phi$ in some Hilbert space $\mathcal H$, and its associated density matrix $\rho_\Phi = \ket \Phi \bra \Phi$. Then, they wanted to isolate the degrees of freedom inside the black hole: thus, they traced over the degrees of freedom outside the black hole, which we will denote with $\mathcal B'$, that is, the complement of the black hole, obtaining the reduced density matrix $\rho_\Phi^{\mathcal B} = \Tr_{\mathcal B'} \rho_\Phi$. They proceed computing the \textit{entanglement entropy}, which they defined as $S_{\mathcal B} = - \Tr \rho_\Phi^{\mathcal B} \log \rho_{\Phi}^{\mathcal B}$. They found that the entanglement entropy is indeed proportional to the area of the horizon, thus providing a natural explanation for the origin of the black hole entropy.

The entanglement entropy became an important measure of entanglement in statistical physics, but its application to black holes soon showed some problems. First, the entanglement entropy is a divergent quantity in the continuum: if $\epsilon$ is the distance between the two entangled regions, the entanglement entropy is
\begin{equation*}
S_{\mathcal B} \sim \frac{A}{\epsilon^{d-2}}
\end{equation*}
where $d$ is the dimension of the space, and we switched to natural units, $G = c = \hbar = k_B = 1$, which we adopt from now on. Moreover, the proportionality coefficient depends on the matter model and on the number of fields present in the theory. This is to be compared to the Bekenstein-Hawking formula, which is a UV-complete (that is, finite), universal quantity, independent on the model considered and reproduced in a variety of approaches and situations. For these reasons, approaches to black hole thermodynamics based on entanglement entropy lost interest (although the research continued, see \cite{Solodukhin11} for a review), as it was assumed, as Bekenstein originally did, that in order to explain the black hole entropy one needs a theory of quantum gravity.

Alternatively to the entanglement entropy, however, one can introduce the \textit{relative entropy}. In non-relativistic Quantum Mechanics, it is defined starting from two states, $\ket \Phi$ and $\ket \Omega$ and their associated reduced density matrices in a subsystem $\mathcal R$, $\rho^{\mathcal R}_\Phi$ and $\rho^{\mathcal R}_\Omega$, as $S(\Phi | \Omega ) = \Tr \rho^{\mathcal R}_\Phi ( \log \rho^{\mathcal R}_\Phi - \log \rho^{\mathcal R}_\Omega)$. 

Although it is constructed in a similar way to the entanglement entropy, the relative entropy is a slightly different measure. In fact, it does not compute the entropy between two regions, but rather, it is a measure of entanglement between two \textit{states}: therefore, it cannot be associate to some degrees of freedom localised in a region, but it is a property of the theory itself. 

Relative entropy admits a formulation for continuum theories in the context of Algebraic Quantum Field Theory, which generalises the formula valid in Quantum Mechanics, and solve the divergence problems of entanglement entropy. The generalisation to continuum theories were first found by Araki \cite{Araki76}, in 1976. The starting consideration is that entanglement is not a property of states, but rather is a property of the algebra of observables, and therefore it is somewhat an inevitable feature of any quantum theory. Araki was able to write the relative entropy as the expectation value of an operator, the \textit{relative modular Hamiltonian} $K_{\Omega, \Phi}$, constructed using the theory of algebra automorphisms (the operators which map the algebra into itself, that are, its symmetries) developed by Tomita and Takesaki \cite{Takesaki70}. The \textit{Araki formula} (given in \eqref{eq:araki-formula}) states that the relative entropy between two states $\ket \Omega$, $\ket \Phi$ is given by the expectation value of the relative modular Hamiltonian:
\begin{equation*}
S(\Omega | \Phi) = - \ev{K_{\Omega, \Phi}}{\Omega} \ .
\end{equation*}
A detailed discussion of the formula is presented in section \ref{sec:E-EE}, following the presentation of \cite{Haag92}, \cite{CasiniGrilloPontiello19}, \cite{Witten18}.

In ordinary Quantum Mechanics, the relative modular Hamiltonian $K_{\Omega, \Phi}$ reduces to a tensor product of the logarithms of the density matrices associated to the two states $\ket \Phi$ and $\ket \Omega$,
\begin{equation*}
K_{\Omega, \Phi} = - \log(\rho_\Omega \otimes \rho_\Phi^{-1}) \ ,
\end{equation*}
while the expectation value of an observable in a state $\ket \Phi$ is given by
\begin{equation*}
\ev{A}{\Phi} = \Tr \rho_\Phi A \ .
\end{equation*}
Thus, the Araki formula reduces to the formula $S(\Phi | \Omega ) = \Tr \rho^{\mathcal R}_\Phi ( \log \rho^{\mathcal R}_\Phi - \log^{\mathcal R}_\Omega)$ for ordinary Quantum Mechanics. However, since it makes no use of the properties of the Quantum Mechanics itself, can be formulated for any type of algebra, and in particular it still holds for the algebra of observables emerging in Quantum Field Theories.

In 2019, a series of papers \cite{CasiniGrilloPontiello19}, \cite{CiolliLongoRuzzi19}, \cite{Longo19} used the Araki formula to compute the relative entropy between coherent states of a Klein-Gordon field. Coherent states are particularly simple to handle, because they can be considered a classical perturbation of a quantum state (see subsection \ref{ssec:coherent-perturbations}). In this application, the Araki formula reduces to an integral of a component of the stress-energy tensor associated to a classical wave, and therefore the abstract setting of the modular theory boils down to a very explicit expression in terms of the wave and its derivatives. Thanks to this formulation, it has been possible to compute the relative entropy between coherent states in the exterior of a Schwarzschild black hole. In their work \cite{HollandsIshibashi19}, Hollands and Ishibashi computed the relative entropy between coherent states, both for scalar and gravitational perturbations of the vacuum, over a Schwarzschild background, reproducing the Bekenstein-Hawking formula for black hole entropy.

The strength of this approach stems from the fact that it directly computes the entropy content of the matter, and it shows that its variation causes an analogue variation in the horizon area. Previous approaches to black hole entropy were all based on a geometrical point of view. They started with an analogy between the first law of black hole mechanics,
\begin{equation*}
\delta M = \frac{\kappa}{8 \pi} \delta A + ... \ ,
\end{equation*}
which relates the variation in black hole mass during a quasi-stationary process to the variation of the horizon area, plus terms related to the angular momentum and electric charge of the black hole, and the first law of thermodynamics,
\begin{equation*}
\dd E = T \dd S + p \dd V \ .
\end{equation*}
Since Hawking's computation fixes the black hole temperature at $T = \kappa/ 2\pi$, where $\kappa$ is the surface gravity of the black hole, then the entropy is fixed to one-fourth of the area of the horizon.

An alternative approach worth mentioning was found by Wald \cite{Wald93}, who identified the black hole entropy with the Noether charge of a theory with diffeomorphism invariance. Again, however, this approach is based on geometric considerations only. 

In this thesis we generalise Hollands and Ishibashi's approach to \textit{apparent horizons} and cosmological horizons. The idea is to compute the variation of relative entropy for matter fields propagating on a black hole background. Since the relative entropy is defined from the states of the theory, it can be associated with the region in which the theory is defined. Varying the region of definition then causes a variation in the relative entropy. Then, we show that a variation in relative entropy is accompanied by a variation of one-fourth of the area of the horizon. This give a direct interpretation of the horizon area as the entropy carried by the black hole.

The approach is simply to take the original Bekenstein's gedanken experiments and apply them to the case of quantum matter: what happens if we take a quantum cup of tea and we throw it into a black hole? To model a quantum cup of tea we consider a free, massless scalar field, propagating over a black hole background from a set of initial data given at past infinity. We consider then the relative entropy between a vacuum state and a coherent state of the scalar field. Since the relative entropy is associated to a scalar field, it naturally depends on the region over which the field propagates: if we vary the surface on which we give initial data, the field can propagate in a different region, and thus the relative entropy varies accordingly. Then, we consider a region in the outside of a black hole, extending from past infinity and the past horizon to a hypersurface at $v = v_0$, where $v$ is the advanced time coordinate. The scalar field is then treated as a perturbation of the metric, and thus causes a variation of the area of the horizon. If we now vary the region over which the field can propagate, we cause a variation in relative entropy and a simultaneous variation in the horizon area, because clearly the area of the horizon can be perturbed only in the region where the field propagates. We then use a conservation law for a current constructed with the field's stress-energy tensor to link the variation of relative entropy to a variation of the area of the horizon, and we obtain that
\begin{equation*}
\dv{v_0}(S_{v_0}(\Omega | \Phi) + \frac{1}{4} \delta^2 A(v_0) ) = 2 \pi \mathcal F \ .
\end{equation*}
The right-hand side is a surface integral of the field current, while in the left-hand side $S_{v_0}(\Omega | \Phi)$ is the relative entropy between two coherent states computed in the region we considered, extended up to $v = v_0$ , and $\delta^2 A(v_0)$ is the second-order perturbation in the horizon area caused by the field at advanced time $v = v_0$ (that is, the 2-dimensional, spherical cross-section of the 3-dimensional event horizon). The derivative is with respect to the boundary of the region considered. This formula allow us to identify one-fourth the area of the horizon as the entropy contribution of the black hole, since it varies in the same way as the relative entropy of the field.

We apply such a computation in different backgrounds. First, we review the argument by Hollands and Ishibashi \cite{HollandsIshibashi19} for the case of a Schwarzschild black hole, with a couple of simple modifications. First, we choose to give the initial data for the field's wave equation at past infinity on the past horizon, rather than at future infinity and on the event horizon, as is done in \cite{HollandsIshibashi19}. We do so in order to refer the computation of the relative entropy to a coherent perturbation of the Unruh state, which is the most natural vacuum state in the presence of a black hole (see subsections \ref{ssec:hawking-temperature} and section \ref{sec:RE-Unruh}, and \cite{DMP11} for a more detailed discussion). Moreover, from a physical point of view, it is more reasonable to give initial data at past infinity, and then let the field propagate over the background, instead of giving initial data at future infinity, though the two approaches are mathematically equivalent. Second, we want to refer the variation of the horizon area to an arbitrary interval, not necessarily extended to infinity. The reason for this requirement is that, in treating dynamical black holes, we consider a distribution of matter falling into the black hole for a finite amount of time, which again is a reasonable assumption from a physical point of view. Although it is true that astrophysical black holes always absorb the Cosmic Microwave Background (CMB) radiation, which can be considered as a distribution of matter permeating the whole Universe, one is rather interested in the accretion of a black hole during the infalling of some distribution of matter. Therefore, we try to model such a situation, although in the highly idealised setting of a Vaidya black hole, considering a distribution of matter falling into the black hole in a finite interval of time, whith the black hole remaining stationary outside this interval. For this reason, also in the stationary case we consider two regions in the black hole exterior, extending from past infinity and the past horizon to a hypersurface, respectively at $v = v_0$ and $v = v_1$, and we compute the difference between the relative entropies associated to the theories defined in these two regions. Then, we compute the variation of this difference as we rigidly translate the two regions by the same amount. Since the Schwarzschild spacetime is static, it admits a globally time-like Killing field which defines a conserved current. We use the Killing conservation law to link the variation of relative entropy to the variation of the stress-energy tensor on the black hole horizon. Finally, the Raychaudhuri equation \eqref{eq:Raychaudhuri} let us link the stress-energy tensor to a variation in the horizon area. In the way we use it, the Raychaudhuri equation is a re-expression of Einstein equations, which relates the variation in the cross-sectional area of a congruence of geodesics (which can be visualised as a bundle of wires) to the projection of the stress-energy tensor along the geodesics.

Then, we apply the procedure to the case of apparent horizons. Apparent horizons are an alternative characterisation of black holes which is most useful for black holes evolving in time, due to the presence of matter. Here, we consider a simple model for dynamical black holes, the Vaidya spacetime \cite{Booth10}. The Vaidya black hole is the most simple generalisation of a Schwarzschild black hole to the dynamical case, since it replaces the Schwarzschild parameter $m$, the black hole mass, with a function of the advanced time, $m(v)$. It describes a black hole in accretion due to the infalling of shells of radiation, or an evaporating black hole emitting radiation. In particular, it preserves spherical symmetry, which is a key technical point. In fact, in spherically symmetric spacetimes it is possible to introduce a vector, called Kodama vector \cite{Kodama79}, which give rise to a conserved current, just as a Killing field. In the case of the Kodama vector, however, the conservation law is not a consequence of a symmetry of the metric, but rather it descends from the geometric properties of the vector itself (in particular, the fact that it is divergence-free).

We again perturb the spacetime with a free, massless scalar field, and we proceed in the same way as in the Schwarzschild case computing the variation of the difference of relative entropy associated to two regions, and the associated variation in the horizon area. In this context, since the Vaidya black hole is spherically symmetric, we can use the Kodama vector to find a conserved current, and we again apply the Raychaudhuri equation to link the conserved current to a variation of the horizon area. We find that we can again introduce a notion of black hole entropy associated to one-fourth the area of the apparent horizon. In the dynamical case, we also find a term which can be interpreted as a work term done over the black hole by the field.

Finally, we see that the same ideas can be applied in the case of cosmological horizons. For simplicity we consider only spaces which asymptotically reduce to flat FLRW spacetimes, and which exhibit an event horizon. We show that the same procedure can be applied in this context, and thus we can conclude that the relative entropy is a good candidate for the computation of entropy contributions in at least three different classes of background models.

The work is organised as follows. In chapter \ref{ch:GR} we review the main tools of General Relativity and of differential geometry we need to talk about black hole and the horizons' properties. We review the notion of globally hyperbolic spacetimes, the Gauss-Stokes theorem, geodesics congruences and the Raychaudhuri equation, and we show how to construct the Penrose diagram for asymptotically flat spacetimes. Penrose diagrams \cite{Penrose64} are one of the best ways to visualise the causal structure of a spacetime with spherical symmetry, and we will use it to explain the geometric settings we consider for the propagation of a scalar field over a curved background. In chapter \ref{ch:QFTCS} we review the construction of a quantum field theory over a curved spacetime. Here, we adopt the algebraic approach, which is mathematically rigorous and physically sound, as we explain in further detail in section \ref{sec:motivations}. We show how one can construct the observables for a free field propagating on a globally hyperbolic spacetime, and how one can define a state to give the expectation value of observables without referring to a preferred notion of a vacuum or to background symmetries. We then illustrate the Kubo-Martin-Schwinger (KMS) conditions \cite{Kubo57} \cite{Haag67}, which characterise thermal states at finite temperature, and its connections with the Tomita-Takesaki theory of modular automorphisms. The Tomita-Takesaki theory is in turn the context in which we define the relative entropy, via the Araki formula. We then show how to explicitly compute the relative entropy from the Araki formula for the simple case of coherent states. In chapter \ref{ch:new-results} we review the classical laws of black hole mechanics and we finally introduce the new computation on the relative entropy for a Schwarzschild black hole, and we show our new results on Vaidya black holes and cosmological spacetimes. We show that General Relativity predicts four dynamical laws for the black holes, which bear a striking resemblance with the four laws of thermodynamics, as was shown in \cite{Bardeen73}. We recall that the Unruh state, when restricted to the exterior of black hole, show a thermal flux of particles directed toward future infinity, and we discuss its relation with the Hawking effect, namely, the fact that black holes radiate with a black body spectrum. We then compute the relative entropy in the case of Schwarzschild, Vaidya, and cosmological spacetimes. Our results are expressed in the three equations \eqref{eq:result-RE-Schw}, \eqref{eq:result-RE-Vaidya}, and \eqref{eq:result-cosmological-horizon-RE}.

Throughout this thesis, we adopt the metric signature $(-,+,+,+)$. We will use natural units, in which $ \hbar = c = k_B = G = 1 $.

\chapter{General Relativity} \label{ch:GR}

    \section{The Most Beautiful of All Theories} \label{sec:most-beautiful-theory}
    
    This is how Landau and Lifschitz called General Relativity (GR), in their volume on Classical Fields \cite{Landau82}. The same sense of beauty inspires us while listening to Beethoven's Quartets, or reading the Hamlet, or staring in awe at the ceiling of the Sistine Chapel: a wordless, almost magic feeling that the creation of a work of art was a logical necessity for the world to be complete, inevitable as the stars in the night sky. Mathematical beauty is no different from artistic or poetic beauty; what rings deeply of GR is its deep unity, the same of Notre-Dame cathedral in Paris or of Leonardo's \textit{Annunciazione}, where few, fundamental principles are able to produce entire new worlds.
       
   The principles of GR are based on what Einstein called "the happiest thought of [his] life": while he was still in a patent office in Bern, in 1907, he realised that a man in free fall would not feel his own weight. This is what became known as the Equivalence Principle: gravitational effects can, locally, be identified with those of non-inertial observers. This identification generalises Galileo's principle of relativity, in the sense that, if one includes gravity, non-inertial frames without gravity are equivalent to inertial frames in a gravitational field. In Einstein's hands, the Equivalence Principle became the key fact that the laws of Physics must be the same for \textit{all} observers, not just inertial ones; in other words, physical laws must be written in a coordinate-independent way, a property that we now call the covariance of a theory.
   
   These physical insights can be expressed in a mathematical consistent way. If all observers, at least locally, are equivalent, each of them believes to be inertial (no matter if they truly are; actually, there is nothing like a "truly" inertial observer, no more than there is something like a "truly" observer in constant motion) and therefore experience the physical laws of inertial observers: they live in Minkowski spacetime. Note there is no requirement that the spacetime is \textit{globally} Minkowski; the point is that it must always be possible to construct a local inertial frame of reference. On the other hand, physical quantities must be written in a coordinate-independent way, because different observers may experience different effects (in other words, the numerical value of physical quantities of course depend on the coordinates) but must agree on the physical law: the theory must be generally covariant. This is precisely the basic property of manifolds studied in differential geometry; the caveat is that one (an expert mathematician at the beginning of the last century) was used to require a local Euclidean spacetime, and therefore the spacetime would have had the structure of a Riemannian manifold; in GR one requires that the spacetime locally reduces to Minkowski space, and so the first, mathematical realisation of the Equivalence Principle is that \textit{the spacetime structure is described by a Lorentzian manifold}. Differential geometry is the right language for gravity because it expresses its equations using tensors, which are objects that have well-defined transformation laws under change of coordinates and admits a description in a coordinate-independent way.
   
    The second consequence is that everything must be locally determined, and, moreover, it must be dynamical; if something were non-dynamical, it would define a preferred frame of reference. In particular, even the background, that is, spacetime itself must evolve according to some dynamical equation. Therefore, both matter and geometry must satisfy partial differential equations with a well-posed initial value problem, so that they are uniquely determined from their initial data within a suitable domain of dependence.
   
   Implementing these two mathematical requirements in a physical theory is not an easy task. It took Einstein eight years, from 1907 to 1915\footnote{and, actually, it took \textit{everyone} eight years: many, most notably the mathematician David Hilbert, looked for the unification of Newton's gravity with Special Relativity; Einstein had been the fastest.} to find that gravity is an effect of geometry itself, the expression of the dynamical deformation of spacetime in the presence of matter. Einstein finally published the field equations of gravity in November, 1915 \cite{Einstein15}. They are so beautiful they can be written on a t-shirt:
   
   \begin{equation} \label{Einstein-equations}
   R_{\mu \nu} - \frac{1}{2}R g_{\mu \nu} = 8 \pi T_{\mu \nu} \ .
   \end{equation}
$R_{\mu \nu}$ is the Ricci tensor, $R$ is the Ricci scalar or Ricci curvature, $g_{\mu \nu} $ is the metric which encodes the gravitational field, and $T_{\mu \nu}$ is the stress-energy tensor. They are actually incomplete: we do not include the cosmological constant term, $\Lambda g_{\mu \nu}$, on the left hand side. This is because we will mainly deal with black holes, whose core properties are captured by a spacetime without cosmological constant. Although black hole solutions with a nonvanishing cosmological constant are known and well-studied, for simplicity we will work with $\Lambda = 0$ only.
    
    In this short introduction we will only review the main tools to deal with black holes and cosmological spacetimes. Our treatment is by no means self-contained nor comprehensive of the many profound directions in which the mathematical theory of GR developed; we will mainly follow \cite{AAQFT15} and \cite{Wald84} for the treatment of globally hyperbolic spacetimes and \cite{Poisson09} for the discussion on geodesic congruences and on hypersurfaces. We introduce the global structure of a spacetime in which one can pose a well-defined Cauchy problem, that is, in which it is possible to make physical predictions on the evolution of a system. This lead to the notion of globally hyperbolic spacetimes and of Green hyperbolic operators. We discuss the Raychaudhuri equation, that determines the evolution of a family of geodesics. We briefly explain how one can visualise the global structure of a globally hyperbolic spacetime using Penrose diagrams, a way to map an infinite spacetime in a finite diagram. We show how to apply these tools discussing the properties of event horizons and apparent horizons.
   
    \section{Globally Hyperbolic Spacetimes} \label{sec:globally-hyperbolic-spaces}

    In this section we discuss the properties of spacetimes in which it is possible to set a well-posed initial value problem. This condition is an \textit{a priori} requirement on any physically sensible spacetime, because it permits to make physical predictions on the evolution of some initial values through partial differential equations. The geometric notions are therefore independent on the particular law of gravitation one chooses, and we will not make use of Einstein's equations in this section. However, we will impose the \textit{principles} of GR, general covariance and the Equivalence Principle, as they are necessary requirements for any generally covariant theory of gravity. Moreover, throughout the thesis we will work assuming that the connection on the spacetime is metric, that is, the covariant derivative of the metric vanishes: $\nabla_{\mu} g^{\mu \nu} = 0$. This is always assumed in the context of GR, so there is no loss of generality.
    
    The starting point is the definition of a manifold. As we said, spacetime in GR should be a topological space which locally resembles the Minkowski spacetime. This is the definition of a \textit{Lorentzian manifold}. The manifold must be equipped with a Lorentzian metric, with signature $(-,+,+,+)$. The technical definition assumes some more additional properties: it must be Hausdorff, second countable, connected, time orientable, and smooth. We will call such a manifold a spacetime, denoted $(\mathscr M, g)$.
    
    We assume a connected spacetime for simplicity; on one hand, we will deal with connected spacetimes only, and, on the other, if the spacetime is not connected one can always consider our definitions in each of its connected parts. The spacetime is time orientable if it admits a global vector which is everywhere time-like; manifolds as the Klein bottle are not admitted. A time orientable spacetime is necessary to avoid discontinuity in the "flow of time", that is, such pathologies as the inversion of the direction of time in some adjacent regions.
    
    The Lorentzian character of the metric $g$ implies that every spacetime comes with a causal structure, determined by the union of the light cones defined at each point. Given a point $p \in \mathscr M$, we denote its \textit{causal} past (future) $J^-(p) \ (J^+(p))$ the set of all points separated by a time-like or a light-like interval from $p$, and which lie in the past (future) of $p$. We further call the \textit{chronological} past (future) $I^-(p)$ ($I^+(p)$) the set of all points separated by a time-like interval from $p$, and which lie in the past (future) of $p$. The fact that the spacetime is time orientable means that, taken two close points, one can smoothly cross their causal futures. Therefore it is possible to define the causal future (past) of a region $\Sigma \subset \mathscr M$ as the union of the causal future (past) of its points. A smooth curve $\gamma : I \subset \mathbb R \to \mathscr M $ is said to be time-like, space-like, or light-like if its tangent vector is time-like, space-like, or light-like; a causal curve is a curve either time- or light-like.
    
    The causal structure allows us to find the conditions that a spacetime must satisfy to avoid breakdowns in predictability, that are, situations in which a suitable set of initial data fails to predict the evolution of some physical quantity (it is clear that 'suitable' is a key requirement here: nobody expects to predict the trajectory of a body from the number of angels dancing on the head of a pin, but we do expect that given, say, the initial position and velocity of a point particle, one can predict its trajectory with deterministic certainty, if quantum effects are negligible. We will specify what we mean with "suitable" in theorem \ref{th:Cauchy-problem}). It turns out that the worst thing that can happen in GR are closed time-like curves, in which one always moves forward in time until she hits her own past, causing all sorts of time-travel paradoxes. We will therefore add one more property to find the physically reasonable spacetimes. We will call a set $\Sigma \subset \mathscr M$ \textit{achronal} if each time-like curve in $\mathscr M$ interstects $\Sigma$ at most once, and we will define the future ($+$) or past ($-$) \textit{domain of dependence} $D^\pm(\Sigma)$ the set of all points $p \in \mathscr M$ such that an inextensible causal curve through $p$ intersects $\Sigma$ in the past (future). Then $\Sigma$ is a \textit{Cauchy hypersurface} if it is closed an achronal. Finally, a \textit{globally hyperbolic spacetime} is a spacetime which admits a Cauchy hypersurface, that it, there exists a closed achronal subset $\Sigma$ such that $D^+(\Sigma) \cup D^-(\Sigma) = \mathscr M$. The set  $D^+(\Sigma) \cup D^-(\Sigma) $ is called \textit{causal development} of $\Sigma$.
    
    Such a definition lets one realise what we required on physical grounds in the opening paragraph: in fact, the Cauchy surface $\Sigma$ is the surface on which one gives the initial data of some physical quantity; the dynamical evolution through partial differential equations then lets one find the value of such a quantity in the causal development of the Cauchy surface. Note that the definition also rules out closed time-like curves. Einstein would be relieved; in fact, for the 70th birthday of Einstein, G\"odel gave him as a gift a solution of GR which admits precisely closed causal curves as a kind of time-travel machines. It was the first time he doubted of the correctness of his theory, Einstein said \cite{Holt05}.
    
    Although global hyperbolicity is usually assumed as a necessary requirement, it is not known if it can hold on any physically reasonable spacetime, like spacetimes sourced by some matter content which is realisable in the Universe. In fact, it is known that global hyperbolicity would fail in the presence of the so-called naked singularities, singularities which are not hidden by an event horizon, and naked singularities are realisable in ordinary spacetimes such as the Kerr family of metrics, describing a rotating black hole. To avoid naked singularities one must in turn assume the \textit{cosmic censorship hypothesis}, but there are examples in which at least some versions of the hypothesis fail. That being said, we will deal only with globally hyperbolic spacetimes, keeping Physics safe for now from the breakdown of predictability.
    
    Globally hyperbolic spacetimes can be characterised in a more practical way, proving the following theorem (see \cite{Bar07}):
    
    \begin{theorem}
    The following facts are equivalent:
    \begin{enumerate}
        \item $(\mathscr M, g) $ is globally hyperbolic;
        \item There exists no closed causal curve and $J^+(p) \cap J^-(q) $ is either compact or empty $\forall p,q \in \mathscr M $
        \item The metric can be written in the form $\dd s^2 = - \beta \dd t^2 + h(t)$, where $\beta$ is some smooth function of $t$ and $h(t)$ is a family of 3D Riemannian metrics.
    \end{enumerate}
    \end{theorem}
    The last property in particular is the most practical tool to determine if the spacetime is globally hyperbolic. It is immediate in fact to notice that Schwarzschild, Vaidya and FLRW spacetimes are all globally hyperbolic, the spacetimes we will deal with in chapter \ref{ch:new-results}.
    
    Globally hyperbolic spacetime are the ideal stage for the dynamics of physical systems. However, to solve a well-posed initial value problem we need not only to select a suitable Cauchy surface on which to assign initial data, but also the class of differential operators useful to discuss the equations that arise in field theory. This class of operators are known as the \textit{normally hyperbolic operators}. In order to completely understand this notion we would need to introduce the structure of vector bundles and its properties, but there exists a practical criterion to identify such operators, and we will stick to it here to avoid any digression outside the scope of this section. Moreover, we will consider the definition only for scalar functions, since we will deal with the scalar field only.
    
    We take the space of compactly supported, smooth functions, denoted $C^\infty_0(\mathscr M) \ni \psi : \mathscr R \subset \mathscr M \to \mathbb C $. A normally hyperbolic operator $P$ is of the form
    
    \begin{equation}
    P\psi = g^{\mu \nu} \partial_\mu \partial_\nu \psi + A^\mu \partial_\mu \psi + A \psi
    \end{equation}
for some set of smooth functions $A, \ A^\mu, \ \mu = 0,...,d$ where $d$ is the dimension of the spacetime. This is just a generalisation of the d'Alembert operator $\square = g^{\mu \nu}\nabla_\mu \nabla_\nu$. Then, the following theorem holds:

\begin{theorem} \label{th:Cauchy-problem}
Let $(\mathscr M, g)$ be a globally hyperbolic spacetime with some Cauchy surface $\Sigma$ with unit, future-directed, normal vector $n$. Let $P$ be a normally hyperbolic operator. Let $j \in C^\infty(\mathscr M)$, $\psi \in C^\infty(\mathscr M), \ f_0, \ f_1 \in C^\infty_0(\Sigma) $. Then, the following initial value problem admits a unique solution in $\mathscr M$:
\begin{align}
P\psi = j \\
\eval{\psi}_{\Sigma} = f_0 \\
\eval{n^\mu \nabla_\mu \psi}_{\Sigma} = f_1
\end{align}
\end{theorem}
We remark that all the definitions can be restricted to a subregion $\mathscr R \subseteq \mathscr M$, using a \textit{partial Cauchy hypersurface} $\Sigma$ such that its domain of dependence covers only a subregion $\mathscr R = D^+(\Sigma) \cup D^-(\Sigma)$.

 \section{Gauss-Stokes Theorem} \label{sec:gauss-stokes-theorem}

    \subsection{Differential Forms} \label{ssec:differential-forms}

The main goal of this section is to discuss the generalisation of Gauss theorem to the general spacetimes required in GR. This theorem will be one of the main tools in the discussion of relative entropy for black holes in chapter \ref{ch:new-results}, when we will use it to rewrite the contribution to relative entropy at past infinity as a surface integral on the black hole horizon.
 
In order to discuss the Gauss-Stokes Theorem, we need first to understand the integration over general manifolds. To do so, a brief digression on differential forms is in order.

A differential p-form is a totally antisymmetric covariant tensor, that is, $\omega$ is a p-form if

\begin{equation}\label{def:p-forms}
\omega_{a_1...a_p} = \omega_{[a_1...a_p]} \ .
\end{equation}
The squared parentheses denotes the total antisymmetrization of the enclosed indices.

In a manifold, at each point $x$ is associated the vector space of all p-forms, $\Lambda^p_x$. To combine a p-form with a q-form, we introduce an antisymmetric product, $\wedge : \Lambda^p_x \cross \Lambda^q_x \to \Lambda^{p+q}_x $ defined as

\begin{equation}
\omega_{a_1...a_p} \wedge \mu_{b_1...b_q} = (\omega \wedge \mu)_{a_1...a_p b_1...b_q} = \frac{(p+q)!}{p!q!} \omega_{[a_1...a_p}\mu_{b_1...b_q]} \ .
\end{equation}
We denote the space of all forms at each point with $\Lambda_x = \bigoplus_{p=0}^\infty \Lambda^p_x$, and the space of differential p-form fields over the manifold with $\Lambda^p$. This space has dimension $\dim \Lambda_p = \frac{d!}{p!(d-p)!}$.

p-forms are a generalisation of linear forms, called, in this broader context, 1-forms. Linear forms are nothing but covariant vectors; they are in duality with the contravariant vectors (or vectors \textit{tout-court}), because the metric defines a natural pairing between the two. In particular, given a set of coordinates $(x^1,..., x^d)$, the vector basis is given by $\{ \partial_\mu \}$, and we can find the dual basis in the 1-form space, denoted $\{ \dd x^\mu \} $, imposing the condition

\begin{equation} \label{eq:duality-vectors-forms}
g(\partial_\mu, \dd x^\nu) = \delta^\nu_\mu  \ .
\end{equation}
 The basis $\{ \dd x^\mu \} $ is called \textit{natural basis}. Therefore, any 1-form can be written as $\omega = \omega_a \dd x^a$. Now, a general p-form can be expressed as
 
\begin{equation}
\omega = \omega_{[a_1...a_p]} \dd x^{a_1} ... \dd x^{a_p} = \frac{1}{p!} \omega_{a_1...a_p} \dd x^{a_1} \wedge ... \wedge \dd x^{a_p} \ .
\end{equation}

The contraction of vectors and p-forms can be written in two, equivalent ways: if one want to emphasize that any vector $\xi$ has a natural action on p-forms, one writes \eqref{eq:duality-vectors-forms} as the map $i_\xi : \Lambda^p \ni \omega \mapsto i_\xi(\omega)_{a_2...a_p} = \xi^{a_1}\omega_{a_1...a_p} $, or, vice versa, if one sees it as the action of a p-form over a vector, one writes $\omega(\xi)$.

Among the p-forms, d-forms (forms of the same dimension of the manifold) are special because they are 1-dimensional, and therefore every d-form is proportional to a given one, which is the basis of the space $\Lambda^d$. In d dimensions, one already has a totally antisymmetric matrix, the Levi-Civita symbol $\epsilon^{a_1...a_d} $. However, the Levi-Civita symbol is not a tensor, since it does not transform covariantly for coordinate transformations. If the manifold is equipped with a metric (which is always our case), one can use the Levi-Civita symbol to defines a canonical basis for d-forms, called the \textit{volume form}:

\begin{equation}
\Omega = \frac{1}{d!} \sqrt{-g} \epsilon_{a_1...a_d} \ \dd x^{a_1} \wedge ... \wedge \dd x^{a_d} = \frac{1}{d!} \sqrt{-g} \ \dd x^{a_1} \wedge ... \wedge \dd x^{a_d} \ .
\end{equation}
It can be proven (see \cite{Wald84}) that this d-form arises simply imposing $\Omega^{a_1...a_d} \Omega_{a_1...a_d} = - d!$.

Before discussing the integration, we introduce a last operation on p-forms, the \textit{exterior differential} $\dd$. This is a map $\dd : \Lambda^p \to \Lambda^{p+1}, \ \omega \mapsto \dd \omega $, which satisfies three properties:
\begin{enumerate}
\item The action on a $0$-form (i.e., a function) $f$ is defined as $\dd f = \partial_\mu f \dd x^\mu $;
\item $\dd^2 \omega = 0$, the \textit{Poincaré Lemma};
\item $\dd (\omega \wedge \mu) = \dd \omega \wedge \mu + (-1)^p \omega \wedge \dd \mu $, where $p$ is the rank of $\omega$ (that is, the number of indices of $\omega$).
\end{enumerate}
From these properties, one can see that the action of the exterior differential on a general p-form $\omega$ is given by

\begin{equation}
\dd \omega = \frac{1}{p!} \nabla_\mu \omega_{a_1...a_p} \dd x^\mu \dd x^{a_1} \wedge ... \wedge \dd x^{a_p} \ .
\end{equation}
The components are $(\dd \omega)_{\mu a_1... a_p} = (p+1) \nabla_{[\mu} \omega_{a_1...a_p]} $. The second defining property is called Poincaré lemma because one could proceed in the opposite way, defining the action of the exterior differential on the general p-form and then prove the above properties.

Finally, we are able to integrate d-forms in a d-dimensional manifold. Suppose we restrict to a subregion $\mathscr R$ in which a chart $\psi : \mathscr R \ni x \mapsto (x^1,..., x^d) $ is defined everywhere. The integration of a d-form is simply

\begin{equation} \label{def:integration-forms}
\int_{\mathscr R} \omega = \int_{\psi(\mathscr R)} \omega_{a_1...a_d} \dd x^{a_1} ... \dd x^{a_d} \ .
\end{equation}
The right-hand side is interpreted as the usual Lebsgue integral in $\mathbb R^d$.

The integral over $\mathscr M$ is then accomplished by integrating over each chart of the atlas, taking care of the overlapping regions so to not integrate twice in the same region. 

The integration of a function is defined associating to the function a suitable d-form. The obvious candidate is the volume form; the integration of a function is then defined as
\begin{equation} \label{def:integration-functions}
\int_{\mathscr R} f \Omega = \int_{\psi(\mathscr R)} f \sqrt{-g} \dd x^1 ... \dd x^n \ ,
\end{equation}
from which we see that we can define the volume element of the manifold as $\dd \text{vol}_{\mathscr M} = \sqrt{-g} \dd x^1...\dd x^d$. Actually, the determinant of the metric guarantees that the integrand transforms covariantly under coordinate change.

 \subsection{Hypersurfaces} \label{ssec:hypersurfaces}
  
  A hypersurface $\Sigma$ is defined giving an equation that constraints the variables,
    \begin{equation} \label{eq:hypersurface}
  S(x) = 0 \ ,
  \end{equation}
or giving parametric equations for the relationship between global coordinates $x^\alpha$ and the intrinsic coordinates $y^a$ on the hypersurface,
\begin{equation}
x^\alpha = x^\alpha(y^a)  \ .
\end{equation}
We will use the Latin alphabet for indices concerning the intrinsic geometry of the hypersurface, $a = 1,2,3$.

As the function $S$ is constant along the hypersurface, the vector $\partial_\alpha S$ is normal to it. If it is not null, then it is possible to introduce a unit normal vector,
\begin{equation}
n^\alpha = \frac{\epsilon}{\abs{g^{\alpha \beta} \partial_\alpha S \partial_\beta S}^{\frac{1}{2}}} g^{\alpha \beta} \partial_\beta S \ ,
\end{equation}
where either $\epsilon = -1 $ if the vector is time-like, and then the hypersurface is called \textit{space-like}, or $\epsilon = 1$ if the vector is space-like, and the hypersurface is called \textit{time-like}.

Null hypersurfaces (that are, hypersurfaces for which the vector $\partial_\alpha S$ is null) do not admit a unit normal vector. Nevertheless the vector $\tilde n^\alpha = g^{\alpha \beta} \partial_\beta S$ still exists, and its sign is chosen so that it is future-directed. As $\tilde n^\alpha \tilde n_\alpha = 0$ in this case, the vector $\tilde n^\alpha$ is also tangent to the hypersurface; in fact, by direct computation, one can see that
\begin{equation}
\tilde n^\alpha \nabla_\alpha \tilde n_\beta = \partial^\alpha S \nabla_\alpha \partial_\beta S =
\partial^\alpha S \nabla_\beta \partial_\alpha S = \frac{1}{2} \nabla_\beta \partial^\alpha S \partial_\alpha S \ ,
\end{equation}
but since $\partial^\alpha S \partial_\alpha S = 0 $ everywhere on the hypersurface, its gradient must be directed along $\tilde n^\alpha$, and therefore the normal vector satisfies the geodesic equation
\begin{equation}
\tilde n^\alpha \nabla_\alpha \tilde n_\beta = \kappa \tilde n_\beta \ .
\end{equation}
The geodesics to which $\tilde n$ is tangent are called \textit{generators} of the hypersurface. Now, it is worth a little digression on the geodesic equation and the parameter of geodesics. In fact, $\kappa$, called the \textit{inaffinity function}, in general is not vanishing, but it can be set to zero choosing a different parameter along geodesics, writing
\begin{equation}
n^\alpha = \dv{\lambda} = e^{-\gamma} n^\alpha = \dv{\tilde \lambda}{ \lambda} \dv{\tilde \lambda} \ ,
\end{equation}
so that $n^\alpha \nabla_\alpha n_\beta = 0$. $\lambda$ is called the \textit{affine parameter}. Imposing this condition on $n$ one finds that
\begin{equation}
    0 = (n^\mu \nabla_\mu n^\nu) = 
    \dv{\lambda}\big(e^{- \gamma}\tilde n^\nu \big) = 
    - \dv{\gamma}{\lambda} n^\nu + e^{- \gamma}\dv{\tilde \lambda}{\lambda}\dv{\tilde n^\nu}{\tilde \lambda} 
    = n^\nu (-  \dv{\gamma}{\lambda} + \dv{\tilde \lambda}{\lambda}\kappa ) \ ,
\end{equation}
so that we find an equation for $\gamma$,
\begin{equation}
\dv{\gamma}{\tilde \lambda} = \kappa \ ,
\end{equation}
from which one determines the equation for the affine parameter, $\dv{\lambda}{\tilde \lambda} = e^\gamma $.

When $\Sigma$ is null, it is natural to set $\lambda$ as one of the intrinsic coordinates on the hypersurface, so we will write $y^a = (\lambda, \theta^A)$, with $A = 1,2$: $\lambda$ is the coordinate along each generator of the hypersurface, while $\theta^A$ move between generators.

As the hypersurface is a differential manifold in its own right, one can make it a metric space defining a metric to measure distances along the hypersurface. In fact, the metric of the embedding spacetime induces a natural metric on the hypersurface, so that the line element for infinitesimal displacements on the hypersurface are measured, as in the embedding spacetime, by
\begin{equation}
\dd s^2_\Sigma = g_{\alpha \beta} \dd x^\alpha \dd x^\beta \ .
\end{equation}

We can introduce the vectors $e^\alpha_a = \pdv{x^\alpha}{y^a}$, which are tangent to the hypersurface ($n_\alpha e^\alpha_a = 0 $), so that the line element becomes
\begin{equation}
\dd s^2_\Sigma = g_{\alpha \beta} \pdv{x^\alpha}{y^a} \pdv{x^\beta}{y^b} \dd y^a \dd y^b \ .
\end{equation} 

In the time- and space-like case, this defines the induced metric on the hypersurface as
\begin{equation}
h_{ab} = g_{\alpha \beta} \pdv{x^\alpha}{y^a} \pdv{x^\beta}{y^b} \ .
\end{equation}
In the null case, however, a simplification occurs if we use the coordinates $y^a = (\lambda, \theta^A)$. Since the vector $e^\alpha_\lambda = \pdv{x^\alpha}{\lambda} = n^\alpha$ is null, $h_{\lambda \lambda} = g_{\alpha \beta} n^\alpha n^\beta = 0$, and moreover $h_{\lambda A} = 0$ because by construction $ n_\alpha e^\alpha_a = 0$. In this case $h_{ab}$ is degenerate along the $\lambda$ direction, and the nondegenerate part of the induced metric is actually two-dimensional,
\begin{equation} \label{eq:surface-element-null}
\dd s^2_\Sigma = \sigma_{AB} \dd \theta^A \dd \theta^B \ ,
\end{equation} 
with 
\begin{equation}
\sigma_{AB} = g_{\alpha \beta} \pdv{x^\alpha}{\theta^A}\pdv{x^\beta}{\theta^B} \ .
\end{equation}

\subsection{Gauss-Stokes Theorem} \label{ssec:gauss-stokes-theorem}
      
We can use the general definition of integration of a d-form in a d-dimensional spacetime \eqref{def:integration-forms} to define the integration of p-forms on a p-dimensional hypersurface $\Sigma$. We restrict to a four-dimensional spacetime with an embedded three-dimensional hypersurface. The hypersurface has a natural structure of differential manifold induced by the embedding spacetime, and in particular there is an induced metric $h_{\alpha \beta}$ (or $\sigma_{AB}$ in the null case). If the hypersurface is not null, it is immediate to introduce the surface element
\begin{equation} \label{eq:surface-element-not-null}
\dd \Sigma = \sqrt{h} \dd^3 y
\end{equation}
which is just \eqref{def:integration-forms} expressed for the hypersurface viewed as a manifold. However, the hypersurface has also a natural \textit{directed} surface element, given by
\begin{equation}
\dd \Sigma^\alpha = n^\alpha \sqrt h \dd^3 y \ .
\end{equation}

If the hypersurface is null we can introduce an oriented surface element, using coordinates $ (\lambda, \theta^A) $, where $A = 1,2$, as
\begin{equation}
\dd \Sigma^\alpha = n^\alpha \sqrt{\sigma} \dd \lambda \dd^2 \theta  \ .
\end{equation}

We are finally in the position to state the Gauss-Stokes theorem. We will first use the abstract language of differential forms, and then we will re-express it in the abstract index notation to make its meaning more concrete. Consider an orientable manifold $\mathscr M$ with boundary. Such a manifold is locally isomorphic not to $\mathbb R^d$ but to half of $\mathbb R^d$, the subset with $x^0 \leq 0$. The boundary $\partial \mathscr M$, is the set of points in $\mathscr M$ mapped to $x^0 = 0$ by local charts. The orientation of $\mathscr M$ induces an orientation in $\partial \mathscr M$, and therefore the boundary is an orientable manifold in itself. The Gauss-Stokes Theorem states that, for a (d-1)-form field $\omega$ on $\mathscr M$,
\begin{theorem}{Gauss-Stokes Theorem} \label{th:gauss-stokes-theorem}
\begin{equation}
\int_{\mathscr M} \dd \omega = \int_{\partial \mathscr M} \omega \ .
\end{equation}
\end{theorem}

This theorem expresses in a unified language the fundamental theorem of calculus, Gauss', and Stokes' theorems, generalising them to the case of any differential manifold with boundary. However, to understand its significance for actual computations, we rewrite it in index notation. We consider a 4-dimensional submanifold $\mathscr R \subset \mathscr M$ with a boundary $\partial \mathscr R$, because we will apply the theorem to this case. A (d-1)-form $\omega$ can be obtained from the volume form $\Omega$ in $\mathscr R$ and any $C^\infty$-vector field $v^a$ by
\begin{equation}
\omega_{a_1 a_2 a_3} = \Omega_{b a_1 a_2 a_3} v^b \ .
\end{equation}
Its exterior differential is
\begin{equation}
(\dd \omega)_{c a_1 a_2 a_3} = 4 \nabla_{[c} ( \Omega_{b a_1 a_2 a_3]} v^b) = 4 \Omega_{b [a_1 a_2 a_3} \nabla_{c]} v^b \ .
\end{equation}
This is immediate because the connection is metric, and therefore $\nabla \Omega = 0$. However, since the right-hand side is a totally antisymmetric tensor of rank $d$ it must be proportional to the volume form,
\begin{equation}
4 \Omega_{b [a_1 a_2 a_3} \nabla_{c]} v^b  = f \Omega_{a_1 a_2 a_3 c} \ .
\end{equation}
To find the proportionality function $f$ we compute the contraction with $\Omega^{a_1 a_2 a_3 c}$ of both sides. The contraction of the left-hand side gives
\begin{equation}
4 \sqrt g \epsilon^{c a_1 a_2 a_3} \epsilon_{b [ a_1 a_2 a_3} \nabla_{c]} v^b =
4 \sqrt g \epsilon^{c a_1 a_2 a_3} \epsilon_{b a_1 a_2 a_3} \nabla_{c} v^b =
24 \ \sqrt g \ \delta^c_b \ \nabla_c v^b = 24 \sqrt g \nabla_b v^b \ .
\end{equation}
The factor of 6 comes from the contraction of the two $\epsilon$.
The contraction of the right-hand side gives
\begin{equation}
f \sqrt{g} \epsilon_{a_1 a_2 a_3 c} \epsilon^{a_1 a_2 a_3 c} = 24 f \sqrt g \ .
\end{equation}
Therefore, we find
\begin{equation}
(\dd \omega)_{c a_1 a_2 a_3} = \nabla_b v^b \Omega_{c a_1 a_2 a_3} \ ,
\end{equation}
and the left-hand side of the Gauss-Stokes theorem can be rewritten as
\begin{equation}
\int_{\mathscr R} \nabla_b v^b \Omega_{c a_1 a_2 a_3} = \int_{\mathscr R} \nabla_b v^b \sqrt{-g} \ \dd^4 x \ .\
\end{equation}

The right-hand side of the Gauss-Stokes theorem can be rewritten as well, noting that the volume form on $\mathscr M$ induces a volume form $\tilde \Omega$ on $\partial \mathscr M$. If we introduce the normal vector $n$ to $\partial \mathscr M$, we write
\begin{equation}
\frac{1}{4} \Omega_{a_1 a_2 a_3 a_4} = n_{[a_1} \tilde \Omega_{a_2 a_3 a_4]} \ ,
\end{equation}
and therefore $v^b \Omega_{b a_1 a_2 a_3} = v^b n_b \tilde \Omega_{a_1 a_2 a_3}$. The induced volume form $\tilde \Omega$ is defined with $\sqrt h$ in the time- and space-like case, and with $\sqrt \sigma$ in the null case. Therefore, the left-hand side becomes
\begin{equation}
\int_{\partial \mathscr R} v^b \Omega_{b a_1 a_2 a_3} = \int_{\partial \mathscr R} v^b n_b \sqrt h \dd^3 y
\end{equation}
in the time- and space-like case, while in the light-like case it becomes
\begin{equation}
\int_{\partial \mathscr R} v^b \Omega_{b a_1 a_2 a_3} = \int_{\partial \mathscr R} v^b n_b \sqrt \sigma \dd \lambda \dd^2 \theta \ .
\end{equation}
The Gauss-Stokes Theorem can finally be written in the more familiar form
\begin{equation}
\int_{\mathscr R} \nabla_\mu v^\mu \dd \text{vol}_{\mathscr R} = \int_{\partial \mathscr R} v_\mu n^\mu \dd \Sigma \ .
\end{equation}
This is the expression we will use in chapter \ref{ch:new-results}. In fact, if one is given a conserved current $\nabla_\mu j^\mu = 0$ in a region $\mathscr R$, one sees immediately from the Gauss-Stokes theorem that its flux across the boundary of the region vanishes.

    \section{Geodesic Congruences} \label{sec:geodesic-congruences}
    
    In a region $\mathscr R$, consider a family of geodesics that cross each point of the region once: this is a \textit{geodesic congruence}. It is best visualised as a bundle of wires held tight with a rubber band. We want to study the deviation of geodesics from one another as we move along the proper time of one geodesic of reference. In particular, we want to find an evolution equation for the \textit{deviation vector} $\xi$ between two neighbouring geodesics.
    The treatment of geodesic congruence is different in the two cases, whether the geodesics are null or not; in general, the null case is a bit more involved because the transverse spacetime is actually 2-dimensional, and therefore one must pay more care constructing vectors in the transverse directions. The space-like case is completely analogous to the time-like one, so we stick with the latter for definiteness.
    
    We start with the time-like case, defining the tangent vector to the geodesic congruence, $u$, that satisfies the geodesic equation $u^\mu \nabla_\mu u^\nu = 0$ and is normalized, $\norm{u}^2 = -1$. The deviation vector $\xi$ satisfies the two relations
    \begin{align}
    \xi^\alpha u^\beta g_{\alpha \beta} &= 0 \label{def-xi-1}\\
   \mathcal L_{u} \xi = 0 \Rightarrow \xi^\alpha \nabla_\alpha u^\beta &= u^\alpha \nabla_\alpha \xi^\beta \label{def-xi-2} \ ,
    \end{align}
where $\mathcal L_u$ is the Lie derivative along $u$. The first condition says that the deviation vector is everywhere orthogonal to the direction of geodesic evolution. Then, we introduce the transverse metric, $h_{\alpha \beta} = g_{\alpha \beta} + u_\alpha u_\beta $ and the tensor field $B_{\alpha \beta} = \nabla_\alpha u_\beta $. This tensor field measures the failure of $\xi$ to be parallelly transported along the geodesics, $B_\alpha^\beta \xi^\alpha = \nabla_u \xi^\beta $ immediately from \eqref{def-xi-2}. $B_{\alpha \beta}$ therefore captures the information on the evolution of $\xi$ along the geodesics. It follows from the definition, the geodesic equation, and the normalization, the fact that $B_{\alpha \beta} u^\alpha = B_{\alpha \beta} u^\beta = 0 $. The tensor $B_{\alpha \beta}$ can be decomposed in its symmetric, antisymmetric parts and its trace:
\begin{align}
B^\alpha_\alpha = \nabla_\alpha u^\alpha = \theta \label{def:expansion} \\
\sigma_{\alpha \beta} = B_{(\alpha \beta)} - \frac{1}{3}\theta \label{def:shear}\\
\omega_{\alpha \beta} = B_{[\alpha \beta]} \label{def:rotation} \ .
\end{align}
They are called, respectively, expansion, shear, and rotation. The tensor field then is written as
\begin{equation}
B_{\alpha \beta} = \frac{1}{3}\theta h_{\alpha \beta} + \sigma_{\alpha \beta} + \omega_{\alpha \beta} \ .
\end{equation}
From the property of $B_{\alpha \beta}$ follows that each term is transverse: $ h_{\alpha \beta} u^\alpha = \sigma_{\alpha \beta} u^\alpha = \omega_{\alpha \beta} u^\alpha = 0 $.

To understand their physical meaning, suppose that at some proper time the cross section of the congruence is a circle. The shear then parametrises the deformation of the circle into ellipses; the rotation, as one may guess, the rotation of the ellipses along the direction of $u$; and the expansion scalar the change in size of the cross section. In particular, if $\theta > 0$ means the cross-sectional area is growing, and contracting otherwise.

In the light-like case there are some subtleties in the definition of the tensor $B_{\alpha \beta} $, because if $u$ is light-like the condition $u^\alpha \xi_\alpha = 0$ is not sufficient to guarantee that $\xi$ has no components along the $u$ direction. In fact, in the light-like case we can introduce a tensor field $ B_\alpha^\beta = \nabla_\alpha u^\beta $, satisfying the condition $B_\alpha^\beta u^\alpha = B_{\alpha \beta}u^\beta = 0 $, but there remains a transverse component of $B_{\alpha \beta}$ that should be eliminated. We therefore introduce an auxiliary null vector, $N$, which satisfies $N^\alpha u^\beta g_{\alpha \beta} = -1 $. If we define the transverse metric $h_{\alpha \beta} = g_{\alpha \beta} + N_\alpha u_\beta + N_\beta u_\alpha $, one sees immediately that it is transverse to both $u$ and $N$, since $h_{\alpha \beta} u^\alpha = h_{\alpha \beta} N^\alpha = 0$, and therefore we can use it to select the transverse part of the deviation tensor, $\tilde B_{\alpha \beta} = h_\alpha^\mu h_\beta^\nu B_{\mu \nu} $.

Then, we can decompose the deviation tensor as before,
\begin{equation}
\tilde B_{\alpha \beta} = \frac{1}{2}\theta h_{\alpha \beta} + \sigma_{\alpha \beta} + \omega_{\alpha \beta} \ .
\end{equation}
It can easily be seen that the expansion is given by $\theta = \tilde B^\alpha_\alpha = B^\alpha_\alpha $, writing $\tilde B$ explicitly and using the fact that $B$ is transverse to $u$. The difference in the numerical factor arises because the relevant space is now two-dimensional.

In the case in which the geodesic congruence is defined as the family of geodesics orthogonal to some hypersurface, a simplification occurs. This is the case of interest for us, since we will deal with congruences orthogonal to a black hole horizon. Therefore, suppose that a hypersurface is given by some equation $S(x) = 0$, with normal vector $n^\alpha \propto \partial^\alpha S$. If the congruence is hypersurface orthogonal then the tangent vector $u$ is proportional to $n$, and since $S$ is a scalar function we can change the ordinary derivative with the covariant derivative:
\begin{equation}
u_\alpha = - \mu \nabla_\alpha S
\end{equation}
for some coefficient of proportionality $\mu$. In this case, the rotation tensor is
\begin{equation}
\omega_{\alpha \beta} = \nabla_{[\alpha}u_{\beta]} = - \nabla_{[\alpha}(\mu \nabla_{\beta]}S) = - \nabla_{[\alpha} \mu \nabla_{\beta]}S = \frac{1}{\mu} \nabla_{[\alpha} \mu \ u_{\beta]} = \frac{1}{2}(u_\beta \nabla_\alpha \mu  - u_\alpha \nabla_\beta \mu )
\end{equation}
since $\nabla_\alpha \nabla_\beta S$ is symmetric. On the other hand, the condition $\omega_{\alpha \beta} u^\beta = 0$ implies that $ \nabla_\alpha \mu = - u_\alpha u^\beta \nabla_\beta \mu $ since $u_\beta u^\beta = -1$, and this gives $\omega_{\alpha \beta} = 0$.

Note that we never mention if the tangent vector is time-like or null, and therefore the computation holds in both cases. If the congruence is light-like, then the geodesics are actually the generators of the hypersurface.

The converse is also true: if the rotation tensor is vanishing, then it is possible to find a family of hypersurfaces to which the geodesic congruence is everywhere orthogonal. The result is known as the \textit{Frobenius Theorem} (see, e.g., \cite{Wald84}, Appendix B.3). 

        \subsection{Geometric Meaning of the Expansion} \label{ssec:geometric-meaning-expansion}
    Now we focus our attention on the expansion scalar. In particular, we want to make clear its interpretation as the variation of the cross-sectional volume of the congruence.
    
    We start deriving the relation for time-like congruences. On a geodesic of the congruence, we select two point, $p, \ q$, with proper time $\tau_p$ respectively $\tau_q$. In their neighbourhood we select a set of points $\delta\Sigma_p$ and $\delta\Sigma_q$, portions of the two hypersurfaces $\tau = \tau_p$ and $\tau = \tau_q$ such that a geodesic of the congruence passes through each point. In $\delta \Sigma_{p,q}$ we construct a local reference frame $(\tau, y^a)$, which is related to the original system of coordinates $x^\alpha$ via
    \begin{equation}
    u^\alpha = \pdv{x^\alpha}{\tau} \qquad e^\alpha_a = \pdv{x^\alpha}{y^a} \ .
    \end{equation}
Finally, we have the induced metric on the hypersurface, $h_{ab} = e^\alpha_a e^\beta_b g_{\alpha \beta} $, with determinant $\sqrt h$. We can find the inverse metric $h^{ab}$ computing
\begin{equation}
 1 = h^{ab} h_{ab} = h^{ab} e^\alpha_a e^\beta_b g_{\alpha \beta} = h^{ab} e^\alpha_a e^\beta_b ( h_{\alpha \beta} - u_\alpha u_\beta ) = h^{ab} e^\alpha_a e^\beta_b h_{\alpha \beta} \ ,
\end{equation} 
and therefore
\begin{equation}
h^{\alpha \beta} = h^{a b} e^\alpha_a e^\beta_b \ .
\end{equation} 
The cross-sectional volume element is $\delta V = \sqrt h \dd^3 y$. Since $y^a$ are transverse coordinates the variation of the volume element along the geodesics is entirely in the determinant, $\sqrt h$. The relative variation is therefore
\begin{equation} \label{eq:middle-passage-geometry-expansion}
\frac{1}{\delta V} \dv{\delta V}{\tau} = \frac{1}{ \sqrt h}\dv{\sqrt h}{\tau} =  \frac{1}{2h} \dv{h}{\tau} = \frac{1}{2} h^{ab} \dv{h_{ab}}{\tau} \ .
\end{equation}
In the last passage, we used the matrix identity
\begin{equation}
\frac{1}{\det A} \dv{\det A}{\tau} = \Tr[ B^{-1} \dv{B}{\tau}] \ .
\end{equation}
This identity can be proven from $e^{\Tr B} = \det e^B$: taking the logarithm of both sides and deriving with respect to $\tau$ one obtains
\begin{equation}
\Tr \dv{B}{\tau} = \frac{1}{\det e^B} \dv{\det e^B}{\tau} \ ,
\end{equation}
and calling $e^B = A$ the identity follows.

Going back to \eqref{eq:middle-passage-geometry-expansion}, we note that the variation of the induced metric can be computed using the deviation tensor. By direct computation one finds
\begin{equation}
    \dv{h_{ab}}{\tau} = g_{\alpha \beta} u^\mu \nabla_\mu ( e^\alpha_a e^\beta_b ) = g_{\alpha \beta} ( e^\beta_b u^\mu \nabla_\mu e^\alpha_a + e^\alpha_a u^\mu \nabla_\mu e^\beta_b ) \ ,
\end{equation}  
     but $u^\mu \nabla_\mu e^\alpha_a = u^\mu \partial_\mu e^\alpha_a + \Gamma^\alpha_{\mu \nu} u^\mu e^\nu_a u^\mu = \Gamma^\alpha_{\mu \nu} u^\mu e^\nu_a = e^\mu_a \partial_\mu u^\alpha + \Gamma^\alpha_{\mu \nu} u^\nu e^\mu_a = e^\mu_a \nabla_\mu u^\alpha $, since the two vectors are orthogonal, and similarly for $u^\mu \nabla_\mu e^\beta_b$. So
     \begin{equation}
     \dv{h_{ab}}{\tau} = g_{\alpha \beta} (e^\beta_b e^\mu_a \nabla_\mu u^\alpha + e^\alpha_a e^\mu_b \nabla_\mu u^\beta ) 
     = g_{\alpha \beta} ( e^\mu_a e^\beta_b B_\mu^\alpha + e^\alpha_a e^\mu_b B_\mu^\beta )
     = e^\alpha_a e^\beta_b ( B_{\alpha \beta} + B_{\beta \alpha}) \ .
\end{equation}
Multiplying by $h_{ab}$, and remembering the relation with $h^{\alpha \beta}$ and the fact that $B_{\alpha \beta}$ is transverse, we arrive at
\begin{equation}
    \begin{split}
\frac{1}{2}h^{ab} \dv{h_{ab}}{\tau} &= \frac{1}{2} h^{ab} e^\alpha_a e^\beta_b (B_{\alpha \beta} + B_{\beta \alpha} ) \\
&= \frac{1}{2} h^{\alpha \beta} (B_{\alpha \beta} + B_{\beta \alpha} ) \\
& = \frac{1}{2} g^{\alpha \beta} (B_{\alpha \beta} + B_{\beta \alpha} ) = B^\alpha_\alpha = \theta 
\end{split} \ .
\end{equation}
Thus, the expansion scalar is the logarithmic derivative of the cross-sectional volume element.

The light-like case is completely analogous. The only difference is that the transverse space is two-dimensional; because of the ambiguity of the condition $u^\alpha h_{\alpha \beta} = 0$, one constructs an auxiliary geodesic congruence tangent to the auxiliary vector $N^\alpha$, and finds the transverse hypersurface fixing the proper time on a geodesic of the original congruence \textit{and} on a geodesic of the congruence tangent to $N$. It is clear that the transverse space is therefore two-dimensional, and then, following the same passages, the expansion scalar is the logarithmic derivative of the cross-sectional \textit{area},
\begin{equation} \label{eq:expansion-area-light-like}
\theta = \frac{1}{\sqrt \sigma} \dv{\sqrt \sigma}{\tau} \ ,
\end{equation}
where $\sigma^{ab}$ is the two-dimensional induced metric.

        \subsection{Raychaudhuri Equation} \label{ssec:Raychaudhuri-eq}
        
Now, we want to derive an evolution equation for the expansion scalar. The expansion scalar is an important quantity because it describes the evolution of congruences but also the temporal evolution of hypersurfaces, thanks to its geometric meaning, and it plays a fundamental role in the singularity theorems by Hawking and Penrose and in Hawking's area theorem, which started the study of black hole thermodynamics. In chapter \ref{ch:new-results} we will use the Raychaudhuri equation to compute the growth of black hole horizons due to matter crossing the horizon.

Raychaudhuri equation follows from direct computation. The derivation for time-like and light-like congruences is identical: the only relevant difference resides in the numerical factor in front of the expansion scalar in the decomposition of the deviation tensor. For definiteness we assume a null congruence, because it is the case we will deal with in chapter \ref{ch:new-results}. We look for the evolution of the expansion along the geodesics, so we start with the deviation tensor $B_{\alpha \beta}$
\begin{equation}
    \begin{split}
    u^\mu \nabla_\mu B_{\alpha \beta} &= u^\mu \nabla_\mu \nabla_\alpha u_\beta = \\
    & = u^\mu \nabla_\alpha \nabla_\mu  u_\beta - R_{ \alpha \mu \beta}^\nu  u^\mu u_\nu \\
    &= \nabla_\alpha(u^\mu \nabla_\mu u_\beta) - \nabla_\alpha u^\mu \nabla_\mu u_\beta - R_{\alpha \nu \beta \mu} u^\mu u^\nu \\
    &= - \tilde B_\alpha^\mu \tilde B_{\mu \beta} - R_{\alpha \nu \beta \mu} u^\mu u^\nu \ .
    \end{split}
\end{equation}
In the second line we used the definition of the Riemann tensor, in the third the Leibniz rule, and in the last the geodesic equation. Taking the trace gives
\begin{equation} \label{eq:Raychaydhuri-middle-passage}
\dv{\theta}{\tau} = - B^{\mu \alpha} B_{\mu \alpha} - R_{\mu \nu} u^\mu u^\nu  \ ,
\end{equation}
but
\begin{equation}
\tilde B^{\alpha \beta} \tilde B_{\alpha \beta} = h^\alpha_\mu h^\beta_\nu B^{\mu \nu} \ h_\alpha^\rho h_\beta^\sigma B_{\rho \sigma} = h^\rho_\mu h^\sigma_\nu B^{\mu \nu} B_{\rho \sigma} = B^{\mu \nu}B_{\mu \nu}
\end{equation}
because $B_{\mu \nu}$ is transverse. On the other hand, we have
\begin{multline}
\tilde B^{\alpha \beta} \tilde B_{\alpha \beta} = (\frac{1}{2} \theta h^{\alpha \beta} + \sigma^{\alpha \beta} + \omega^{\alpha \beta} ) ( \frac{1}{2} \theta h_{\alpha \beta} + \sigma_{\alpha \beta} + \omega_{\alpha \beta}) = \\
= \frac{1}{2}\theta^2 + \sigma^{\alpha \beta} \sigma_{\alpha \beta} + \omega^{\alpha \beta} \omega_{\alpha \beta} \ ,
\end{multline}
and so we can write the \textit{Raychaudhuri equation}:
\begin{equation} \label{eq:Raychaudhuri}
\dv{\theta}{\tau} = - \frac{1}{2} \theta^2 - \sigma^{\alpha \beta} \sigma_{\alpha \beta} - \omega^{\alpha \beta} \omega_{\alpha \beta} - R_{\mu \nu} u^\mu u^\nu \ .
\end{equation} 

As we said, the time-like case follows the same derivation, in particular equation \eqref{eq:Raychaydhuri-middle-passage} holds. The only difference is in the decomposition of the deviation tensor, which results in a factor of $\frac{1}{3}$ in front of the expansion term, instead of $\frac{1}{2}$. 

We remark here on a couple of modifications to the Raychaudhuri which we will use later, in chapter \ref{ch:new-results}. First, we note that in the null case we can rewrite the equation in terms of the stress-energy tensor. In fact, Einstein's equations
\begin{equation}
R_{\mu \nu} - \frac{1}{2}Rg_{\mu \nu} = 8 \pi T_{\mu \nu}
\end{equation}
can be used to rewrite the last term in the right-hand side. Since $u$ is light-like, the curvature scalar term vanishes, and we find
\begin{equation}\label{eq:Raychaudhuri-SET}
\dv{\theta}{\tau} = - \frac{1}{2} \theta^2 - \sigma^{\alpha \beta} \sigma_{\alpha \beta} - \omega^{\alpha \beta} \omega_{\alpha \beta} - 8 \pi T_{\mu \nu} u^\mu u^\nu \ .
\end{equation}

A second modification happens if we do not consider generators affinely parametrised, but that satisfy the more general geodesic equation
\begin{equation}
u^\alpha \nabla_\alpha u^\beta = \kappa u^\beta \ ,
\end{equation}
where $\kappa$ is called the \textit{inaffinity function}. In this case, the expansion becomes
\begin{equation}
\theta = \nabla_\alpha u^\alpha - \kappa
\end{equation}
and the Raychaudhuri equation acquires an additional term,
\begin{equation} \label{eq:Raychaudhuri-inaffine}
\dv{\theta}{\tau} = \kappa \theta - \frac{1}{2} \theta^2 - \sigma^{\alpha \beta} \sigma_{\alpha \beta} - \omega^{\alpha \beta} \omega_{\alpha \beta} - R_{\mu \nu} u^\mu u^\nu \ .
\end{equation}
As a conclusive remark we want to show how Raychaudhuri equation can be used to prove a theorem on the evolution of congruences, called \textit{Focussing Theorem}. From the right hand side of \eqref{eq:Raychaudhuri-SET} one immediately sees that if the null energy condition holds,
\begin{equation}
T_{\mu \nu} u^\mu u^\nu \geq 0 \ ,
\end{equation}
then $\dv{\theta}{\tau} \leq 0$, a consequence of the fact that gravity is always attractive. If we assume that the shear and rotation tensor vanish, $\dv{\theta}{\tau} \leq -\frac{1}{c}\theta^2 $, where $c=2$, $c=3$ in the light-like respectively time-like case, implying
\begin{equation}
\dv{\theta^{-1}}{\tau} \geq \frac{1}{c} \Rightarrow \theta^{-1}(\tau) \geq \theta^{-1}(0) + \frac{1}{c} \tau \ .
\end{equation}
As a consequence, if the congruence is initially converging, $\theta(0) < 0$, then $\theta^{-1}$ must become zero, that is, the expansion scalar diverges in a proper time $\tau \leq \frac{c}{\theta(0)}$. This divergence of the expansion occurs when the geodesics converge into a point, called caustics. This is just a divergence in the congruence, but together with global properties of the spacetime can be used to prove the existence of singularities of the structure of the spacetime itself, a fundamental result contained in Penrose and Hawking singularity theorems.

    \section{Penrose Diagrams} \label{sec:Penrose-diagrams}

Penrose diagrams are probably the most useful tool to visualise the causal structure of a spacetime. The goal is to reduce the spacetime to a finite size, so that it is possible to represent it on a sheet of paper, while preserving the light-cone structure. Basically, since light-cones, even in highly curved spacetimes, are identified by perpendicular light rays, the idea is to map the spacetime $\mathscr M$ in a subregion of another, unphysical one, so that drawing a finite portion of the unphysical spacetime $\tilde{\mathscr M}$ one can actually represents the original spacetime $\mathscr M$. By doing so, we can identify a set of points in the unphysical spacetime which represent the causal boundaries of the original spacetime, while we need to preserve the angles between the light rays so that it is still possible to visualize the original causal structure.

 This compression-without-loss of the spacetime is achieved through what is called \textit{conformal compactification}: given a spacetime $\mathscr M$ with a metric $g$ such that
\begin{equation}
\dd s^2 = g_{\alpha \beta} \dd x^\alpha \dd x^\beta
\end{equation}
we want to find a new, unphysical spacetime $\tilde{\mathscr M}$ with a metric $\tilde g$, given by
\begin{equation}
\omega^{-2} \dd \tilde s^2 =  \dd s^2
\end{equation}
where $\omega > 0$.

Since the two metrics are related by a conformal transformation, they share the same causal structure, because angles are locally preserved. The original spacetime is however embedded in the unphysical one, and therefore the intervals of the original coordinates, which usually extend to infinity (at least in one direction, as the radial coordinate in polar coordinates) are compressed to a finite interval in the new one, and taking the extremes of these intervals, we can identify the causal boundaries of $\mathscr M$.  Drawing the portion of $\tilde{\mathscr M}$ corresponding to $\mathscr M$ and its boundaries, one finally obtains the Penrose diagram of $\mathscr M$.

There is one more caveat, namely, a Penrose diagram cannot be more than 2-dimensional, due to the well known limitations of sheets of paper. Therefore, one needs to choose the relevant directions to draw, leaving the others implicit. This is the reason why Penrose diagrams are most useful in dealing with spacetimes with spherical symmetry, because one can draw the $t-r$ plane only, assuming that each point represents a 2-sphere. It is also possible to draw Penrose diagrams for axi-symmetric spacetimes, as the Kerr black hole, but in that case the Penrose diagram is no more compact but only periodic.

The best way to understand Penrose diagrams is through examples. We show now how to construct the simplest of all, that of Minkowski spacetime:
\begin{equation}
\dd s^2 = - \dd t^2 + \dd r^2 + r^2 \dd \Omega^2 \ , 
\end{equation}
where we adopted polar coordinates, with $\dd \Omega^2 = \dd \theta^2 + \sin^2 \theta \dd \varphi^2 $ being the angular metric of a unit 2-sphere.

As we said, the goal is to write the line element in the form $\dd s^2 = \omega^{-2} \dd \tilde s^2$, so we need to perform a couple of change of coordinates, carefully keeping track of the intervals in which the variables are defined. First, we introduce a set of variables adapted to the light-cone structure, called \textit{light-like coordinates}:
\begin{align}
u = t - r \\
v = t + r
\end{align}
$u = u_0 $ and $v = v_0$ identify, respectively, outgoing and ingoing lightrays. Both $u,v$ take values in all $\mathbb R$, but since $r = \frac{v - u}{2} \geq 0$ we have the additional condition $u \leq v$. In these coordinates the metric becomes
\begin{equation}
\dd s^2 = - \dd u \dd v + \frac{(v-u)^2}{4} \dd \Omega^2 \ .
\end{equation}
Since the coordinates are extended to infinity, we need a function to compactify them. The best \textit{squishification} function is the $\arctan$, so we define
\begin{align}
U = \arctan u \\
V = \arctan v
\end{align}
Now, $U,V \in (-\frac{\pi}{2}, \frac{\pi}{2})$, and $U \leq V $. The differentials in the metric become
\begin{equation}
\dd u \dd v = \dd(\tan U) \dd (\tan V) = \frac{1}{\cos^2 U \cos^2 V} \dd U \dd V
\end{equation}
while the prefactor in front of the angular variables becomes
\begin{equation}
v - u = \tan V - \tan U = \frac{\sin(V-U)}{\cos U \cos V}
\end{equation} 
and so the metric is
\begin{equation}
\dd s^2 = \frac{1}{4 \cos^2 U \cos^2 V} \bigg[ - 4 \dd U \dd V + \sin^2(V-U) \dd \Omega^2  \bigg] \ .
\end{equation}
The conformal factor is $\omega^2 = \cos^2 U \cos^2 V$. The unphysical metric is given by
\begin{equation}
\dd \tilde s^2 = - 4 \dd U \dd V + \sin^2(V - U) \dd \Omega^2 \ .
\end{equation}
To better understand its structure, we can go back to time-like and radial-like coordinates, 
\begin{align}
T = U + V \\
R = V - U
\end{align}
so that the unphysical metric becomes
\begin{equation}
\dd s^2 = - \dd T^2 + \dd R^2 + \sin^2 R \dd \Omega^2 \ .
\end{equation}
We are still assuming that $ - \pi < T + R < \pi, \ - \pi < T - R < \pi, R \geq 0$, but now we see that we can extend the coordinates so that $T \in \mathbb R $ while $ R, \theta, \varphi \in S^3$, that are, angular variables of a unit 3-sphere. Such a spacetime is called \textit{Einstein universe}. From the metric, we see that it is similar to Minkowski, with the difference that the radial variable is not extended to infinity but rather is periodic: if we suppress the angular variables $\theta, \varphi$, the Einstein universe can be visualised as a cylinder.

\begin{figure}
\centering
\includegraphics[scale=0.3]{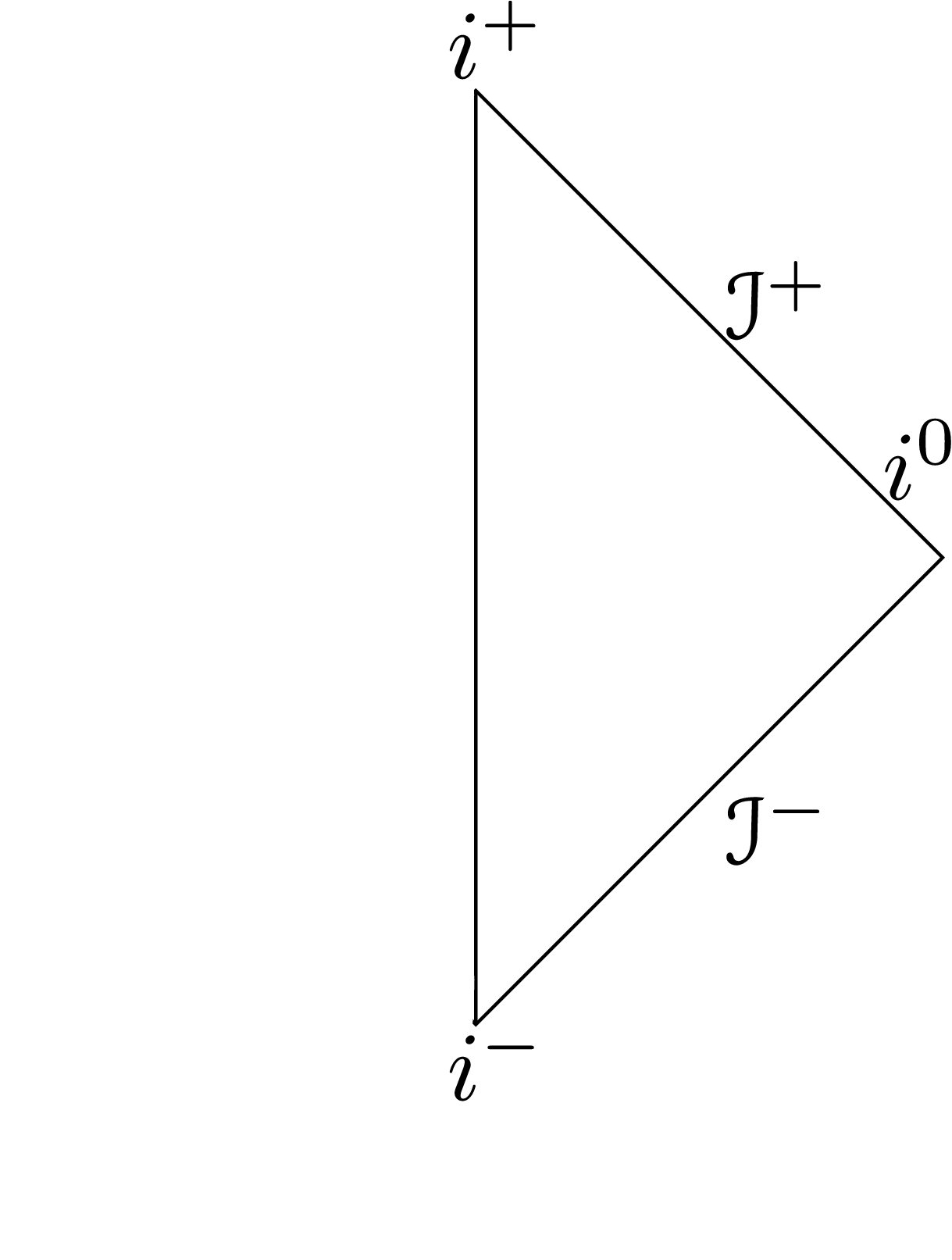}
\caption{Penrose diagram of Minkowski spacetime. Each point corresponds to a 2-sphere, expect $i^\pm$ and $i^0$, which are single points, and the vertical line, corresponding with the single point $r = 0$.}
\end{figure}

The portion of Einstein Universe covered by $- \frac{\pi}{2} \leq U \leq V \leq \frac{\pi}{2} $ (note the inclusion of the extremes) is the conformal compactification of Minkowski spacetime; the lines $ U,V = - \frac{\pi}{2} $ and $U,V = \frac{\pi}{2} $ are the causal boundaries of Minkowski space. In the $T-R$ plane (with $T$ on the vertical axis and $R$ on the horizontal axis, as usual) these four lines intersect to form a square, or a diamond, centered at the origin, with the four vertices at the intersection of $V = \pm \frac{\pi}{2}$ with $U = \pm \frac{\pi}{2}$. The condition $R \geq 0 $ finally identifies the Minkowski compactification with the right half-square, the triangle in the $R \geq 0$ half-plane.

In the compactified Minkowski spacetime we can identify three types of boundaries: the points $(U,V) = ( \frac{\pi}{2}, \frac{\pi}{2} ) $ and $(U,V) = (-\frac{\pi}{2}, - \frac{\pi}{2} ) $ are called, respectively, \textit{time-like future} and \textit{time-like past infinity}, denoted with $i^\pm$, because they are the points where time-like geodesics terminates and originates (or, rather, we should say that time-like geodesics approaches $i^\pm$ for arbitrary large/negative amounts of proper time). In the same way, space-like geodesics originates and terminates at $i^0 = (U = - \frac{\pi}{2}, V = \frac{\pi}{2}) $ and therefore this point is called \textit{space-like infinity}. The two segments $U = -\frac{\pi}{2}$ and $V = \frac{\pi}{2}$ are the set of points where null geodesics originates and terminates, and so are called \textit{past and future null infinity}.

The behaviour of Minkowski conformal infinity is somewhat the "standard" one: a spacetime is called \textit{asymptotically flat} if it has the same conformal structure at infinity of Minkowski spacetime.

Minkowski spacetime is the simplest case, but we will see that the construction of the Penrose diagram for different spacetimes (in particular, we will deal with Schwarzschild and FLRW spaces) goes along the same lines: i) introduce light-like coordinates; ii) use a squishification function, usually the arctangent function, to bring the infinity into a finite interval; iii) identify the conformal factor and the unphysical spacetime; iv) draw the portion of the unphysical spacetime corresponding to the original one, keeping the relevant coordinates.

\chapter{Quantum Field Theory in Curved Spacetime} \label{ch:QFTCS}
    \section{Motivations} \label{sec:motivations}
       
The two greatest achievements of theoretical Physics over the last century are General Relativity (GR) and Quantum Field Theory (QFT). The former changed our understanding on the beginning and fate of the Universe, giving us the picture of a dynamical spacetime with ripples and even fractures, which both have been directly observed in recent years with the first detection of gravitational waves \cite{GW2016} and the first image of the event horizon of a black hole \cite{EHTelescope19}. The latter described the constituents of matter at the smallest scales ever reached. QED is the most precise theory we have so far \cite{Gabrielse06}, while the Standard Model of particle physics proved correct with outstanding observational confirmations, as the discovery of the Higgs boson at LHC in 2012 \cite{HiggsBoson2012}. Together, these two theories are the pillars of modern science, and will stand at least as incredibly accurate Effective Field Theories for experiments with energy scales below the threshold of new, yet unknown physics.

Famously, the striking difference between the geometric interpretation of gravity and the quantum nature of all other forces lead physicists to quest for a quantum theory of gravity, reconciling both views. In 1975, Hawking began his celebrated work on the radiation of black holes \cite{Hawking74} saying that, despite fifteen years of research, no fully satisfactory and consistent theory of gravity was found yet. Although enormous progress has been made since 1975, we can say that, despite at least five decades of research, we do not yet have a fully satisfactory and consistent quantum theory of gravity. However, even if the full quantum gravity theory is not known, there are situations in which the quantum nature of matter and the geometric nature of gravity combine, giving rise to new effects which can be predicted neither from General Relativity nor from Quantum Field Theory alone, but from the combination of both. The archetypical examples are black hole thermodynamics, to which Hawking himself gave seminal contributions with his paper, and early Universe cosmology. In both cases, quantum fluctuation becomes relevant because they are enhanced by the dynamical nature of spacetime: in the case of Hawking radiation, spacetime curvature polarize the vacuum, creating particle-antiparticle pairs; in the inflationary era, quantum fluctuations of the fields are expanded to cosmic scale, seeding the large-scale structure of the Universe.

To correctly account for these phenomena an hybrid theory is needed, describing the propagation of quantum fields over a classical, curved spacetime. Therefore, Quantum Field Theory in Curved Spacetime (QFTCS) has to be considered an approximation for those phenomena where gravitational effects are relevant, but quantum gravity itself may be neglected, and must be based on a suitable unification of the principles of General Relativity and Quantum Field Theory.

We already reviewed the principles of GR in section \ref{sec:most-beautiful-theory}. They are that i) the spacetime structure is described by a Lorentzian manifold; ii) the metric and the matter fields are dynamical, with their evolution locally determined by partial differential equations. On the other hand, it is much more difficult to identify the correct principles of QFT that can be generalised to curved backgrounds. The main formal difference from ordinary quantum mechanics is the failure of Stone-von Neumann theorem, which states that, there is a unique (up to isomorphism) irreducible, unitary representation of the Weyl relations of the canonical commutation relations (or, equivalently, of the Heisenberg group) on a finite number of generators. In QFT, instead, there are an infinite number of unitarily inequivalent representations of the fundamental commutation relations. In Minkowski spacetime, Poincaré invariance selects a preferred representation built on a Poincaré invariant state, the vacuum. In general backgrounds one does not have such a criterion, and therefore the standard procedure, of quantizing a classical field solution of a suitable differential operator separating the positive and negative frequencies, is not possible. Moreover, even for free fields the operator $\phi(x)$ sharply evaluated at a point $x$ is not mathematically well-defined, since arbitrary high frequencies contribute and its expectation value is often divergent. The fields instead do make sense as distributions, smeared with continuous test functions with compact support. What is considered fundamental, then, are the commutation relations of the field observables, instead of a particular representation, and therefore we take as a first, basic principle the fact that i) the quantum fields are distributions valued in an algebra. Then, the principles and requirements of physical theories become conditions that the algebra must satisfy to describe a field theory. 

A fundamental condition of QFT in Minkowski spacetime is the positivity of the energy. As energy is defined using the property of invariance of field observables under time translations, this is no longer applicable in general backgrounds. However, a suitable generalisation has been formulated, called \textit{microlocal spectral condition}, which  is a local condition on the quantum field theory which corresponds to positivity of energy on Minkowski and is sensible in curved backgrounds. Although we will not deal with such a condition in the rest of the work, it is important to remark that such a condition on the energy exists. Therefore, we will assume that ii) the fields must satisfy suitable microlocal spectral conditions.

We can see a third principle as a requirement of compatibility with the principles of General Relativity. In Minkowski spacetime, one assumes that the fields are representations of the Poincaré group, that is, they transform covariantly under Poincaré transformations. The Poincaré group plays such a role because it is an isometry of the Minkowski metric, but it can be seen as the special relativistic case of the more general invariance of the metric under diffeomorphisms introduced in GR. What we require, then, is that the quantum fields are compatible with the background isometries. Therefore, in QFTCS one assumes that iii) the quantum fields must be locally and covariantly constructed.

Locality is achieved assigning an algebra $\mathcal A(\mathscr O) $ to each open region $\mathscr O $ of a spacetime $\mathscr M$. The union of the local algebras assigned to each subregion must generate an algebra over the full spacetime, $\bigcup_{\mathscr O} \mathcal A(\mathscr O) = \mathcal A(\mathscr M)$. Locality then gets translated in the two properties of \textit{isotony} and \textit{causality}. Isotony means that, for any pair of open regions such that $\mathscr O_1 \subset \mathscr O_2 $, then $\mathcal A(\mathscr O_1) \subset \mathcal A(\mathscr O_2)$. Causality means that, if two regions $\mathscr O_1, \mathscr O_2$ are space-like separated, the two algebras should commute: $[\mathcal A(\mathscr O_1), \mathcal A(\mathscr O_2)] = 0 $.

With these principles at hand a QFTCS can be constructed. As we said before, however, the standard approach followed in flat space, in which one selects a preferred vacuum state and then constructs particle states with ladder operators, cannot be applied here. In fact, although notions of vacuum states and particles can be defined on curved spacetimes, there is no natural way to select a preferred one. In general, the notion of the vacuum is observer-dependent, a key fact which we will show explicitly in the discussion of the Unruh effect in flat spacetime (section \ref{sec:Unruh-effect}) and of the thermal properties of the Unruh state in a Schwarzschild background (see section \ref{sec:Hawking-temperature}). Even if the spacetime of interest has suitable asymptotic properties in the past and in the future, so that one can define a natural notion of particle states and construct an S-matrix, many interesting questions regard the dynamical behaviour of the fields at finite times, where the particle notion becomes ambiguous. A path integral approach does not seem ideal neither, again because it is based on a preferred vacuum state for which the correlation functions are computed.

For the above reasons, the approach we will take here is the algebraic quantization, in which the commutation relations satisfied by the fields are taken as the fundamental property of the theory.
 
The algebraic approach follows a two-step procedure: first, one defines for a physical system a suitable $\ast$-algebra $\mathcal A$ (or a C*-algebra) which encodes the algebraic relations between observables, such as the canonical commutations relations, as well as the required properties of isotony, locality and covariance. This can be done explicitly for free fields. The second step is to select a so-called algebraic state $\omega$, a continuous, linear, positive functional on $\mathcal A$ which allows to define the expectation values of observables. Finally, the GNS reconstruction (see subsection \ref{ssec:GNS}) allows one to recover the usual interpretation of the algebra elements $A \in \mathcal A$ as linear operators over a Hilbert space $\mathcal H$.

We will describe the algebraic quantization in some detail for the case of the free, neutral, scalar field propagating in globally hyperbolic spacetimes, as it is the simplest case and will be enough to discuss the results in this work. The following sections, on the algebraic quantization follow closely \cite{AAQFT15} and \cite{HollandsWald14}. Then, we will describe the characterization of thermal states using the KMS conditions, the definition of relative entropy via the Araki formula, the connection between the entropy and the temperature (in other words, of the thermodynamic properties of the algebra itself, and of the states), and the explicit formula for the relative entropy between coherent states. For these sections, we refer to \cite{CasiniGrilloPontiello19}, \cite{CiolliLongoRuzzi19}, \cite{Haag92}, \cite{HollandsIshibashi19}, and \cite{Witten18}.

\section{Algebraic quantization for the free neutral scalar field} \label{sec:AQFT}
    \subsection{Classical field theory} \label{ssec:classical-field-theory}
We begin with the classical Klein-Gordon field, which is a solution of the Klein-Gordon (K-G) equation on a globally hyperbolic spacetime $(\mathscr M, g)$, with metric signature $(-,+,+,+)$,
\begin{equation}
    (\square - m^2 - \xi R) \psi = 0 \ ,
\end{equation}
where $\square = g^{\mu \nu} \nabla_\mu \nabla_\nu$ is the d'Alembertian operator.

The K-G equation is a well-posed initial value problem on the domain of dependence $J^+(\Sigma) \cap J^-(\Sigma)$ of some (partial) Cauchy surface $\Sigma$ if we give initial data $\eval{\psi}_\Sigma = f_0$, $n^\mu \nabla_\mu \eval{\psi}_{\Sigma} = f_1$, where $n$ is a unit vector normal to $\Sigma$.  Moreover, the solution does not change at $x$ if we change the initial data outside $J^-(x) \cap \Sigma$, thus preserving causality. We remember that $J^{\pm}(x)$, defined in section \ref{sec:globally-hyperbolic-spaces}, denotes the causal future/past  of $x$.

The Klein-Gordon operator $P = \square - m^2 - \xi R$ is a formally self-adjoint operator, meaning that, given two solutions of the K-G equation $\psi_1, \ \psi_2$,
\begin{equation}
\int_{\mathscr M} \psi_1 P \psi_2 \ \dd \text{vol}_{\mathscr M} = \int_{\mathscr M} P \psi_1 \psi_2 \ \dd \text{vol}_{\mathscr M} \ .
\end{equation}
Moreover, it is a Green hyperbolic operator, and thus admits unique retarded and advanced Green operators $E^{\pm}$, which are distributions that satisfy the Green equation $P E^\pm(x,y) = \delta(x,y)$, where the Klein-Gordon operator acts distributionally over the first variable, and are supported in the causal development of the support of $f$,  $\supp (E^\pm f) \subseteq J^\pm(\supp f) $. From this, one can show that $PE^\pm f = f $.

To construct the space of classical observables, one start with the space of test functions $ f \in C^\infty_0(\mathscr M)$. We then introduce the linear functionals $F_f : C^\infty(\mathscr M) \to \mathbb R$, which take any smooth function $\psi$ and give
\begin{equation}
F_f(\psi) = \int_{\mathscr M} f \psi \dd \text{vol}_{\mathscr M}  .
\end{equation}

These linear functionals label the \textit{off-shell configurations}, that is, the space of smooth functions $C^\infty(\mathscr M)$, since it is possible to find a test function $f$ such that $F_f(\psi_1) \neq F_f(\psi_2) $ for any two  $\psi_1, \psi_2 \in C^\infty(\mathscr M)$. However, if we consider solutions of the K-G equation only, the linear functionals $F_f$ are still enough to label the solutions (as the space of solutions is a restriction of the space of smooth functions), but the test functions $f$ are redundant, because if we add $Ph$, for some test function $h$, to any $f$, the two resulting linear functionals are the same: $F_f(\psi) = F_{f + Ph}(\psi) $ for any solution $\psi$ of the K-G equation.  Therefore, we define an equivalence class $[f]$, such that $f \sim g \Leftrightarrow F_f(\psi) = F_g(\psi) \ \forall \ \psi \in C^\infty(\mathscr M)$, that is, $f = g + Ph$ for some $h \in C^\infty_0(\mathscr M)$. Now, for every solution of the Klein-Gordon equation (the \textit{on-shell field configurations}), the linear functionals $F_{[f]}(\psi) $ label the solutions.  For this reason, the functionals $F_{[f]}$ are called the \textit{classical observables}. We denote the space of classical observables $\mathcal E$.

The classical field theory is completed introducing a symplectic structure on the space of the observables. The symplectic form is naturally induced by the equation using the advanced and retarded operator to introduce a solution of the Klein-Gordon equation
\begin{equation}
E = E^- - E^+
\end{equation}
which is called retarded-minus-advanced operator, or causal propagator, or commutator function, or Lichnérowicz function (and probably in some other way; we will call it causal propagator). The causal propagator, thanks to its distributional nature, can be seen as a bilinear, antisymmetric form over the space of test functions $C^\infty$. However, such a form would be a \textit{pre-symplectic form} because it would be degenerate. We introduce therefore a symplectic form $E : \mathcal E \times \mathcal E$, that, given any two representatives of the classes $[f], \ [g]$, is defined by
\begin{equation}
    E([f],[g]) := E(f,g) := \int_{\mathscr M} f E g \ \dd \text{vol}_{\mathscr M} \ .
\end{equation}

The classical structure outlined here can be linked to a different approach, more familiar from the physical viewpoint, in which one defines the symplectic structure over the space of continuous solutions of the Klein-Gordon equation with space-like compact support, $\mathsf{Sol}$, with the symplectic form:
\begin{equation} \label{eq:symplectic-form}
w(\phi, \psi) = \int_{\Sigma} \psi n^\mu \nabla_\mu \phi - \phi n^\mu \nabla_\mu \psi \ \dd \Sigma
\end{equation}
where $\Sigma$ is a partial Cauchy hypersurface and $n$ is a unit future-directed normal vector to the surface. We say that $\Sigma$ is a \textit{partial} Cauchy surface to include the case in which the QFT is not defined in the global spacetime, but only in the domain of dependence of $\Sigma$. It can be seen using the Klein-Gordon equation that the symplectic form does not depend on the choice of the Cauchy hypersurface, while the restriction to the compactly supported solutions ensures that the integral converges.

Although $\Sigma$ in general is taken to be a space-like hypersurface, the quantization procedure can be generalised to the limit in which $\Sigma$ becomes null. In this case, the normal vector  $n$ becomes tangent to the hypersurface, while we can adopt the coordinates $y^a = (\lambda, \theta^A)$ we introduced in subsection \ref{ssec:hypersurfaces}, where $\lambda$ is the affine parameter along the hypersurface. Then, $n = \partial_\lambda$, and the symplectic form becomes

\begin{equation}
w(\psi_1, \psi_2) = \int_\Sigma \psi_2 \partial_\lambda \psi_1 - \psi_1 \partial_\lambda \psi_2 \ \sqrt \sigma \dd \lambda \dd^2 \theta
\end{equation}

It is possible to prove the equivalence between the two approaches based on $(\mathsf{Sol},w)$ and $(\mathcal E, E)$ showing that there is a symplectomorphism (an isomorphism which maps the two symplectic forms into each other) between the two symplectic spaces,
\begin{equation}
I : [f] \mapsto E f
\end{equation}
where $f$ is any representative of $[f]$. Here, we only show that $E([f],[g]) = w(E f,Eg)$. Recalling that $\psi_f = Ef, \psi_g = Eg$ are solutions of the Klein-Gordon equation, we have
\begin{equation} \label{eq:symplectic-form-causal-propagator-relation}
    \begin{split}
E([f],[g]) &= 
\int_{\mathscr M} f E g \ \dd \text{vol}_{\mathscr M} = 
\int_{J^+(\Sigma)} f \psi_g \ \dd \text{vol}_{\mathscr M} + \int_{J^-(\Sigma)} f \psi_g \ \dd \text{vol}_{\mathscr M} =  \\
&= \int_{J^+(\Sigma)} (PE^-f) \psi_g \ \dd \text{vol}_{\mathscr M} + 
\int_{J^-(\Sigma)} (PE^+ f) \psi_g \ \dd \text{vol}_{\mathscr M} = \\
&= \int_{J^+(\Sigma)} P(E^-f \psi_g) \ \dd \text{vol}_{\mathscr M} + 
\int_{J^-(\Sigma)} P(E^+ f \psi_g) \ \dd \text{vol}_{\mathscr M} = \\
&= \int_\Sigma \psi_g \nabla_n(E^- f) + E^-f \nabla_n \psi_g \ \dd \Sigma -
\int_\Sigma \psi_g \nabla_n (E^+f) + E^+f \nabla_n \psi_g \ \dd \Sigma \\
&= \int_\Sigma \psi_g \nabla_n(E^- f) + \psi_g \nabla_n(E^+f) \ \dd \Sigma - \int_\Sigma \psi_g \nabla_n(E^+ f) + \psi_g \nabla_n(E^-f) \ \dd \Sigma = \\
&= w(\psi_f,\psi_g)
    \end{split}
\end{equation}

In the first step we separated the integral over the volume in the past and future domains of dependence of a Cauchy hypersurface $\Sigma$; in the second we used the property $P E^\pm f = f$; in the third we used the fact that, $P \psi_g = 0$. In the fourth we used Gauss-Stokes theorem to convert the bulk integrals in surfaces integrals; the minus sign in the second integral appears because the normal vector is future-directed, and we use the Leibniz rule to split the action of $\nabla_n$ on $\psi_g$ and $E^\pm f$. In the last passage we integrated by parts the second and fourth term, as the boundary terms vanishes because are evaluated at infinity, and $\phi,\psi$ have compact support, to obtain the causal propagator $(E^- - E^+)f = Ef = \psi_f $.

The approach based on test functions is more useful when one is interested in the local behaviour of some observable, because in that case the test functions with compact support define the observables in some finite region. The approach based on the classical solutions is more useful when is interested in the actual propagation of the wave: in this case, one gives initial data for the wave on some Cauchy surface, $\Sigma$, and then study the propagation of the field into the spacetime. In chapter \ref{ch:new-results}, when we will apply this formalism to the case of a scalar field over a black hole spacetime, we will consider the latter approach.

The symplectic form $w$ is causal, as $E$, in the sense that the symplectic pairing of classical solutions localized in causally disjoint regions vanishes. The symplectic structure of classical observables is the starting point for the procedure of quantization. 

    \subsection{Algebra of Linear Observables} \label{ssec:algebra-observables}

Now we construct the algebra of observables $\mathcal A(\mathscr M)$ for the quantum theory. We will not consider here the enlarged algebra generated by \textit{composite fields}, that is, the algebra generated by polynomials of the fields.

An algebra is a complex vector space equipped with a product that is associative and distributional with respect to the sum of vectors, and also is associative with respect to the multiplication by complex scalars. A $\ast$-algebra is an algebra, together with am involutive map $\ast : \mathcal A \ni A \mapsto A^* $ with the following properties:
\begin{align}
(A^*)^* &= A \\
(A + B)^* &= A^* + B^* \quad \forall B \in \mathcal A\\
(\lambda A)^* &= \bar \lambda A^* \quad \forall \lambda \in \mathbb C \\
\end{align}
where the overbar denotes conjugation of a scalar.

An algebra can be defined giving the set of \textit{generators} of the algebra and their \textit{relations}. An algebra is generated by a set of elements if any $A \in \mathcal A$ can be written as a finite complex linear combinations of products of the generators.

 As we said, the field $\phi$ is a distribution smeared by a test function $f \in C^\infty_0(\mathscr M)$. We define the smeared field $\phi(f)$ via the pairing
\begin{equation}
 \phi(f) = \langle \phi, f \rangle = \int_{\mathcal M} f \phi \dd \text{vol}_{\mathscr M} \ .
\end{equation}   

The \textit{algebra of linear observables} is a $\ast$-algebra generated by the symbols $\mathbf 1$ and $\phi(f)$ and by the following relations, satisfied for all $f,g \in \mathcal E$ and for all $\alpha,\beta \in \mathbb R$:
\begin{align}
\textbf{Linearity}&: \phi(\alpha f + \beta g) = \alpha \phi(f) + \beta \phi(g) \\
\textbf{Klein-Gordon Equation}&: \phi((\square - m^2 - \xi R)f) = 0 \\
\textbf{Hermiticity}&: \phi(f)^* = \phi(\bar{f}) \\
\textbf{Canonical Commutation Relations (CCR)}&: [\phi(f), \phi(g)] = i E(f,g) \mathbf 1 = i w(\psi_f,\psi_g) \mathbf{1}
\end{align}

This algebra should be considered "abstract", in the same sense in which a Lie algebra is defined with no reference to any particular representation.

Although different approaches exist, here we consider a familiar approach in Physics, in which one finds a classical solution for a relativistic wave equation $\phi(f)$, and then one quantizes it with suitable commutation relations.

This algebra does not contain all the physically interesting observables, since, for example, there is no element in $\mathcal A$ which can be identified as the stress-energy tensor of the field, nor there are point-wise interactions as $\phi^2(x)$. However, this algebra is sufficient for our purposes.

The algebra defined this way is not a C*-algebra. A C*-algebra is a $\ast$-algebra complete with respect to a norm with the properties that $\norm{AB} \leq \norm{A} \norm{B}$, and $\norm{A^*A} = \norm{A}^2$. This is not always a necessary requirement, but in some situations it may prove useful. To construct a C*-algebra for the Klein-Gordon field one uses the algebra generated by $\mathbf 1$ and the symbol $W(f)$, together with the \textit{Weyl relations}
\begin{align} \label{eq:Weyl-relations}
W(0) &= \mathbf 1 \\
W(f)^* &= W(-f) = W(f)^{-1} \\
W(f)W(g) &= e^{- \frac{i}{2} w(\psi_f, \psi_g)}W(f + g)
\end{align} 
where $\psi_f, \ \psi_g $ are compactly supported, classical solutions of the Klein-Gordon equation, $\psi_f = Ef $ and $\psi_g = Eg$. The Weyl operators are formally related to the $\ast$-algebra operators as exponentials,
\begin{equation} \label{eq:Weyl-operators-classical-solutions}
W(f) = e^{i \phi(f)} \ .
\end{equation}

Two important properties of the quantum theory are
\begin{enumerate}
    \item \textbf{Causality}: elements of the algebra localized in causally disjoint regions commute, i.e., if $f,g \in C_0^\infty(\mathscr M)$ are such that $supp(f) \cap J^\pm(supp(g)) = \emptyset \Rightarrow [\phi([f]),\phi([g])] = 0$.
    \item \textbf{Time-slice axiom} If $\Sigma$ is a partial Cauchy surface, let $\mathcal O$ be any fixed open neighbourhood of $\Sigma$. Then, the algebra $\mathcal A$ of the full space is generated by $\mathbf 1$ and by $\phi(f)$, where $suppf \subseteq \mathcal O$. This axiom (which is actually a proven theorem in the case of free fields) lets one introduce a dynamical structure in the theory, and in particular means that a measurement in a finite time interval suffices to predict all other observables.
\end{enumerate}

\subsection{States} \label{ssec:states}
The algebra of observables encodes the physical requirements of locality, causality and covariance, together with the quantum structure of the commutation relations (in the form of CCR or the Weyl relations). The second step of the quantization procedure is the definition of a state. This enables one to evaluate the expectation values of observables, and thus to confront the theory with the experiment. This way, the abstract framework becomes more concrete, and it is possible to understand the physical interpretation of the theory.

In the algebraic approach, a state is a linear, positive, normalized, and continuous functional $\omega: \mathcal A \to \mathbb C$. Positivity means $\omega(AA^*) \geq 0 \ \forall A \in \mathcal A$, while normalization requires $\omega(\mathbf 1) = 1$. Continuity may be a tricky issue, but for C*-algebras it is ensured by positivity, and this is one of the reasons why C*-algebras are preferable from a mathematical point of view.

Since a generic element of the algebra is in the form
\begin{equation}
A = c_0\mathbf 1 + \sum_{i_1} c_1^{i_1}\phi(f^1_{i_1}) + \sum_{i_1,i_2} c_2^{i_1.i_2}\phi(f^2_{i_1})\phi(f^2_{i_2}) + ... + \sum_{i_1,...,i_n}c^{i_1...i_n}_{n}\phi(f^n_{i_1})...\phi(f^n_{i_n}) \ ,
\end{equation}
a state is specified by the collection $(W_n)_{n>0}$ of its \textit{n-point functions},
\begin{equation}
W_n (f_1,...f_n) = \omega(\phi(f_1),...,\phi(f_n)) \ ,
\end{equation}
which are seen as distributions over $C^\infty_0(\mathscr M^n) $. Therefore, we can understand the continuity requirement of the state as the continuity of the n-point functions in the topology of $\ \mathbb C_0^\infty(\mathscr M)$. In the following, we will often denote a n-point function by its distributional kernel,
\begin{equation}
W_n(f_1,...,f_n) = \int_{\mathscr M} W_n(x_1,...,x_n)f(x_1)...f(x_n) \ \dd \text{vol}_{\mathscr M} \ .
\end{equation}

The Klein-Gordon equation and the CCR conditions on the algebra become conditions on the n-point functions, which in the simplest case $n=2$ are:
 \begin{align}
 \omega(\phi(Pf)\phi(g)) = W_2(Pf, g) = W_2(\phi(f),\phi(Pg)) = 0 \\
 W_2(f, g) - W_2(g,f) = i w(\psi_f,\psi_g) = i E(f,g) \ .
 \end{align}
Since the symplectic form is real, from the second condition one immediately gets
\begin{equation}
\Im(W_2(f,g)) = \frac{1}{2}w(\psi_f,\psi_g) \ .
\end{equation}
The positivity requirement on the state becomes a complicated hierarchy of conditions on the n-point functions, the simplest of which is
\begin{equation}
W_2(f, \bar f) \geq 0 \ .
\end{equation}

The mathematical definition of a state encompasses infinitely many states which are not physically sensible, because their behaviour for some observables of physical interest is too singular. Thus, we need to select some regularity conditions to single out unphysical states and determines which ones are suitable for a physical interpretation of the theory. A natural criterion is the compatibility of the state with the background isometries. Under mild hypothesis, for globally hyperbolic spacetimes with maximal symmetry this requirement admits a unique vacuum state, so that, for instance, in Minkowski spacetime, Poincaré invariance selects the familiar vacuum state, while in de Sitter spacetime one uses the Bunch-Davies vacuum. If the spacetime is not maximally symmetric, uniqueness may fail, so in more general backgrounds one needs some further regularity condition to single out (if possible) the most physically sensible state. The Schwarzschild spacetime, for example, admits different states which can seem reasonable, as the Boulware, Unruh, and Hartle-Hawking states; only physical considerations on the regular behaviour of observables permit to rule out the Boulware vacuum, while the choice between the Unruh state or the Hartle-Hawking state depends on the physical conditions of the system under consideration (the Unruh state is the most suitable for describing the formation and evaporation of black holes, while the Hartle-Hawking state describes an eternal black hole in thermal equilibrium with its exterior). Before discussing the regularity properties of the states we see how one can reconstruct, from the abstract algebra of observables and the choice of a state, the usual notions of QFT as linear operators over a Hilbert space.

\subsection{GNS Reconstruction} \label{ssec:GNS}
  The Gel'fand, Naimark, Segal (GNS) theorem lets one recover the usual representation of a QFT over a Hilbert space in terms of linear operators.
  
  A \textit{representation} of an algebra, denoted $(\mathcal H, \pi, \mathcal D)$ is a Hilbert space $\mathcal H$, together with a dense subspace $\mathcal D$ and a map from $\mathcal A$ to the algebra of closable operators over $\mathcal D$, $\pi : \mathcal A \to \mathcal \pi(\mathcal A)$, which satisfies the usual properties of a representation, for $A,B \in \mathcal A$ and $\alpha, \beta \in \mathbb C$:
  \begin{align}
  \pi(\mathbf 1) &= 1 \\
  \pi(A) \pi(B) &= \pi(AB) \\
  \pi(\alpha A + \beta B) &= \alpha \pi (A) + \beta \pi (B) \\
  \pi(A^*) &= \pi(A)^*
\end{align} 

Any state $\omega$ on $\mathcal A$ determines a unique (up to unitary equivalence) representation, called the GNS representation, or Wightman reconstruction argument, denoted by the quadruple $(\mathcal H_\omega, \pi_\omega, \mathcal D_\omega, \ket{\Omega_\omega})$. $\ket{\Omega_\omega}$ is a unit vector in $\mathcal D$ such that
\begin{equation}
\omega(A) = \mel{\Omega_\omega}{\pi_\omega(A)}{\Omega_\omega}
\end{equation} 
for any $A \in \mathcal A$, and $\pi_\omega(\mathcal A) \ket{\Omega_\omega} = \mathcal D_\omega$, implying that $\ket{\Omega_\omega}$ is cyclic for the representation. Cyclic means that, given a vector $\ket \Omega \in \mathcal H$ and an operator $\pi(A)$ over $\mathcal H$, $ \pi(A) \ket \Omega$ is dense in $ \mathcal H $.

The GNS theorem shows the correspondence between states and Hilbert space representations by explicit construction. In the case of C*-algebras, $\mathcal D_\omega = \mathcal H_\omega$ and this is why $(\mathcal H_\omega, \pi_\omega, \ket{\Omega_\omega)}$ is called the GNS triple (and this is another reason why C*-algebras are preferred).

In this representation, the smeared fields are mapped into the \textit{smeared field operators}, which are symmetric densely defined operators $\hat \phi(f) = \pi_\omega(\phi(f))$ over $\mathcal D_\omega$. We will consider \textit{regular states} only, for which the smeared field operators are essentially self-adjoint.

The induced von Neumann algebra over $\mathcal H_\omega$ is $\mathfrak U(\mathcal H) = \pi(\mathcal A)''$, where the double prime denotes the double commutant and coincides with the weak closure of $\pi(\mathcal A)$. A von Neumann algebra is $\ast$-algebra of bounded operators on a Hilbert space $\mathcal H$, closed in the weak operator topology, and containing the identity operator.

This construction shows that the algebraic approach encompass the Hilbert space formulation of QFTCS, with one fundamental generalisation: in fact, the representations of the algebra $\mathcal A$, in general, are not unitarily equivalent. This means that choosing a state (or, equivalently the cyclic vector $\ket{\Omega_\omega}$) is a highly non-trivial step, leading to different physical interpretations of the same algebraic structure. This is why the algebraic approach is preferable, as it permits to highlight the properties of the theory without referring to any specific representation. The choice of a state must be motivated in a second step, on physical grounds.

As a brief aside we comment here that the positivity requirement is motivated by physical assumptions, but there are situations, as in the Gupta-Bleuler, or the BRST quantization of gauge theories, in which states of non-definite norm must be included, called ghost fields. In these situations physically meaningful states can fail the positivity requirement to include the algebra generated by the ghost fields. In this case, the GNS reconstruction fails, as the positivity requirement is a crucial assumption, and in any case the representation should be over a space with indefinite norm instead of a Hilbert space; however, there are several generalisations of the GNS theorems which replace the positivity requirement and represents the theory over the correct space. 

\subsection{Quasifree States and One-Particle Structure} \label{ssec:quasifree-states}
The first class of states we consider are the \textit{Quasifree} or \textit{Gaussian} states. These states are completely determined by their 2-point functions. In particular, a quasifree state is defined by the two conditions:
\begin{enumerate}
    \item for all $n$ odd, $W_n = 0$;
    \item for all $n$ even, $W_n(f_1,...,f_n) = \sum_{partitions}W_2(f_{i_1},f_{i_2})...W_2(f_{i_{n-1,}, f_{i_n}})$.
 \end{enumerate}
 The "sum over partitions" is intended as the sum over all possible $n/2$ disjoint pairs $\{i_{k-1},i_k \}$ with $k = 1,...,n$.
 
 Quasifree states always exist in globally hyperbolic spacetimes, an important theorem which we will not prove. However, quasifree states admits a representation in terms of a Fock space. To see this in some detail, we first introduce another structure. Consider the symplectic space $(\mathsf{Sol},w)$, and consider a scalar product $\mu : \mathsf{Sol} \times \mathsf{Sol} \to \mathbb C$ such that
 \begin{equation} \label{ineq:1-particle-structure}
\frac{1}{4} w(\psi_f,\psi_g)^2 \leq  \mu(\psi_f,\psi_g) \mu(\psi_g,\psi_g)
 \end{equation}
In this case, we have an associated \textit{one-particle structure} $(K, \mathcal H_1)$, where $\mathcal{H}_1$ is a complex Hilbert space and $K: \mathsf{Sol} \to \mathcal H_1$ a map such that
\begin{enumerate}
        \item K is $\mathbb R$-linear and $K(\mathsf{Sol}) + i K(\mathsf{Sol})$ is dense in $\mathcal H_1$;
    \item $\braket{K \psi_f}{K \psi_g} = \mu(\psi_f, \psi_g) + \frac{i}{2} w(\psi_f, \psi_g)$ \label{eq:internal-product} 
\end{enumerate}
Conversely, if $(K, \mathcal H_1)$ satisfies the first property and $w(\psi_f,\psi_g) = 2 \Im\braket{K\psi_f}{K \psi_g}$, the scalar product $\mu(\phi, \psi) = \braket{K \phi}{K \psi} -  \frac{i}{2} w(\phi, \psi) $ satisfies the inequality \eqref{ineq:1-particle-structure}.

Using the one-particle structure we can characterise a quasifree state with its GNS representation. Although we use a characterisation in terms of $\ast$-algebras, the C*-algebra approach is equivalent for quasifree states. 

If we consider the $\ast$-algebra of linear observables, together with the scalar product $\mu$ on $\mathsf{Sol}$ defined above, then there exists a quasifree state whose 2-point function satisfies
\begin{equation} \label{eq:2PF-Quasifree-1-particle-structure}
W_2(f,g) = \mu(Ef,Eg) + \frac{i}{2}E(f,g)
\end{equation}
so that the internal product \eqref{eq:internal-product} is the map of the 2-point function into the one-particle Hilbert space.
Moreover, its GNS triple $(\mathcal H_\omega, \pi_\omega, \ket{\Omega_\omega})$ is characterised by the following properties:
\begin{enumerate}
    \item $\mathcal H_\omega$ is a symmetrized Fock space with the one-particle subspace given by $\mathcal H_1$;
    \item $\ket{\Omega_\omega}$ is the vacuum vector of the Fock space;
\end{enumerate} 
 
In this sense, quasifree states are the natural generalisation of the Minkowski vacuum state to general backgrounds, as they always admit a representation as the vacuum vector of a Fock space.

The quasifree state can be equivalently characterised by its action on the Weyl operator, $W(f)$. This can be computed by Taylor expanding the exponential,
\begin{equation}
\omega(W(f)) = \sum_{n=0}^\infty \frac{i^n}{n!} \omega(\phi(f)^n) \ .
\end{equation}

For n odd, the term in the series vanishes, while for n even, the action of the state on the product of $2n$ terms is given by the product of $n$ 2-point functions summed over all the possible partitions:
\begin{equation}
\omega(W(f)) = \sum_{n=0}^\infty -\frac{1}{(2n)!} \sum_{partitions}W_2(f,f)^{n} \ ,
\end{equation}
where we have redefined $n$. This means that each term is actually the product of $2n$ terms. How many partitions are there? The number of ways we can choose the first pair is given by the binomial coefficient, $\binom{2n}{2}$. The second pair can be chosen out of $2n-2$ elements, and so on. The total number of ways we can choose the pairs is
\begin{equation}
\binom{2n}{2} \binom{2n-2}{2}... \binom{2}{2}  = \frac{(2n)!}{2!(2n-2)!} \frac{(2n-2)!}{2!(2n-4)!} \frac{(2n-4)!}{2!(2n-6)!}...= \frac{(2n)!}{2^n} \ .
\end{equation}
However, this is actually more than the true number of different partitions, because we are counting as different partitions the configuration with the same pairs, in different order. Thus, we must divide by the number of ways one can arrange the $n$ pairs, that is, $n!$. The total number of partitions is therefore
\begin{equation}
\frac{(2n)!}{n!2^n} \ .
\end{equation}

The quasifree state functional over the Weyl operator then is
\begin{equation}
\omega(W(f)) = \sum_{n=0}^\infty - \frac{1}{n!} \frac{W_2(f,f)^n}{2^n} = e^{-\frac{1}{2}W_2(f,f)} \ .
\end{equation}

Using the one-particle structure, this can be written in two other equivalent ways, as the symplectic form is antisymmetric:
\begin{equation}
\omega(W(f)) = e^{- \frac{1}{2}\norm{K \psi_f}^2} = e^{-\frac{1}{2} \mu(\psi_f,\psi_f)} \ ,
\end{equation}
where the norm is taken with respect to the internal product of the one-particle Hilbert space.

Now we can finally recover the usual mode-decomposition quantization. In fact, we can introduce the usual creation and annihilation operators $a(\psi), a^*(\psi)$ on the Fock space, corresponding to the classical solution $\psi$, that satisfy the commutation relations $[a(\psi_1), a^*(\psi_2)] = \braket{K\psi_1}{K\psi_2}$, $[a(\psi_1), a(\psi_2)] = [a^*(\psi_1), a^*(\psi_2)] = 0$. The representation $\pi_\omega$ is then completely determined by the ladder operators, with
 \begin{equation}
 \hat \phi(f) = \pi_\omega(\phi(f)) = a(Ef) + a^*(Ef) \ .
 \end{equation}
 
In Minkowski spacetime, if we consider a set of smooth, complex-valued, "mode functions" $\psi_{\mathbf k}$, classical solutions of the Klein-Gordon equation, then we can construct $K$ for all $f \in C^\infty_0(\mathscr M)$ as the map
\begin{equation}
X \ni \mathbf k \mapsto K \psi_f(\mathbf k) = \int_{\mathscr M} \bar{\psi}_{\mathbf k}(x) f(x) \dd \text{vol}_{\mathscr M}
\end{equation} 
This is a map into $L^2(X)$. We can also assume that $\Im\braket{K \psi_f}{K \psi_g} = \frac{1}{2}E(f, g)$, so that
\begin{equation}
 W_2(f, g) = \braket{K \psi_f}{K \psi_g} = \int_{\mathscr M} \int_{\mathscr M} \int_X  \psi_{\mathbf k}(x) \bar{\psi}_{\mathbf k}(y) \frac{\dd^3 \mathbf k}{(2 \pi)^{\frac{3}{2}}\sqrt{\omega_{\mathbf k}}} \ f(x) g(y) \ \dd \text{vol}_{\mathscr M} \dd \text{vol}_{\mathscr M}
\end{equation}
defines the 2-point function of a Gaussian state. From it one can find $\mu$ and reconstruct the one-particle structure. As $K$ is a map into $L^2(X)$, the Hilbert space $\mathcal H$ of the GNS reconstruction is the bosonic Fock space over $L^2(X)$, with the classical mode-function solutions
\begin{equation}
\psi_{\mathbf k}(x) = \frac{1}{\sqrt{2 \omega_{\mathbf k}}} e^{-i(\omega_{\mathbf k}t-\mathbf{k}\cdot \mathbf x)}
\end{equation} 
 
The quantum field operator representative over the Fock space can be informally written with as
\begin{equation}
\hat \phi = \int_0^\infty  \frac{1}{(2 \pi)^{\frac{3}{2}}\sqrt{2 \omega_{\mathbf k}}} \bigg( a_{\mathbf k} e^{-i(\omega_{\mathbf k}t - \mathbf{k}\cdot \mathbf x )} + a^*_{\mathbf k} e^{i(\omega_{\mathbf k}t - \mathbf{k}\cdot \mathbf x )} \bigg) \dd^3 \mathbf{k}
\end{equation} 
with $a(\phi) = \int a_{\mathbf k} \psi_{\mathbf k}(x) \dd \mathbf k $.
Such a construction is always possible when the spacetime has a globally defined time-like isometry, as one can take the mode functions to have positive frequency with respect to global time translations.

As a final remark here we note that the scalar product $\mu$ remains arbitrary in the construction. As $\mu$ is in one-to-one correspondence with the Hilbert space construction, this in turn means that the particle content of the theory is arbitrary, an explicit realisation of the theoretical considerations made in section \ref{sec:motivations}.

\subsection{Coherent perturbations} \label{ssec:coherent-perturbations}
A coherent state is constructed starting from a cyclic and separating vector $\ket \Omega$, the GNS-representative of a state functional $\omega(A) = \ev{A}{\Omega}$ for all $A \in \mathcal A$. As we said in subsection \ref{ssec:GNS}, cyclic means that $U \ket \Omega$ for $U \in \mathfrak U$ is dense in $ \mathcal H $, where $ \mathfrak U(\mathcal H) $ is the induced von Neumann algebra over $\mathcal H$, while separating means that $U \ket \Omega = 0 \Rightarrow U = 0$.  For some test function $f$, $W(f)$ is the corresponding Weyl operator, and a coherent state of $\ket \Omega$ is a vector $W(f) \ket \Omega$ with state functional $\omega_f(A) = \omega(W(f)^*AW(f)) = \ev{W(f)^*A W(f)}{ \Omega}$. $ \psi_f = Ef$ as before is the classical solution associated with the test function. The quantum field operator $\hat \phi$ is related to the Weyl operator by \eqref{eq:Weyl-operators-classical-solutions},
\begin{equation} \label{eq:Weyl-operators-fields}
W(f) = e^{i \hat \phi(f)} = e^{ia(Ef) + ia^*(Ef)}
\end{equation}
Using the Weyl relations, then, one can show that
\begin{equation}
W(f) \ \hat \phi \ W(f)^* = \hat \phi + W(f) [\hat \phi, W(f)^*] = \hat \phi + \psi_f \mathbf 1
\end{equation}
and so $\omega_f(\hat \phi) = \omega(W(f) \hat \phi W(f)^*) = \psi_f$, thus justifying the name coherent state.

\subsection{Hadamard States} \label{ssec:hadamard-states}

As we saw, quasifree states are a large class of states generalising the Minkowski vacuum. However, since there is no regularity condition on their 2-point function, this class is still much larger than the class of physically sensible states. This additional requirement selects the so-called \textit{Hadamard states}. The conditions one requires for a state to be physically acceptable are that i) quantum fluctuations of all observables are finite; ii) the ultraviolet divergence of the 2-point function is of the same form as in the Minkowski vacuum. Under these assumptions, one can show that the 2-point function of a quasifree state for on-shell configurations of the fields must be of the Hadamard form:
\begin{equation}
W_2 (x,y) = \lim_{\epsilon \to 0} \frac{1}{8 \pi^2} \bigg( \frac{u(x,y)}{\sigma_\epsilon(x,y)}+ v(x,y) \log{\frac{\sigma_\epsilon(x,y)}{\lambda^2}} + z(x,y) \bigg)
\end{equation} 
where $\sigma_\epsilon(x,y)$ is the squared geodesic distance and $u,v,z$ are smooth functions, with $u,v$ being completely determined in terms of the metric and parameters of the equation of motion. The freedom in the choice of the quasifree state relies on $\lambda$, a free reference length and on $z$, an arbitrary symmetric function constrained by the positivity requirement and the equation of motion.

\section{Thermal Properties} \label{sec:thermal-properties}

     \subsection{KMS States} \label{ssec:KMS}

In the previous section we focussed on the general properties of any free bosonic algebra, and on the properties of those states most similar to the vacuum in Minkowski. Now we want to characterize equilibrium states at finite temperature in QFT. Even in Minkowski, these states cannot be described by a density matrix, as is done in non-relativistic Quantum Mechanics (QM), since the infinite number of degrees of freedom makes the operator $e^{-\beta H}$ unbounded and its trace $ Z = \Tr e^{-\beta H} $ diverges. Instead, as in the generalisation from flat to curved spacetimes, we need to find those fundamental properties which, on one hand, uniquely characterize the equilibrium states, and, on the other, are suitable to a generalisation to the continuum limit. In this context the properties we are interested in must be preserved by the thermodynamic limit $N,V,E \to \infty$ (the particle number, the volume, and the energy of the system) with $N/V,E/V$ finite. This limit is necessary in ordinary statistical mechanics to give an unambiguous meaning to concepts like temperature and to explain phase transitions. In that context, one usually first study the system in a finite volume, inside a box with some boundary conditions, to make it finite in size, and eventually remove the boundary and investigates the limit of a box covering the whole space. If one uses \textit{intensive quantities}, which do not depend on the volume in which they are evaluated, one expects that they remain meaningful in the thermodynamic limit. In particular, one can use the thermodynamic limit of the 2-point function to define a quasifree state in thermal equilibrium at temperature $T$. In general, this state will not be represented by a density matrix in the Fock space, but as some other vector in a Hilbert space in the GNS reconstruction.

Let's start then by an analysis of ordinary QM at finite temperature. In the \textit{canonical ensemble}, denoted with $H$ the Hamiltonian of the system and $\beta = T^{-1}$ the inverse temperature, the density matrix and the partition function are
\begin{equation}
\rho = Z^{-1} e^{- \beta H }, \quad Z = \Tr e^{- \beta H} \ .
\end{equation}
The density matrix can be seen as the representation of a state over the Fock space of the algebra of observables, with the expectation value given by 
\begin{equation} \label{eq:Gibbs-state}
\omega_\beta (A) = \Tr\rho A , \quad \forall A \in \mathcal A \ ,
\end{equation}
where as is customary in physics we denote with the same symbol the algebra and its elements $\mathcal A, \ A$, and their representation as operators over a Hilbert space, $ \pi_\omega(A), \ \mathfrak U$.
The time evolution of observables is given in the Heisenberg picture by
\begin{equation} \label{eq:time-translation-group}
\alpha_t A = e^{iHt} A e^{-iHt} \ .
\end{equation}
Due to the invariance of the trace under cyclic permutations, we can find the properties satisfied by the Gibbs state:
\begin{multline}\label{eq:property1-Gibbs-state}
\omega_\beta(\alpha_t A B) = Z^{-1} \Tr e^{-\beta H} e^{iHt} A e^{-iHt} B = Z^{-1} \Tr e^{-\beta H} B e^{iH(t+i\beta)} A e^{-iH(t+i\beta)} = \\
= \omega_\beta(B \alpha_{t+i\beta}A)
\end{multline}
where now the time evolution $ \alpha_z $ is extended to complex values, $z = t + i \gamma$. We see that the function of $z$ 
\begin{equation}
F_{A,B}(z) = \omega_\beta(B \alpha_z A) = Z^{-1} \Tr e^{-\beta H} B e^{iHt} e^{-\gamma H} A e^{-iHt} e^{\gamma H}
\end{equation} is analytic in the strip
\begin{equation}
0 < \Im z < \beta
\end{equation}
because the term $e^{-\beta H}$ provides "enough dumping" to make the trace finite in this range. On the boundary of the strip, we have the relation \eqref{eq:property1-Gibbs-state}
\begin{equation} \label{eq:twisted-time-periodicity}
F_{A,B}(t+i\beta) = G_{A,B}(t)
\end{equation}
where $G_{A,B}(t) = \omega_\beta(\alpha_t(A) B) $, a sort of "twisted" periodicity in the imaginary time. The two relations \eqref{eq:property1-Gibbs-state}, \eqref{eq:twisted-time-periodicity} are stable under thermodynamic limit, and so they represent exactly the properties needed to generalise the thermal structure to QFT. In particular, they still hold for a thermal state even when the density matrix representation which ultimately motivate them does not exist.

The first to note the analyticity property was Kubo \cite{Kubo57}, and it was used to define the thermal Green functions by Martin and Schwinger \cite{Martin59}. The adaptation to the algebraic approach due to Haag, Hugenholtz and Winnick \cite{Haag67} lead to the definition of a \textit{KMS State} which does not refer to any particular representation, nor to the time evolution itself, but can be formulated as a condition on the expectation values of a state as follows. Consider a 1-parameter automorphism of the C*-algebra of observables $\alpha_t : \mathcal A \to \mathcal A$. A state $\omega_\beta : \mathcal A \to \mathbb C $ is a \textit{KMS state} at inverse temperature $\beta$ with respect to $\alpha_t$ if it satisfies the following two conditions:
\begin{enumerate}
    \item For $A,B \in \mathcal A$, the function $F_{A,B}: \mathbb C \ni z \mapsto \omega_\beta(B \alpha_z A)$ is analytic in the strip $0 < \Im z < \beta$ and is bounded and continuous at the boundary;
    \item On the boundary, $F_{A,B}(t+i\beta) = G_{A,B}(t) $
\end{enumerate}

In the case of C*-algebras generated by Weyl relations, the conditions for the 2-point function implies analogue conditions for the more general n-point functions. The second condition can be reformulated in terms of the (integral kernel of) the 2-point function as
\begin{equation}
W_2(x,y + i \beta e^0) = W_2(y,x) \ ,
\end{equation}
where $e^0$ is a unit vector in the direction of the translation generated by $\alpha_t$. Although, in general, a KMS state is not quasifree, in the cases we will consider it will always be, and therefore the condition on the 2-point function uniquely determines the state.

The definition assumes only the existence of an algebra $\mathcal A$ with a 1-parameter family of automorphisms $\alpha_t$, making no references to the representation nor to the automorphism itself, as promised. Moreover, it makes no reference to a box of volume $V$ in which the system is confined, to boundaries, or to the thermodynamic limit. Thus, it is applicable to a broad range of quantum systems, even relativistic, so that it is suitable for a description of the thermal properties of a quantum scalar field in curved spacetimes.

In the case in which the spacetime at hand has a global Killing symmetry, this implements an isometry of the metric and therefore an automorphism of the algebra, leading to a natural notion of thermal state on the whole spacetime as a KMS state with respect to this global automorphism; however, the general definition reinforces the notion that thermal or vacuum states depend on the automorphism, and therefore on the isometry of the spacetime. Once again, we see that different observers, which move differently with respect to some isometry of the spacetime, can detect a different particle content. This phenomenon is at the heart of the Unruh effect.

Now, we can check that the KMS conditions are sufficient to single out the thermal states. As we did for the algebraic quantization, then, we see how to recover the usual notion of Gibbs state for ordinary QM starting from the algebraic definitions. If we consider a self-adjoint Hamiltonian for which $ \Tr e^{-\beta H} < 0$, then it admits a representation in terms of a density matrix, that is, a state $\omega_\beta(A) = \Tr \rho A $. The automorphism $\alpha_z$ can be defined by the translations generated by the Hamiltonian $H$. It is then implemented by a set of unitaries $U(z) = e^{-H z} \in \mathfrak U$. If we take an invariant element of the algebra $A$ under $\alpha_z$, then $F_{A,B}(z) = \omega_\beta(B\alpha_z A)$ is independent of $z$, and the KMS condition says that $F_{A,B} = G_{B,A}$. This implies that
\begin{equation}
0 = F_{A,B} - F_{B,A} = Z^{-1}(\Tr\rho BA - \Tr \rho AB) ) = Z^{-1} \Tr[A,\rho]B
\end{equation}
If $A$ commutes with $U$, and $\rho$ commutes with $A$, then $\rho \propto U(\zeta) = e^{iH\zeta}$ (we use $\zeta$ to distinguish from the action of $\alpha_z$ in $F_{A,B}(z)$, which we use in a moment). The second KMS condition at $z = i \beta$, i.e. $t=0$, for generic $A,B$ says
\begin{equation}
\begin{split}
& \Tr e^{iH \zeta}Be^{-\beta H} A e^{\beta H} = \Tr e^{iH\zeta}AB \Rightarrow \Tr e^{-\beta H}Ae^{\beta H} e^{iH\zeta} B = \Tr e^{iH \zeta} AB \Rightarrow \\
& \zeta = - i \beta
\end{split}
\end{equation}
The normalization condition finally impose $\rho = \frac{e^{-\beta H}}{\Tr e^{-\beta H}}$.

     \subsection{Interplay Between the Vacuum and KMS States} \label{ssec:interplay-vacuum-KMS}
     
As a first application of the thermodynamics properties of the states, we want to show  explicitly that thermal states with respect to some automorphism can be vacuum states with respect to another. This example wants to show the computation at the heart of the Hawking temperature and the Unruh effect (see sections \ref{sec:Hawking-temperature} and \ref{sec:Unruh-effect}), which are related but distinct phenomena exhibited by two different quantum theories, applied in two different spacetimes. Therefore, in order to keep the relevant computation clear, here we will introduce the minimum amount of technology to discuss the example, and we will discuss the quantization schemes of a scalar field in Schwarzschild and in Minkowski spacetime in the respective sections.

Consider a globally hyperbolic spacetime with spherical symmetry $\mathscr M$, on which one finds the space of compactly supported solutions of the K-G equation, $\mathsf{Sol}(\mathscr M) $. In order to define a symplectic structure, we need to identify a suitable hypersurface. For this hypersurface $\Sigma$, we can take the null limit of a Cauchy hypersurface, as we briefly explained in subsection \ref{ssec:classical-field-theory}. An example of this null limit of a Cauchy hypersurface is the past null infinity in Minkowski spacetime. We suppose that the null hypersurface $ \Sigma $ has topology $\mathbb R \otimes \mathbb S_2 $, where $\mathbb S_2$ is the topology of the 2-sphere, on which some affine parameter $V$ is defined. As we saw, the unit normal vector $n$ of $\Sigma$ is tangent to the geodesic generators, $n = \partial_\lambda$, and therefore we define the symplectic form
\begin{equation}
w_\Sigma(\psi_1, \psi_2) = \int_\Sigma \psi_2 \partial_V \psi_1 - \psi_1 \partial_V \psi_2 \dd V \dd \mathbb S_2 \ .
\end{equation} 
Here, $\dd \mathbb S_2 $ is the volume element of the 2-sphere, $\dd \mathbb S_2 = \sqrt \sigma \dd \theta \dd \varphi$. For simplicity we can assume that the radius is constant along $\Sigma$, so that we can safely leave implicit the determinant in the volume element.

Now, suppose that we are interested in the quantization of a scalar field not in the full spacetime $\mathscr M$, but on a subregion $\mathscr R$, say, the causal development of the portion of $\Sigma$ where the affine parameter assumes positive values, $ V > 0 $. $\mathscr R$ can be considered a globally hyperbolic spacetime, and the restriction to $\mathscr R$ of the solutions of the K-G equation in $\mathscr M$ gives the space of solutions of the K-G equation in $\mathscr R$. In the same way, the restriction of the symplectic form to the positive values of $V$ gives a symplectic form in $\mathscr R$; denoting $\mathscr H$ the portion of $\Sigma$ with $V > 0$, we have the symplectic space $(\mathsf{Sol}(\mathscr R), w_{\mathscr H})$, with the symplectic form simply given by
\begin{equation}
w_{\mathscr H}(\psi_1,\psi_2) = \int_{\mathscr H} \psi_2 \partial_V \psi_1 - \psi_1 \partial_V \psi_2 \ \dd V \dd \mathbb S_2
\end{equation}

Although such a construction can look somewhat artificial, this is exactly what we will consider discussing the Hawking temperature in a Schwarzschild spacetime. There, anticipating section \ref{sec:Hawking-temperature}, we consider a Schwarzschild black hole in Kruskal coordinates, and we take the null hypersurface as the past horizon of the black hole, which is a null hypersurface with topology $\mathbb R \otimes \mathbb S_2 $. The portion of past horizon with negative affine parameter is the null limit of a Cauchy hypersurface for the black hole exterior, and we consider the relation between a state defined on the whole spacetime and one in the exterior of the black hole only. Therefore, apart a couple of signs and the name of the affine parameter, the construction is exactly the same. In the Unruh effect instead we will see that the 2-point function coincides with the thermal 2-point function we computed here. So, although the artificial setting we consider here is obscure for now, it is actually a way to compute once the central quantities (the 2-point functions of the vacuum and of a KMS state) of both the Hawking and Unruh effects.

Now, starting from the symplectic structure $(\mathsf{Sol}, w) $, it is possible to construct the Weyl algebra of the observables for the free, scalar quantum field, both in $\mathscr M$ and in its subregion $\mathscr R$. In particular, from the symplectic form it is possible to determine the causal propagator, which in turn defines the antisymmetric part of the 2-point function. From there, one can defines any state of choice; but if we restrict our attention to quasifree, KMS states (including the limiting case $T \to 0$, i.e. the vacuum), we select a particular class of states, identified by their inverse temperature.

Now, the goal of this section is to study the relation between a quasifree, vacuum state in $\mathscr M$ and its restriction to $\mathscr R$. Such a restriction defines a state in $\mathscr R$, which \textit{a priori} is not a vacuum (and in fact it will not be): we want to see what state arises from the restriction of the state to a subregion, or, in other words, from the restriction of the algebra to a subalgebra.

As we saw in (\ref{eq:symplectic-form-causal-propagator-relation}), the integral kernel of the symplectic form gives the causal propagator. We show this in the case at hand (we suppress the angular variables for compactness):
\begin{equation} \label{eq:symplectic-form-V}
\begin{split}
&w(\psi_1,\psi_2) = \int_{\Sigma} \psi_2 \partial_V \psi_1 - \psi_1 \partial_V \psi_2 \dd V = \\
 &= \int_{\Sigma} \int_{-\infty}^\infty \delta(V_1-V_2) \psi_2(V_1) \ \partial_{V_2} \psi_1(V_1) - \delta(V_1 - V_2) \psi_1(V_2) \ \partial_{V_1}\psi_2(V_1) \ \dd V_1 \dd V_2 = \\
 &= 2 \int_{\Sigma}\int_{-\infty}^\infty \psi_1(V_1)\psi_2(V_2)\partial_{V_1}\delta(V_1-V_2) \ \dd V_1 \dd V_2 \ .
 \end{split}
\end{equation}
In the first step we used the properties of the Dirac delta, and in the second step we integrated by parts. As the last integral is over the whole domain of dependence of $\Sigma$, that is, on $\mathscr M$, we see that
\begin{equation}
E(V_1,V_2) = 2 \partial_{V_1}\delta(V_1-V_2) \ .
\end{equation}
The causal propagator in turn determines the antisymmetric part of the 2-point function via the canonical commutation relations:
\begin{equation}
W_2(V_1,V_2) - W_2(V_2,V_1) = i E(V_1,V_2) = 2i \partial_{V_1}\delta(V_1 - V_2) \ .
\end{equation}
On the other hand, the KMS condition $ W_2(V_2,V_1) = W_2(V_1,V_2 + i \beta) $ gives
\begin{equation}
W_2(V_1,V_2) - W_2(V_1, V_2 + i \beta) = 2i \partial_{V_1}\delta(V_1 - V_2) \ .
\end{equation}
Now, as the null hypersurface $\Sigma$ is translationally invariant in the $V$ direction, we assume that the state is invariant too, and in particular the 2-point function is a function of the difference $ V = V_2-V_1$ only (since we can subtract to both variables $V_1$, we denote $W_2(V_2 - V_1) := W_2(0, V_2 - V_1) $). We can therefore take the Fourier transform with respect to the difference $ V = V_2 - V_1 $, of both sides of the above condition. We denote with $k$ the Fourier variable, and we use the convention
\begin{equation}
\hat f(k) = \int_{\mathbb R} e^{- i k x} f(x) \dd x \ .
\end{equation}
The hat denotes the Fourier transform; there should be no confusion with the GNS-representation of the smeared field, $\hat \phi(f) = a(Ef) + a^*(Ef)$ of subsection \ref{ssec:GNS}.

The right-hand side transform is easily computed with an integration by parts:
\begin{equation}
2 i \int_{-\infty}^\infty e^{- i k V} \partial_{V} \delta(V) \dd V = - 2 k \int_{-\infty}^\infty e^{-i k V} \delta(V) \dd V = - 2k \ .
\end{equation}
while the left-hand side is computed using the known property of Fourier transform, $\hat f(x+t) = e^{ikt} f(x) $, which gives
\begin{equation} \label{eq:state-Fourier}
(1 - e^{-\beta k}) \hat W_2(k) = - 2k \Rightarrow \hat W_2(k) = - \frac{2 k}{1 - e^{-\beta k}} \ .
\end{equation} 
This condition uniquely defines the 2-point function, apart from the free parameter $\beta$; now we need to antitransform to find the 2-point function in $\mathscr M$ of a state in real space. Assuming that the state is quasifree, this condition in turn uniquely defines the state, again given its inverse temperatura $\beta$.

The antitransformation can be done in two different cases, at zero temperature ($\beta \to \infty$), to find the vacuum, or at finite temperature, to find a KMS state.

\paragraph{Vacuum with Respect to V}

First we consider the limit $\beta \to \infty$, that is, we look for a quasifree state at zero temperature. In this limit, expression (\ref{eq:state-Fourier}) reduces to
\begin{equation}
    \hat W_2(k) = - 2 \theta(k) k \ .
\end{equation}

As the 2-point function selects the positive frequencies only, it is the vacuum with respect to translations along $V$. Now, we compute the two-point function for the vacuum in real space, with an anti-Fourier transformation. As the Fourier integral is divergent, we need to introduce a regularisation $e^{-\epsilon k} $:
\begin{equation} \label{eq:vacumm2PF-computation}
    \begin{split}
        W_2(V_1 - V_2) &= \lim_{\epsilon \to 0} -2 \int_{-\infty}^\infty e^{i(V_1 - V_2) k} e^{- \epsilon k} k \theta(k) \frac{\dd k}{2 \pi} = \lim_{\epsilon \to 0} -2 \int_0^{\infty} e^{i(V_1-V_2 + i \epsilon) k} k \frac{\dd k}{2 \pi} = \\
        &= \lim_{\epsilon \to 0} \ 2 i \pdv{V_1} \int_0^{\infty} e^{i(V_1-V_2+i \epsilon) k} \frac{\dd k}{2 \pi} = \\
        &= \lim_{\epsilon \to 0} - \frac{1}{ \pi} \pdv{V_1} \frac{1}{V_1 - V_2 + i \epsilon}  = \lim_{\epsilon \to 0} \frac{1}{\pi} \frac{1}{(V_1-V_2 + i \epsilon)^2}
    \end{split}
\end{equation}

This gives the 2-point function in $\mathscr M$. However, we can find the restriction of the two-point function to $\mathscr R $, that is, the corresponding 2-point function for a quasifree state in $\mathscr R$. As $V > 0$ on $\mathscr H$, we can consider test functions with positive $V$ only. Restoring the angular variables we get 
\begin{equation} \label{eq:2PF-vacuum}
    W_2(\psi_1,\psi_2) = \lim_{\epsilon \to 0^+} \frac{1}{\pi} \int_{\mathbb R_+ \times \mathbb R_+ \times \mathbb S_{2}} \frac{\psi_1(V_1) \psi_2(V_2)}{(V_1-V_2 +i \epsilon)^2} \dd V_1 \dd V_2 \dd \mathbb S_2
\end{equation}

This is the 2-point function for a quasifree state with respect to translations in $V$ restricted to $\mathscr R$. The point is that the restriction of the vacuum state to a subregion is no longer a vacuum. We will show this key fact computing the 2-point function of a KMS state at finite temperature in $\mathscr R$, and showing that it coincides with \eqref{eq:2PF-vacuum}. In order to do so, we first make a change of variables. As $\mathscr H$ is restricted to $V > 0$, we can perform the substitution $V = e^{\kappa v} $, for any $\kappa \in \mathbb R$ (but the choice of the symbol should be rightfully suspicious), to rewrite the 2-point function in an equivalent way that will be easily connected with the 2-point function of a thermal state:
\begin{equation} \label{eq:2PF-vacuum-v}
\begin{split}
    W_2(\psi_1,\psi_2) &= \lim_{\epsilon \to 0^+} \frac{\kappa^2}{\pi} \int_{\mathbb R_+ \times \mathbb R_+ \times S^{2}} 
    \frac{\psi_1(v_1) \psi_2(v_2)}{(e^{\kappa v_1}-e^{\kappa v_2} +i \epsilon)^2} e^{\kappa(v_1+v_2)} \dd v_1 \dd v_2 \dd \mathbb S_2 = \\
    &= \lim_{\epsilon \to 0^+} \frac{\kappa^2}{\pi} \int_{\mathbb R_+ \times \mathbb R_+ \times S^{2}} \frac{\psi_1(v_1) \psi_2(v_2)}{\big(e^{\frac{\kappa}{2}(v_1-v_2)}-e^{-\frac{\kappa}{2}(v_1- v_2)} +i \epsilon \big)^2} \dd v_1 \dd v_2 \dd \mathbb S_2 \\
    &=  \lim_{\epsilon \to 0^+} \frac{\kappa^2}{4 \pi} \int_{\mathbb R_+ \times \mathbb R_+ \times S^{2}} \frac{\psi_1(v_1) \psi_2(v_2)}{\bigg( \sinh{\frac{\kappa}{2}(v_1-v_2)} + i\epsilon \bigg)^2} \dd v_1 \dd v_2 \dd \mathbb S_2
\end{split}
\end{equation}

We want to show that this is indeed a thermal state with inverse temperature $\beta = \frac{2 \pi}{\kappa}$ with respect to translations in the $v$ variable. To do this, we take a step back to \eqref{eq:state-Fourier}.

\paragraph{Thermal States with Respect to v}

Now, we compute the 2-point function for a KMS state at inverse temperature $\beta$ by anti-transforming the expression (\ref{eq:state-Fourier}). However, to understand the connection with the vacuum 2-point function in $\mathscr M$, we note that we can change variable in the integral defining the symplectic form in \eqref{eq:symplectic-form-V} from $V$ to $v$, leading to a causal propagator directly defined in $\mathscr R$:
\begin{equation} \label{eq:causal-propagator-thermal}
E(v_1,v_2) = 2 \partial_{v_1}\delta(v_1 - v_2) \ .
\end{equation}
This causal propagator in $\mathscr R$ is equivalent to the one we wrote in terms of $V$ in $\mathscr M$. With the exact same computation we did before, we find
\begin{equation}
W_2(v_1,v_2) - W_2(v_1,v_2 + i \beta) = i E(v_1,v_2) \ .
\end{equation}
Now, assuming as before the translational invariance of the hypersurface, we can take the Fourier transform with respect to the variable $v = v_2 - v_1$. As the Fourier transform depends on the coordinates, here comes the fact that the 2-point function defined this way is \textit{not} equivalent to \eqref{eq:2PF-vacuum}. However, the computation is the same and the Fourier transform is again
\begin{equation} \label{eq:state-fourier-v}
\hat W_2(k) = - \frac{2 k}{1-e^{-\beta k}} \ .
\end{equation}
It is important to remark however that this is \textit{not} the same as \eqref{eq:state-Fourier}, although they formally coincide (this is actually our fault: for laziness we use the same letter $k$, to denote the Fourier transform with respect of different variables, $V$ and $v$; to be rigorous we should call the new Fourier variable with a different letter, but the alphabet is short, and physics is long, so we will stick to $k$ and we will remark the difference with the above computation to the point of being pedantic).

The point is that \eqref{eq:state-fourier-v} gives the (Fourier transform of) the 2-point function directly in $\mathscr R$; to compute the 2-point function of a thermal state with respect to $v$, we just need to antitransform the equation keeping $\beta$ finite. To do so, we note that, for $k > 0$,
\begin{equation}
    \frac{1}{1 - e^{-\beta k}} - \theta(k) = \frac{1}{1 - e^{-\beta k}} - 1 = \frac{e^{- \beta k}}{1- e^{- \beta k }} \ ,
\end{equation}
while for $k < 0$
\begin{equation}
    \frac{1}{1 - e^{-\beta k}} - \theta(k) = - \frac{e^{\beta k}}{1 - e^{\beta k}} \ .
\end{equation}
Putting these two expressions together we have
\begin{equation} \label{eq:aF-2pt-v-intermediate}
    k \bigg( \frac{1}{1 - e^{-\beta k}} - \theta(k) \bigg) = \abs{k} \frac{e^{- \beta \abs{k}}}{1 - e^{- \beta \abs{k}}} \ .
\end{equation}

As we computed the Fourier anti-transform of $k \theta(k)$ in the last paragraph, we can compute the right-hand side to arrive at the formula we need. We first rewrite the integral to get rid of the absolute value, using the same regularisation as before, with a  limit $\epsilon \to 0^+$ intended: 
\begin{equation}
    \begin{split}
   & \int_{-\infty}^\infty \abs{k} \frac{e^{- \beta \abs{k}}}{1 - e^{- \beta \abs{k}}} e^{iv k - \epsilon \abs{k}} \frac{\dd k}{2 \pi} =
     \int_0^\infty k \frac{e^{- \beta k}}{1 - e^{- \beta k}} e^{(iv - \epsilon) k} \frac{\dd k}{2 \pi} - 
     \int_{- \infty}^0 k \frac{e^{\beta k}}{1 - e^{\beta k}} e^{(iv + \epsilon) k} \frac{\dd k}{2 \pi} = \\
   &= \int_0^\infty k \frac{e^{- \beta k}}{1 - e^{- \beta k}} e^{(iv - \epsilon) k} \ \frac{\dd k}{2 \pi} 
   + \int^\infty_0 k \frac{e^{- \beta k}}{1 - e^{- \beta k}} e^{-(iv + \epsilon) k} \ \frac{\dd k}{2 \pi} = \\
   &= 2 \Re{\int_0^\infty k \frac{e^{- \beta k}}{1 - e^{- \beta k}} e^{(iv - \epsilon) k} \ \frac{\dd k}{2 \pi}} \ .
    \end{split}
\end{equation}
In the first step we separate the integrals over positive and negative frequencies, in the second step we make the change $k \to -k$ in the second integral, and in the last step we note that the integral is in the form $(z + \bar z)/2$.
 
Now, as the integral is restricted to $k >0$, we can take the limit $\epsilon \to 0^+$, thanks to the decaying real exponential. Moreover, as $e^{-\beta k} < 1$, we can use the geometric series
\begin{equation}
\sum_{n=0}^\infty \frac{1}{x^n} = \frac{1}{1-x}
\end{equation}
to rewrite the integral, $x = e^{-\beta k} $:
\begin{equation}
    2\Re{\sum_{n=0}^\infty \int_0^\infty \ k \ e^{ivk} e^{-(n+1) \beta k } \frac{\dd k}{2 \pi}} =
    - 2 \Re{\sum_{n=1}^\infty i \pdv{v} \int_0^\infty e^{(iv - n\beta)k} \ \frac{\dd k}{2 \pi}} \ .
\end{equation}
We can now integrate term by term, arriving at
\begin{equation}
 \frac{1}{\pi} \Re{ i \pdv{v} \sum_{n=1}^\infty \frac{1}{iv - n \beta} } \ .
\end{equation}
Taking the real part, we get
\begin{equation}
    \frac{1}{2\pi} \pdv{v} \sum_{n=1}^\infty \bigg( \frac{i}{iv - n\beta} - \frac{i}{- iv - n\beta}\bigg) =  \frac{1}{2\pi} \pdv{v} \sum_{n=1}^\infty \frac{2 v}{v^2 + n^2 \beta^2} \ .
\end{equation}
The series has a known, exact form:
\begin{equation} \label{eq:anti-fourier-computation}
    \frac{1}{\pi} \pdv{v}\frac{1}{2v}\bigg( \frac{\pi v}{\beta} \coth{\frac{\pi v}{\beta}} - 1\bigg) \ .
\end{equation}
We recognise the Fourier anti-transform of $\theta(k) k$; in fact, from (\ref{eq:vacumm2PF-computation}) we see that
\begin{equation}
   \widehat{\theta(k)k} = \frac{1}{2 \pi} \pdv{v} \frac{1}{v} \ .
\end{equation}
To get the 2-point function we go back to \eqref{eq:aF-2pt-v-intermediate}, from which
\begin{equation}
\hat W_2(k) = -2\bigg ( \abs{k} \frac{e^{-\beta \abs{k}}}{1-e^{-\beta \abs{k}}} + k \theta(k) \bigg) \Rightarrow
\end{equation}

\begin{equation}
    W_2(v_1, v_2) = - \frac{1}{\beta} \pdv{v_1} \coth{\frac{\pi (v_1 - v_2)}{\beta}} = \frac{\pi}{\beta^2} \frac{1}{\sinh^2{\big [\frac{\pi}{\beta}(v_1 - v_2) \big ]}} \ .
\end{equation}
Restoring again the angular variables, the two-point function then is
\begin{equation} \label{eq:TPF-thermal}
    W_2(\psi_1, \psi_2) = \lim_{\epsilon \to 0^+} \frac{\pi}{\beta^2} \int_{\mathbb R \times \mathbb R \times S^{2}} \frac{\psi_1(v_1) \psi_2(v_2)}{\bigg( \sinh{\frac{\pi}{\beta}(v_1-v_2)} + i\epsilon \bigg)^2} \dd v_1 \dd v_2 \dd \mathbb S^2 \ .
\end{equation}
So the causal propagator \eqref{eq:causal-propagator-thermal} give rise to this 2-point function, which corresponds to a thermal state at inverse temperature $\beta$. Now, the connection with the previous computation should be clear: in fact, this 2-point function in $\mathscr R$ coincides with the restriction of the vacuum state in $\mathscr R$. This conclude the reasoning: the vacuum with respect to $V$ in $\mathscr M$ is a thermal state with respect to $v$ in $\mathscr R$ with inverse temperature $\beta = \frac{2 \pi}{\kappa}$.
    
\section{Entropy and relative entropy} \label{sec:E-EE}
Almost at the same time when the thermal properties of QFT, described by the theory of KMS states, was developed, the mathematical theory of modular automorphisms of algebras was discovered by mathematicians (in fact, both the two theories were announced during the Baton Rouge conference in 1967. While Tomita was distributing a preprint of his work on modular theory in the audience, Haag announced the results he obtained with Hugenholtz and Winnick on KMS states from the stage). Between the two, a "prestabilized harmony", as Haag called it \cite{Haag92} is highlighted, since the modular theory perfectly describes the induced von Neumann algebra of the GNS representation of KMS states. Here we review the main results of the modular theory. Our goal is to find the correct generalisation of entropy to continuum theories, and in the meantime to show how this beautifully connects with the KMS condition.

In non-relativistic QM the principal notion of entropy is that of entanglement entropy. Given a pure state $\ket \Phi$ in a Hilbert space $\mathcal H$, its associated density matrix is simply $\rho_\Phi = \ket \Phi \bra  \Phi$. If we now want to consider a subsystem localized in a subregion $\mathscr R$ of the full space, we can trace over the degrees of freedom outside $\mathscr R$, that is, in the complement $\mathscr R'$. This way, we introduce the reduced density matrix $\rho_\Phi^{\mathscr R} = \Tr_{\mathscr R'} \rho $. The quantity which measures the entanglement between different subregions of the system is the \textit{entanglement entropy}, defined as $S_{\mathscr R} = - \Tr\rho_\Phi^{\mathscr R} \log \rho_\Phi^{\mathscr R} $. This quantity was first introduced by Bombelli et al. (1986) \cite{Bombelli86} in an attempt to characterize Bekenstein-Hawking entropy as the entanglement entropy between the degrees of freedom inside and outside the event horizon, and was subsequently studied in statistical and condensed matter physics. It was early shown that the entanglement entropy is divergent in the continuum limit, and thus it is generally used introducing a cutoff, as a lattice, assuming that physical quantities are meaningful in the continuum limit as long as they do not depend on the cutoff. However, one can introduce a better-behaved quantity subtracting the vacuum entropy. More generally, given two states $\ket \Omega, \ket \Phi$ and their associated densitiy matrices $\rho_\Omega$ and $\rho_\Phi$, one can define the \textit{relative entropy} as $S(\Omega|\Phi) = - \Tr \rho_\Omega (\log\rho_\Phi - \log \rho_\Omega) $. This quantity, rather than measuring the entanglement between two localized regions,  measures the indistinguishability between states. The relative entropy is a non-local quantity, as opposed to the entanglement entropy, but in turn admits a generalisation in the continuum via the Araki formula \cite{Araki76}, constructed from the Tomita-Takesaki theory of modular automorphisms. As usual, we require that the relative entropy reduces to the known notion of entropy in the case of QM, therefore providing the correct generalisation for QFT algebras. The generalisation will be given in equation \eqref{eq:araki-formula}.  

The main result of this section is the computational tool of the relative entropy for coherent states. Indeed, relative entropy for coherent states can be computed using the symplectic form only, without referencing to the quantum theory and the expectation value of some observable. This way, from very abstract algebraic notions we are able to actually compute the entropy and to use it in connection with the energy conservation in the following sections.

    \subsection{Modular Hamiltonian and Modular Flow} \label{ssec:modular-hamiltonian}
We start very generally, giving the definitions of Tomita-Takesaki theory for general algebras, and we slowly reduce to the case of coherent states for a QFT. We will prove the properties of Tomita-Takesaki theory which are most useful for the characterization of the relative entropy of coherent states, with no attempt in being general or complete. In this subsection we mainly follows \cite{CasiniGrilloPontiello19} and \cite{Witten18}.

So, let $ \mathfrak U(\mathcal H) $ be a von Neumann algebra on a Hilbert space $ \mathcal H $, and let $\ket \Omega \in \mathcal H$ be a cyclic and separating vector. Then, there exists a unique antilinear operator $ S_\Omega$, called the \textit{modular involution} or \textit{Tomita operator}, such that
\begin{equation} \label{def:tomita-operator}
   S_{\Omega} A \ket{\Omega} = A^* \ket{\Omega} \ .
\end{equation}
From the definition, it is clear that $S^2_\Omega = 1 $, and therefore it is invertible. Moreover, $ S_\Omega \ket \Omega = \ket \Omega $.

An invertible, closed\footnote{The Tomita operator as defined in \eqref{def:tomita-operator} is not closed, but is closable. We denote the closure of $S_\Omega$ with the same symbol.} operator always has a unique polar decomposition, $S_\Omega = J_\Omega \Delta_\Omega^{1/2}$, where the \textit{modular conjugation} $J_\Omega$ is an anti-linear, unitary operator and the \textit{modular operator} $\Delta_\Omega$ is self-adjoint and non-negative. From these basic properties one can show that
\begin{equation}
S_\Omega^*S_\Omega = \Delta_\Omega^{1/2}J_\Omega^*J_\Omega\Delta_\Omega^{1/2} = \Delta_\Omega \ .
\end{equation}
As $S_\Omega \ket \Omega = S_\Omega^* \ket \Omega = \ket \Omega $, it is immediate to see that $ \Delta_\Omega \ket \Omega = \ket \Omega $ and therefore $ J_\Omega \ket \Omega = \ket \Omega $. The first important relation follows, that is,
\begin{equation}
\Delta_\Omega^{is} \ket \Omega = \ket \Omega
\end{equation}
for any $s \in \mathbb R$, which can be found by Taylor expanding in $s$ and using for each term the invariance of $\ket \Omega$ with respect to $\Delta_\Omega$. Moreover, from $S^2_\Omega = 1$, we have that
\begin{equation}
1 = J_\Omega \Delta_\Omega^{1/2} J_\Omega \Delta_\Omega^{1/2} \Rightarrow J_\Omega \Delta_\Omega^{1/2} J_\Omega = \Delta_\Omega^{-1/2} \ ,
\end{equation}
and as before $J_\Omega \Delta_\Omega^{is} J_\Omega = \Delta_\Omega^{is} $ also holds (without the minus sign in the exponential, because the modular conjugation act as conjugation on the complex coefficients of the Taylor expansion). Using twice the above property, we have
\begin{equation}
S_\Omega = J_\Omega \Delta^{1/2}_\Omega = J_\Omega^2 \Delta^{-1/2}_\Omega J_\Omega = J^2_\Omega J_\Omega \Delta^{1/2} \ .
\end{equation}
By uniqueness of the polar decomposition it follows that $J^2_\Omega = 1$.

The \textit{modular Hamiltonian} is then defined as
\begin{equation}
    K_\Omega = - \log \Delta_\Omega \ .
\end{equation}
$K_\Omega$ is a self-adjoint operator with generally unbounded spectrum. By Stone Theorem, it defines a 1-parameter group of unitary operators on the von Neumann algebra,
\begin{equation}
\alpha_s(A) = e^{-iK_\Omega s}A e^{iK_\Omega s} = \Delta_\Omega^{is} A \Delta_\Omega^{-is} \ .
\end{equation}
The one-parameter group of unitaries $ \Delta^{is} $ is called \textit{modular flow}. Note the minus sign of difference with respect to the definition of the action of the time translation group \eqref{eq:time-translation-group} in the subsection on KMS states \ref{ssec:KMS} (could we define it in a coherent way? Yes. But this way the associated temperature is positive, as we will show in a moment, and as physicists that's the best we can hope for).

We can as well define the Tomita operator for the commutant of $ \mathfrak U(\mathcal H) $, that is, the set of bounded operators $\mathfrak U'(\mathcal H) = \{ \ U' \ | \ [U,U'] = 0 \  \forall U \in \mathfrak U(\mathcal H) \  \} $, and perform the same polar decomposition with the same properties as above. The relations between the von Neumann algebra, the modular flow, and its commutant are summed up in the following
\begin{theorem}[Tomita-Takesaki]
\begin{align}
    & J_\omega \mathfrak U(\mathcal H) J_\omega = \mathfrak U(\mathcal H)' \\
    & \Delta_{\omega}^{is} \mathfrak U(\mathcal H) \Delta_{\omega}^{-is} = \mathfrak U(\mathcal H) \quad \Delta_{\omega}^{is} \mathfrak U(\mathcal H)' \Delta_{\omega}^{-is} = \mathfrak U(\mathcal H)'
\end{align} 
\end{theorem}
The first property says that the modular conjugation maps the algebra into its commutant. The second one states that the modular flow defines an automorphism of the algebra. 
We can ask, now, how the state $\omega$ defined by the cyclic and separating vector $\ket \Omega$ behaves with respect to this automorphism, and the answer is the connection between the modular theory and the KMS theory. In fact, one can show that the state $\omega(A) = \ev{A}{\Omega} $ is KMS with respect to the modular automorphism $\alpha_s(A) $ with parameter $ \beta = 1 $:
\begin{multline}
\omega(\alpha_s (A)B) = \ev{\Delta_\Omega^{is}A \Delta_\Omega^{-is}B}{\Omega} =
 \ev{A \Delta_\Omega^{-is}B}{\Omega} =
 \ev{A \Delta^{-is}_\Omega J_\Omega \Delta^{1/2}_\Omega B^*}{\Omega} = \\
= \ev{A J_\Omega \Delta^{-is + 1/2}_\Omega B^* }{\Omega} =
 \ev{A \Delta_\Omega^{-is + 1} B}{\Omega} = \ev{A \alpha_{-(s+i)}B}{\Omega} \ .
\end{multline}
In the first step we used the invariance of $ \ket \Omega$ under the modular operator, in the second the definition of the Tomita action, in the third the invariance of the modular flow under modular conjugation and the property $J_\Omega^2 = 1 $, in the last we acted again with the Tomita operator and used again the invariance of the modular operator under modular conjugation. The unfortunate minus sign is the same unfortunate minus sign in the definition of the modular flow, but now our choice is understandable: the alternative is a negative temperature.

The consequences of this simple computation are deep: in fact, this shows that a KMS state with respect to the time translation automorphism can be represented as a cyclic and separating vector for an algebra of observables whose modular automorphism $\alpha_s$ is related to the time translations as $t = -\beta s$.

The last issue to consider is how to interpret the action of the modular flow over the observables. This is very difficult in general, but in the case of a free theory in flat spacetime there is a result, the \textit{Bisognano-Wichmann Theorem}, which links the modular flow to a geometric action. To understand its statement, let's introduce some notions.

In Minkowski $\mathscr M$, we consider the region $\mathscr W = \{ \mathbf x \in \mathbb R^{1,3} \ | \ x > \abs{t}\} $, called \textit{Rindler wedge}. This is the domain of dependence of the surface $\Sigma = \{ \mathbf x \in \mathbb R^{1,3} \ | \ t=0, \ x > 0 \} $. Consider the vacuum vector $\ket \Omega$ for the theory defined on the whole Minkowski spacetime $\mathfrak U(\mathscr M) $. By the Reeh-Schlieder theorem, this is a cyclic and separating vector for the algebra of observables restricted to the Rindler wedge $\mathfrak U(\mathscr W)$. We want to understand the action of $J_\Omega$ and $\Delta_\Omega$ on $\mathfrak U(\mathscr W) $. Then we have
\begin{theorem}[Bisognano-Wichmann] \label{th:bisognano-wichmann-theorem}
\begin{equation}
J_\Omega = \Theta U(R_1(\pi)) \quad \Delta_\Omega = e^{-2 \pi K_1} \ ,
\end{equation}
\end{theorem}
where $\Theta$ is the CPT operator, $U(R_1(\pi)) $ is the unitary operator representing a space rotation of $ \pi $ degrees around the $x$ axis and $K_1$ is the generator of the one-parameter group of boosts in the plane $(t,x)$. This means that at least in this particular case, we can compute the action of the modular flow as a geometric transformation on the fields. For a free QFT at finite temperature, this is connected to the KMS properties, and therefore the time translations generated by the Hamiltonian, the temperature, and the translations generated by boosts are all intertwined.

    \subsection{Relative Modular Hamiltonian and Relative Modular Flow} \label{ssec:relative-modular-hamiltonian}
Before proceeding, we introduce a technical notion which will be useful in discussing the properties of relative entropy. Given the induced von Neumann algebra $\mathfrak U(\mathcal A)$ and a cyclic and separating vector $\ket \Omega$ representative of some state $\omega$, another state $\omega_\Phi$ has a unique representative vector $\ket \Phi$ in the so-called \textit{natural cone}, that is, the set of vectors
\begin{equation}
\mathcal P_\Omega = \overline{\{ A j(A) \ket{\Omega} | A \in \mathfrak U \}} \ ,
\end{equation}
where the bar means the closure and $j_\Omega(A) = J_\Omega A J_\Omega$. Then the state is realized in the usual way, $\omega_\Phi(A) = \ev{A}{\Phi} $. The one-to-one correspondence between vectors and states comes at hand when we want to discuss the properties of states under some automorphism, and in turn does not introduce technical difficulties.

Now, given two cyclic and separating vectors, $\ket{\Omega} \in \mathcal H, \ \ket{\Phi} \in \mathcal P_\Omega$, we can generalise the construction of the entropy to the notion of relative entropy. We define the \textit{relative Tomita operator} (or \textit{relative modular involution})  as
\begin{equation} \label{eq:relative-tomita-operator}
    S_{\Omega, \Phi} A \ket{\Phi} = A^* \ket{\Omega} \ ,
\end{equation}
which admits the unique polar decomposition $S_{\Omega, \Phi} = J_{\Omega, \Phi} \Delta_{\Omega, \Phi}^{1/2}$. Since $\ket \Phi$ is in the natural cone, one can show that $ J_{\Omega, \Phi} = J_\Omega$. The relative modular Hamiltonian is defined as
\begin{equation}
    K_{\Omega, \Phi} = - \log \Delta_{\Omega, \Phi} \ .
\end{equation}
Finally, the Araki formula gives the relative entropy:
\begin{equation} \label{eq:araki-formula}
    S(\omega_\Omega | \omega_\Phi) = \ev{\log \Delta_{\Omega, \Phi}}{\Omega} = - \ev{K_{\Omega, \Phi}}{\Omega} \ .
\end{equation}
We can check that this reduces to the usual notion of relative entropy for non-relativistic quantum mechanics. In fact, in QM states are realised with density matrices, $\omega(A) = \Tr \rho_\Omega A $, while the relative modular operator reduces to
\begin{equation}
\Delta_{\Omega, \Phi} = \rho_\Omega \otimes \rho_\Phi^{-1} \ .
\end{equation}
Then, the relative entropy takes the familiar form $S(\Omega|\Phi) = - \Tr \rho_\Omega (\log\rho_\Phi - \log \rho_\Omega) $.

\subsection{Relative entropy for coherent states} \label{ssec:EE-coherent-states}

We start proving some useful properties of relative entropy in a slightly more general setting than what we will eventually consider. The main references for this subsection are \cite{CasiniGrilloPontiello19}, \cite{CiolliLongoRuzzi19}, and \cite{HollandsIshibashi19}. For now, we take two cyclic and separating vectors $\ket \Omega$ and $\ket \Phi$ and some unitary operator $U \in \mathfrak U(\mathcal H) $ which is an automorphism of the algebra, $ U^* \mathfrak U U = \mathfrak U $. 

In this case, for $A \in \mathfrak U$, we have
\begin{equation}
    US_{\Omega, \Phi}U^* A U \ket{\Phi} =  US_{\Omega, \Phi}(U^* A U) \ket{\Phi} = U (U^* A U)^* \ket{\Omega} = A^* U \ket{\Omega} \ ,
\end{equation}
which implies 
\begin{equation} \label{eq:U1}
    US_{\Omega,\Phi} U^* = S_{U\Omega, U\Phi}
\end{equation}
By taking the polar decomposition,
\begin{equation}
UJ_\Omega \Delta^{1/2}_{\Omega, \Phi} U^*= U J_\Omega U^* U \Delta^{1/2}_{\Omega, \Phi} U^* = J_\Omega \Delta^{1/2}_{U \Omega,U \Phi}  \ ,
\end{equation}
one gets $ \Delta_{U \Omega, U \Phi} = U \Delta_{\Omega, \Phi} U^*$, proving that $U J_\Omega U^* =J_\Omega j_\Omega(U) J j_\Omega( U^*) J_\Omega = J_\Omega $, using $ J_\Omega J_\Omega = 1$ and $ j_\Omega(U) \in \mathfrak U'$. This becomes a property of relative entropy, expanding the logarithm in Taylor series:
\begin{equation} \label{eq:K}
\begin{split}
    & - K_{U \Omega, U \Phi} = \log (\Delta_{U \Omega , U\Phi}) = \log (U\Delta_{\Omega \Phi}U^*) = \sum \frac{(-1)^n}{n!}(U\Delta_{\Omega \Phi}U^*)^n = \\
    &= \sum \frac{(-1)^n}{n!} U\Delta_{\Omega \Phi}U^* \underbrace{\cdots\,}_\text{$n$ times} U\Delta_{\Omega \Phi}U^* = \\
    &=U \sum \bigg( \frac{(-1)^{n}}{n!} \Delta_{\Omega \Phi}^n \bigg) U^* = U\log \Delta_{\Omega \Phi}U^* \\
    &= - U K_{\Omega, \Phi} U^* \ .
\end{split}
\end{equation}
In the first line we use the property we just proved and expanded the logarithm, in the second we expanded the $n$-th power of the $n$-th coefficient of the series, in the third we noted that $U^*U= \mathbf 1$.

Now, we write explicitly $\ket \Phi = j_\Omega(U) U \ket \Omega$, for some unitary operator $U \in \mathfrak U(\mathcal H)$. The operator $\tilde U = j_\Omega(U) U$ is still a unitary operator, and the corresponding state functional is $\omega_\Phi(A) = \braket{\Omega U}{A U \Omega} = \omega(U^*AU)$. The above property in this special case becomes
\begin{equation}
    S_{\Omega,\Phi} = S_{\Omega, \tilde U \Omega} = S_{\tilde U \tilde U^* \Omega, \tilde U \Omega} = \tilde U S_{\tilde U^* \Omega, \Omega} \tilde U^* = U j_\Omega(U) S_{\tilde U^*\Omega,\Omega} (Uj_\Omega(U))^* \ . 
\end{equation}
This again becomes a property of the relative modular operator, $ \Delta^{1/2}_{\Omega, U \Omega} = U \Delta^{1/2}_{U^* \Omega, \Omega}U^* $. These properties let one express the relative entropy between two coherent states in terms of the relative entropy of a coherent state with respect to the vacuum. In fact, suppose we have \textit{two} unitary operators, $\tilde U, \tilde V$, defining two different vectors, $\ket \Phi = \tilde U \ket \Omega $ and $\ket \Psi = \tilde V \ket \Omega$ in the natural cone of $\ket \Omega$. Suppose we want to compute the relative entropy between the vectors $\ket \Phi$ and $\ket \Psi$. Then we have
\begin{equation}
S(\omega_\Psi | \omega_\Phi) = - \ev{\tilde U^* K_{\tilde V\Omega, \tilde U\Omega} \tilde U}{\Omega} = - \ev{K_{\tilde U^* \tilde V \Omega,\Omega}}{\Omega} = S(\omega_{ \tilde U^* \tilde V \Omega}|\omega_\Omega) \ .
\end{equation}
In particular, if $U = W(f)$, $V = W(g)$, for some test functions $f,g$, the vectors $\ket \Phi$ and $\ket \Psi$ are coherent vectors, and the vector $\tilde U^* \tilde V \ket \Omega$ can be written as
\begin{equation}
\tilde U^* \tilde V \ket \Omega = W(-f) j_\Omega(W(-f)) \  W(g) j_\Omega(W(g)) = W(g-f)j_\Omega(W(g-f)) \ket \Omega
\end{equation}
using the Weyl relations and the fact that $J_\Omega$ conjugates the phase in the Weyl relation, so the two phases cancel out. The result is still a coherent vector, and therefore we can restrict our attention to the relative entropy of a coherent vector with respect to $\ket \Omega$.

We can now prove the first important result, that is, that the relative Tomita operator between $\ket \Omega$ and $\ket \Phi$ can be computed using the Tomita operator only.
\begin{claim} \label{claim-relative-tomita-and-tomita-relation}
    \begin{equation*}
        S_{\Omega, \Phi} = U S_{\Omega} j_{\Omega}(U^*) 
    \end{equation*}
And by polar decomposition
\begin{equation}
\Delta^{1/2}_{\Omega, \Phi} = j_{\Omega}(U) \Delta^{1/2}_{\Omega} j_{\Omega}(U^*) \ \ K_{\Omega, \Phi} = j_\Omega(U) K_\Omega j_\Omega(U^*) \ .
\end{equation}

\end{claim}

\begin{proof}
    \begin{equation*}
        \begin{split}
        & ( U S_{\Omega} j_{\Omega}(U^*) ) \ A \ket{\Phi} =
         U S_{\Omega} (AU) \ket{\Omega} = A^* \ket{\Omega} = S_{\Omega, \Phi} A \ket{\Phi} \Rightarrow \\
        & J_{\Omega} \Delta_{\Omega, \Phi}^{1/2} = U J_{\Omega} \Delta_{\Omega}^{1/2} j_{\Omega}(U^*) \Rightarrow \\
        &\Delta_{\Omega, \Phi}^{1/2} = j_{\Omega}(U) \Delta_{\Omega}^{1/2} j_{\Omega}(U^*) \ .
        \end{split}
    \end{equation*}
\end{proof}
This means that the relative entropy between $ \ket \Omega$ and $\ket \Phi $ can be computed from the modular operator only. This is shown by direct computation. We start from the expression of the relative entropy, using the claim (\ref{claim-relative-tomita-and-tomita-relation}):
\begin{equation}
    S(\omega_{\Omega} | \omega_{\Phi} ) = \ev{\log\Delta_{{\Omega}, U \Omega}}{\Omega} = 2  \ev{\log\Delta_{\Omega,\Phi}^{1/2}}{\Omega} = 2 \ev{\log(j_{\Omega}(U) \Delta^{1/2}_{\Omega} j_{\Omega}(U^*))}{\Omega} \ .
\end{equation}
By the same computation of (\ref{eq:K}), one has $  \log (j_{\Omega}(U)\Delta_{\Omega}j_{\Omega}(U^*) ) =j_{\Omega}(U) \log \Delta_{\Omega} j_{\Omega}(U^*) $. Exploiting this, one can rewrite the above expression making use of the invariance of $\ket \Omega$ under $J_\Omega$:
\begin{equation*}
   \mel{\Omega}{J_{\Omega}U J_{\Omega} \log \Delta_{\Omega} J_{\Omega}U^*J_{\Omega}}{\Omega} = \mel{\Omega U^*}{J_{\Omega} \log \Delta_{\Omega} J_{\Omega}}{U^* \Omega} \ .
\end{equation*}
Now, as $J_\Omega \Delta_\Omega J_\Omega = \Delta_\Omega^{-1}$, one can conclude:
\begin{equation*}
= - \mel{\Omega U^*}{\log \Delta_\Omega}{U^* \Omega} \\
= i \eval{ \dv{t} \mel{\Omega U^*}{\Delta_\Omega^{it}}{U^* \Omega}}_{t=0} \Rightarrow
\end{equation*}
\begin{equation} \label{eq:relative-entropy-modular-hamiltonian}
S(\omega_\Omega |\omega_{U\Omega} ) = - \ev{K_\Omega}{U^*\Omega} = i \eval{\dv{t}\ev{e^{-iKt}}{U^* \Omega}}_{t=0} \ .
\end{equation}

Finally, we specialise to the case in which $\ket \Omega$ is the vacuum and $U = W(f)$ is a Weyl operator for some smearing test function $f$. Then, $\ket \Phi = j_\Omega(W(f))W(f) \ket \Omega$ is a coherent vector with respect to $\ket \Omega$ in its natural cone, with state representative $\omega_f(A) = \ev{A}{\Phi} = \omega(W^*(f) A W(f))$. We use $f$ as a label to remember that is a coherent state, and can be considered a classical perturbation $\psi_f= Ef$ of the vacuum. The entropy for coherent states is indeed computed using the classical structure only, that is, from the symplectic form of $\mathsf{Sol}$. In fact, substituting the definitions of coherent state in \eqref{eq:relative-entropy-modular-hamiltonian}, one has
\begin{equation}
S(\omega | \omega_f) = i \eval{ \dv{t} \ev{W(f)^* \Delta^{it}W(f)}{\Omega}}_{t=0} \ .
\end{equation}
Using the invariance of the vacuum with respect to $\Delta_\Omega$ it is possible to rewrite this expression as the action of the automorphism $\alpha_t$ over the Weyl operator:
\begin{equation}
i \eval{\dv{t} \ev{W^*(f)\Delta^{it} W(f) \Delta^{-it}}{\Omega}}_{t=0} = i \eval{\dv{t} \ev{W^*(f) \alpha_tW(f)}{\Omega}}_{t=0} \ .
\end{equation}
Defining $\alpha_tW(f) = W(\alpha_t(f)) $, one can use the Weyl relations,
\begin{equation}
W^*(f)W(\alpha_t(f)) = e^{-\frac{i}{2}w(-\psi_f, \alpha_t(\psi_f))}W(\alpha_t(f)-f)
\end{equation}
and
\begin{equation}
 \ev{W(f)}{\Omega} = e^{-\frac{\norm{K\psi_f}^2}{2}} \ ,
\end{equation}
from which one can write the expectation value as
\begin{multline}
\ev{W^*(f)W(\alpha_t(f))}{\Omega} = e^{-\frac{i}{2}w(-\psi_f, \alpha_t(\psi_f))}\ev{W(\alpha_t(f)-f)}{\Omega} = \\
\exp{-\frac{i w(-\psi_f, \alpha_t(\psi_f))}{2} - \frac{\norm{K(\alpha_t(f)-f)}^2}{2} } \ .
\end{multline}

Now, we can take the derivative with respect to $ t $:
\begin{equation}
\begin{split}
&i \dv{t} \eval{\ev{W^*(f)W(f)}{\Omega}}_{t=0} = \\
&= \Bigg[ \bigg( \frac{1}{2} \dv{t}\bigg(w(-\psi_f, \alpha_t(\psi_f)\bigg) -i \norm{K(\alpha_t(\psi_f) - \psi_f)} \dv{t} \norm{K(\alpha_t(\psi_f) - \psi_f)} \bigg) \times \\
&e^{-(iw(-\psi_f, \alpha_t(\psi_f))+\norm{K(\alpha_t(\psi_f) - \psi_f)}^2)/2} \Bigg]_{t=0} =\\
&= \frac{1}{2} w \bigg( \eval{\dv{\alpha_t(\psi_f)}{t}}_{t=0}, \psi_f \bigg)  \ .
\end{split}
\end{equation}
This conclude our derivation: the relative entropy between a coherent state and the vacuum is given by the derivative of the symplectic form of the associated classical solution:
\begin{equation} \label{eq:EE-foundamental-formula}
    S(\omega | \omega_f) = \frac{1}{2} w(\eval{\dv{t} \alpha_t(\psi_f)}_{t=0}, \psi_f) \ .
\end{equation}

What remains to compute is the action of the automorphism $\alpha_t$ over the classical field $\psi_f$, which can be computed in some cases (as ours) as a Poincaré transformation thanks to the Bisognano-Wichmann theorem.
    
\section{Unruh Effect} \label{sec:Unruh-effect}
In this section we want to show an application in flat spacetime of the formalism we introduced, namely, the Unruh effect. The Unruh effect can be considered as the physical interpretation of the Bisognano-Wichmann theorem, another example of "prestabilized harmony" between mathematics and physics; in fact, the two were independently discovered in the same years (the original papers were both published in 1976). Historically, the Unruh effect has been the first known phenomenon in which the particle content of a theory was shown to be manifestly observer-dependent. The Unruh effect consists in the fact that, if a system is in Minkowski vacuum, an observer in constant acceleration detects a thermal bath of particles, where an inertial observer would observe none. In the algebraic framework this can be explained as the fact that to different observers are associated different states, which are invariant under different algebra automorphisms. In particular, an accelerated observer sees a restricted algebra of observables, because the causal development of its trajectory does not cover the whole Minkowski spacetime but only a portion of it, called \textit{Rindler wedge}, while the causal development of an inertial observer is the whole Minkowski spacetime, and therefore one naturally associates to an inertial observer the Minkowski vacuum. Crucially, in the Rindler wedge the orbits of boosts isometries are globally time-like and complete, so the Rindler wedge can be seen as a globally hyperbolic space on its own, with the "time translations" defined by Lorentz boosts. An accelerated observer in Minkowski is therefore associated to a quasifree state invariant under boosts symmetries, which are related to the action of the modular operator through the Bisognano-Wichmann theorem, and therefore such a state is KMS.

Let's see the actual computation. This is similar to what we did in subsection \ref{ssec:interplay-vacuum-KMS}. The Rindler spacetime is defined, as we said in the introduction to the Bisognano-Wichmann theorem \ref{th:bisognano-wichmann-theorem}, as $\mathscr W = \{ \mathbf x \in \mathbb R^{1,3} \ | \ x > \abs{t} \} $ equipped with Minkowski metric $\dd s^2 = - \dd t^2 + \dd x^2 + \dd y^2 + \dd z^2 $. The boost
\begin{align} \label{eq:boosts}
t \mapsto t \cosh \eta + x \sinh \eta \\
x \mapsto t \sinh \eta + x \sinh \eta
\end{align}
defines the orbits of accelerated observers with constant acceleration in the $x-t$ plane. The transformation \eqref{eq:boosts} is generated by the Killing vector $k = x \partial_t + t \partial_x$. This justifies a change of coordinates,
\begin{align}
t = \rho \sinh \eta \ \Rightarrow \dd t = \sinh \eta \dd \rho + \rho \cosh \eta \dd \eta \\
x = \rho \cosh \eta \ \Rightarrow  \dd x = \cosh \eta \dd \rho + \rho \sinh \eta \dd \eta
\end{align}
so that the Killing vector becomes $k = \partial_\eta$, while the metric takes the form
\begin{equation} \label{eq:rindler-metric}
\dd s^2 = - \rho^2 \dd \eta^2 + \dd \rho^2 + \dd y^2 + \dd z^2 \ .
\end{equation}
The advantage of this form is that the worldlines along $\eta$ are the trajectories of the boost transformation. If we make another coordinate change, $\rho =a^{-1} e^{a \xi} , \eta = a \tau$, which is related to the original Minkowski coordinates as
\begin{equation} \label{eq:change-of-variables-Lass}
\begin{split}
t &= \frac{e^{a \xi}}{a} \sinh a \tau \\
x &= \frac{e^{a\xi}}{a} \cosh a \tau
\end{split}
\end{equation}
for some $a \in \mathbb R_+$, the metric becomes
\begin{equation} \label{eq:lass-metric}
\dd s^2 = e^{2a \xi} (- \dd \tau^2 + \dd \xi^2) + \dd y^2 + \dd z^2 \ .
\end{equation}

The hypersurfaces $\tau = const.$ are Cauchy surfaces for the Rindler spacetime. Observers moving along $\tau$ are accelerated observers in Minkowski, as can be seen writing \eqref{eq:rindler-metric} with $\eta = a \tau$. As the metric does not depend on $\eta$, the transformation $\alpha_s : \tau \mapsto \tau + s$ is an isometry, corresponding to a boost in Minkowski, and as promised the geodesics tangent to $k = \partial_\eta$ are everywhere time-like, $\norm{k}^2 = - e^{2 a \xi}  $. In fact, this vector becomes null at infinity, which corresponds to a Rindler horizon for accelerated observers in Minkowski.

The algebraic quantization in Minkowski proceeds as we outlined in the previous sections. In particular, for a free, scalar, massless field there exists a quasifree state whose 2-point function is
\begin{equation}
W_2(x_1, x_2) = \frac{1}{2\pi^2}\frac{1}{( x_1 - x_2 + i\epsilon)^2} \ \forall x_1,x_2 \in \mathbb R^{1,3} \ ,
\end{equation}
where $\epsilon$ is any future-directed time-like vector.

We want to prove that this indeed is a KMS 2-point function with respect to translations in $\tau$, and we want to determine the temperature. For simplicity, we choose a coordinate system such that $y = z = 0$. We perform the change of variables \eqref{eq:change-of-variables-Lass}, obtaining
\begin{multline}
    W_2(x_1,x_2) = \frac{1}{2 \pi^2}  \frac{1}{- (t_1-t_2)^2 + (x_1-x_2)^2 - \epsilon^2} = \\
    = \frac{a^2}{2 \pi^2} \frac{e^{- 2 a \xi}}{- (\sinh a\tau_1 - \sinh a \tau_2)^2 + (\cosh a \tau_1 - \cosh a \tau_2)^2 - \epsilon^2} \ .
\end{multline}

Now we use the hyperbolic identities $\sinh x - \sinh y = 2 \cosh(\frac{x+y}{2}) \sinh(\frac{x-y}{2}) $, $ \cosh x - \cosh y = 2 \sinh(\frac{x+y}{2}) \sinh(\frac{x-y}{2}) $, and $\cosh^2 x - \sinh^2 x = 1 $, in order to rewrite the 2-point function as
\begin{equation}
    \begin{split}
&W_2(x_1,x_2) = \\
&=  \frac{a^2 e^{-2 a \xi}}{8 \pi^2} \bigg[ - \cosh^2 \frac{a}{2}(\tau_1+\tau_2) \sinh^2 \frac{a}{2}(\tau_1 - \tau_2) + \sinh^2 \frac{a}{2}(\tau_1+\tau_2) \sinh^2 \frac{a}{2}(\tau_1-\tau_2) - \epsilon^2 \bigg ]^{-1} \\
&= - \frac{a^2}{8 \pi^2} \frac{e^{-2 a \xi}}{\sinh^2 \frac{a}{2}(\tau_1 - \tau_2) - \epsilon^2} \ .
    \end{split}
\end{equation} 
On the other hand, the Jacobian of the transformation \eqref{eq:change-of-variables-Lass} is $J = e^{2a\xi} $ (no identification with the modular conjugation is intended), so the 2-point function becomes (again omitting the orthogonal variables $y$ and $z$)
\begin{equation}
W_2(\psi_1(\tau_1, \xi_1), \psi_2(\tau_2,\xi_2)) = \lim_{\norm{\epsilon} \to 0} - \frac{a^2}{8 \pi^2} \int \frac{\psi_1(\tau_1, \xi_1) \psi_2(\tau_2,\xi_2)}{\sinh^2 \frac{a}{2}(\tau_1 - \tau_2) - \epsilon^2} \dd \tau_1 \dd \tau_2 \dd \xi_1 \dd \xi_2 \ .
\end{equation}
But we have already met a close relative of this formula in subsection (\ref{ssec:interplay-vacuum-KMS}): it was the 2-point function for a KMS state, (\ref{eq:TPF-thermal}). The coefficient in the hyperbolic sine therefore gives the Unruh temperature,
\begin{equation}
\frac{\pi}{\beta} = \frac{a}{2} \Rightarrow \beta = \frac{2 \pi}{a} \ .
\end{equation}

This is the temperature Unruh introduced in his original work \cite{Unruh76}. As we said, the Unruh effect exhibits in concrete form some crucial facts we have already mentioned in different passages: i) The restriction of a vacuum state to a subregion is a thermal state with respect to the natural automorphism in that region; in this context it is applied because an accelerated observer sees a Rindler horizon, and therefore is causally disjoint from some part of the Minkowski space; ii) The Unruh effect is just a manifestation of the Bisognano-Wichmann theorem, in the fact that the modular operator acts as a boost; this mathematical theorem becomes an observable phenomenon if one sees that the boosts are the time-like isometry of an accelerated observers, and therefore the time translations and the modular flow coincides: the connection between temperature, modular flow and time evolution we saw before becomes an identification for the accelerated observer. The identification of time with modular time has profound implications for the nature of the arrow of time and its relation with the second law of thermodynamics, the so-called thermodynamic arrow of time. For a fascinating discussion, see \cite{Longo19b}. 

\chapter{Entropy from coherent states for black hole and cosmological horizons} \label{ch:new-results}

\section{Schwarzschild Geometry} \label{sec:Schw-geometry}

We start reviewing the structure of Schwarzschild spacetime and its maximal extension, the Kruskal spacetime, mainly to set the notation and remark on some facts which will prove useful later.

We start from the Schwarzschild metric in Schwarzschild coordinates:
\begin{equation} \label{eq:schwarzschild-metric}
\dd s^2 = - \bigg ( 1 - \frac{2m}{r} \bigg) \dd t^2 + \bigg ( 1 - \frac{2m}{r} \bigg)^{-1} \dd r^2 + r^2 \dd \Omega^2 \ ,
\end{equation}
where $m$ is a positive parameter and $\dd \Omega$ is the unit angular metric of a 2-sphere, $\dd \Omega^2 = \dd \theta^2 + \sin^2 \theta \dd \varphi^2 $. We will often denote $f(r) =  \bigg ( 1 - \frac{2m}{r} \bigg) $.

Schwarzschild solution has been the first exact solution of Einstein's equations, found in 1916 \cite{Schwarzschild16}; it describes the spacetime outside a spherical distribution of mass $m$. Karl Schwarzschild was a German physicist, and at the outbreak of World War I he served in the Prussian Army as a lieutenant. While at the Russian front, he wrote a letter to Einstein exposing for the first time his solution of Einstein's equations, adding, "As you see, the war treated me kindly enough, in spite of the heavy gunfire, to allow me to get away from it all and take this walk in the land of your ideas" \cite{Schwarzschild15}.

Although the first derivation assumed a static (independent on time) and spherically symmetric spacetime, it actually is the unique spherically symmetric solution of Einstein's equations in the vacuum, a fact known as the \textit{Birkhoff Theorem} (see e.g. \cite{Wald84}, chapter 6). As is known, Schwarzschild solution describes the exterior spacetime of a non-rotating, neutral spherical mass, which can be a black hole as well as a planet or a star, with an event horizon $\mathscr H$ located at $r_S = 2m$. If the mass distribution extends to a radius $R > r_S$, then the Schwarzschild solution is valid for $r > R$, since it holds in vacuum. If the mass is enclosed by the event horizon, we have a Schwarzschild black hole.

We denote the \textit{Schwarzschild exterior} $\mathscr E $ (that is \textit{not} the same symbol for the space of classical observables, $\mathcal E$; although they can be confused, context should make clear what we are talking about) the set of points with $t \in \mathbb R, \ \theta, \varphi \in \mathbb S^2, \ r > r_S $, and the \textit{black hole interior} $\mathscr B$ the set of points $t \in \mathbb R, \ \theta, \varphi \in \mathbb S^2, \ 0 < r < r_S  $. Schwarzschild coordinates have two singularities, at $r = r_S$ and $r = 0$. We will see, however, that it is possible to define a different set of coordinates extending in the black hole interior: therefore, the singularity at $r = r_S$ is not a physical singularity, but it is a \textit{coordinate singularity}, an artifact of the choice of coordinates, similar to the singularity at $ r = 0$ in polar coordinates. In the Schwarzschild metric, however, the singularity at $r = 0$ is a \textit{spacetime singularity}, since it cannot be removed by any change of coordinates. One can check that the scalars constructed from the Riemann tensor, invariant quantities by definition, diverge at $r = 0$, implying that the spacetime curvature becomes infinite at this point. 

We now show how to construct a Penrose diagram for a Schwarzschild black hole. The metric \eqref{eq:schwarzschild-metric} covers only the black hole exterior, $r > r_S$. In order to draw the Penrose diagram, we need to introduce a new set of coordinates, regular at the event horizon, so that they can cover both the interior and the exterior of the black hole. In other words, these coordinates should extend in the interval $r \in (0, \infty)$.

First, we introduce the \textit{tortoise coordinate}, $\dd r_*=  \bigg ( 1 - \frac{2m}{r} \bigg)^{-1} \dd r $, in terms of which the metric becomes
\begin{equation}
\dd s^2 = - \bigg( 1 - \frac{2m}{r} \bigg) ( \dd t - \dd r_*)(\dd t+ \dd r_*) + r^2 \dd \Omega^2 \ .
\end{equation}

It is natural to define the \textit{Eddington-Finkelstein coordinates},
\begin{align}
r_* = r + 2m \log \abs{\frac{r}{2m} - 1} \\
u = t - r_* \\
v = t + r_*
\end{align}
so that the metric becomes
\begin{equation} \label{eq:BH-lightlike-coordinates}
\dd s^2 = -  \bigg ( 1 - \frac{2m}{r} \bigg) \dd u \dd v + r^2 \dd \Omega^2 \ .
\end{equation}
Eddington-Finkelstein coordinates are adapted to light-rays; we call \textit{ingoing light-rays} those with $v = const.$, which move along $u$, and \textit{outgoing light-rays} those with $u = const$.

To find a global chart we finally introduce the Kruskal-Szekeres light-cone coordinates,
\begin{align} \label{Kruskal-light-like-coordinates}
U = - e^{-u/4m} \ , \ V = e^{v/4m} \ \text{in} \ \mathscr E \\
U = e^{-u/4m} \ , \ V = e^{v/4m} \ \text{in} \ \mathscr B
\end{align}
Clearly, $U,V$ are still light-like coordinates, in the sense that light-rays are described by trajectories with $U= const.$ or $V = const.$

From their definitions $\dd U = \frac{1}{4m} e^{-u/4m} \dd u $ and $ \dd V = \frac{1}{4m} e^{v/4m} $, which substituted in the metric give
\begin{equation}
\dd s^2 = - 16 m^2 \bigg(1 - \frac{2m}{r} \bigg) e^{(u-v)/4m} \dd U \dd V + r^2 \dd \Omega^2
\end{equation} 
but $u-v = - 2 r_*$, and using the definition of tortoise coordinate we find
\begin{equation} \label{eq:BH-Kruskal}
\dd s^2 = - \frac{32 m^3}{r} e^{- r/2m} \dd U \dd V + r^2 \dd \Omega^2 \ .
\end{equation} 

We see then that the metric in these coordinates is regular on the horizon, and therefore can cover both the interior and the exterior of the black hole. Therefore, as we've said, we see that the singularity at the event horizon was of coordinate type, and in fact it has been removed by a change of coordinates. On the other hand, the singularity at $r = 0$ remains even in these coordinates; clearly, this does not imply that it is a spacetime singularity, since in principle we could find an even better set of coordinates regular at $r = 0$. It has been proven, however, that the singularity at $r = 0$ is a true spacetime singularity, one which cannot be removed by a change of coordinates \cite{HawkingEllis73}.

As we discussed dealing with the Penrose diagram of Minkowski space, it is important to keep track of the interval of definition of the variables; so far, we have $t \in \mathbb R$ and $r \in \mathbb R_+$, and so $u,v \in \mathbb R$. In the black hole exterior, $U < 0$ and $V > 0$, while in the interior $U,V > 0$. The singularity at $r = 0$ gets mapped into the set $r_* = 0 \Rightarrow u = v \Rightarrow UV = 1$. The event horizon corresponds to $r_* \to - \infty$, which implies $u \to \infty \Rightarrow U = 0$, but the metric in Kruskal-Szekeres coordinates is regular on the horizon. Therefore we can extend the interval of definition of $U$ to the whole real line, to treat the exterior and the interior of the black hole on equal footing. 

The original Schwarzschild solution is covered by $U \in \mathbb R_+, \ V \in \mathbb R_+$, that is, the black hole exterior. The interior and the exterior of the black hole are covered by $U \in \mathbb R, \ V \in \mathbb R_+$, but we can see that the metric \eqref{eq:BH-Kruskal} is actually richer, admitting an extension to $V \in \mathbb R$ as well. The spacetime described by $U,V \in \mathbb R$ is called \textit{Kruskal spacetime} or Kruskal extension.

Now, to construct the Penrose diagram we introduce the squishified variables
\begin{align}
U' = \arctan U \\
V' = \arctan V
\end{align} 
Now, $U',V' \in (-\frac{\pi}{2}, \frac{\pi}{2}) $. Therefore, we have $\dd U = \cos^{-2} U' \dd U' $ and $\dd V = \cos^{-2} V' \dd V' $, so the metric becomes
\begin{equation}
\dd s^2 = \frac{1}{\cos^2 U' \cos^2 V'} \bigg[ - \frac{32m^3}{r} e^{-r/2m} \dd U' \dd V' + r^2 \cos^2 U' \cos^2 V' \dd \Omega^2 \bigg ] \ .
\end{equation}
We managed to find the unphysical metric, given in squared brackets. We can now extend the interval of definition of both $U',V'$ to $\mathbb R$, just as we did in the Minkowski case (section \ref{sec:Penrose-diagrams}), in order to find the full extension of the unphysical spacetime. As in the Minkowski case, the Kruskal extension is mapped in a finite region of the unphysical spacetime, and so we can find its boundaries. First of all, the gravitational singularity at $r = 0 $ is now mapped to
    \begin{multline}
UV = 1 \Rightarrow \tan U' \tan V' = 1 \Rightarrow \sin U' \sin V' - \cos U' \cos V' = 0 \Rightarrow \cos( U' + V' ) = 0 \\ \Rightarrow U' + V' = \pm \frac{\pi}{2} \ ,
    \end{multline}
which, if we draw a diagram with the axis $U' = 0, \ V' = 0$ at $45^{\circ}$ degrees, are two horizontal lines. The event horizon is at $U = 0, \ V > 0$, which means the half-positive axis $U' = 0$. Infinite values of $U,V$ are mapped at $U' = - \frac{\pi}{2}$ and $ V' = \frac{\pi}{2}$; however, they are not extended for all values of $V'$ respectively $U'$, but are limited by the two singularity segments. Moreover, if we take into account negative values of $V$ as well, we have two more segments which represent infinity, at $U' = \frac{\pi}{2}$ and $V' = \frac{\pi}{2}$. The result is a tilted square, as in Minkowski case, but cut at the gravitational singularity, obtaining an hexagon.

\begin{figure}
\centering
\includegraphics[scale=0.5]{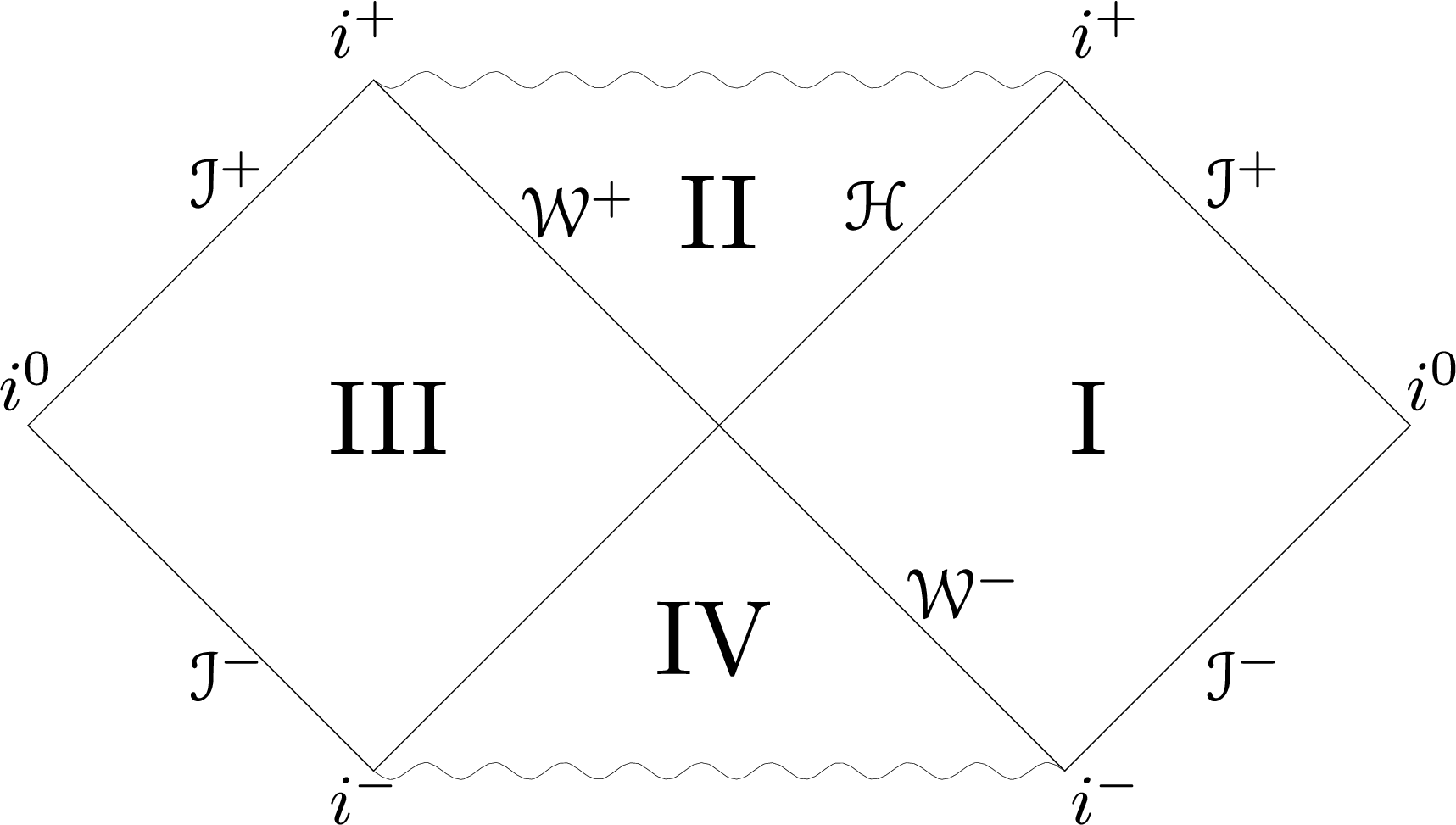}
     \caption{Penrose diagram for the Kruskal extension.  The diagonal should be read with the axis $U = 0$ on the diagonal from bottom-left to top-right, and the axis $V = 0$ from top-left to bottom-right. The original Schwarzschild solution corresponds to region I. The two wavy lines correspond to the spacetime singularity, $r = 0$. The spacetime we consider is the one realisable in a gravitational collapse, the union of region I and II.}
\end{figure}

Now, in the conformal embedding of Kruskal spacetime we can identify the past and future null infinity, $\mathscr I^\pm$, given by the segments $U',V' = \pm \frac{\pi}{2}$ the spatial infinity, at $i^0 = (U' = -\frac{\pi}{2}, \ V' = \frac{\pi}{2})$ and $ ( U' = \frac{\pi}{2}, \ V' = - \frac{\pi}{2}) $, time-like future and past infinity, at $i^\pm = (U' = 0, \ V' = \frac{\pi}{2}) $ and $ (U' = \frac{\pi}{2}, \ V' = 0) $, and the two singularity segments at $U' + V' = \pm \frac{\pi}{2}$. However, Kruskal spacetime exhibits the additional structure of horizons: we have the event horizon, defined as $\mathscr H = \{ x \in \mathbb R^{1,3} \ | \ U' = 0, \ V' > 0 \} = \{  x \in \mathbb R^{1,3} \ | \ U = 0, \ V > 0 \}$, and the past horizon, $\mathscr W$, defined as $ V' = V = 0, \ U' \in [- \frac{\pi}{2}, \ \frac{\pi}{2}] \Rightarrow U \in \mathbb R $; they both correspond to $r = 2m$. For future reference we split the past horizon in two disjoint segments, $\mathscr W^+ = \{ x \in \mathscr W \ | \ U \geq 0 \} $ and $\mathscr W^- = \{ x \in \mathscr W \ | \ U < 0 \} $; $\mathscr W^-$ is called the \textit{white hole horizon}. The values of the different variables on the boundaries are summed in table \ref{table:boundaries}.

Of the four different regions of Kruskal spacetime, we are interested in the two regions I, II only, the interior and the exterior of the black hole. They are bounded by the future singularity, the boundaries at infinity on the right half of the Penrose diagram, and the past horizon. These are the regions which can be realised in the gravitational collapse of a compact object: the Schwarzschild solution is valid only outside the matter in collapse, so that region I can be realised in the outside of a spherical compact object, and region II is obtained when the matter falls beyond the Schwarzschild radius. There is no known gravitational process which can realise region III and IV. They are, rather, the maximal extension of a Schwarzschild black hole which always existed and always will, an object known as an \textit{eternal black hole}. We will denote the union of region I, II and their boundaries $ \mathscr I^\pm$, $i^0$, $i^\pm$ and $\mathscr W$ with $\mathscr M$, since it will be the spacetime we will interested in in the following sections, and we will call it the Schwarzschild spacetime or Schwarzschild black hole; although this term is usually referred to region I only, we will consider both the interior and the exterior of the black hole.

\begin{table} \label{table:boundaries}
\centering
\begin{tabular}{|| c c c c c c c ||}
\hline Boundary & $U$ & $V$ & $U'$ & $V'$ & u & v \\
\hline Event Horizon $\mathscr H$ & 0 & $\mathbb R_+ $ & 0 & $[0, \frac{\pi}{2}]$ & & \\
\hline Past Horizon  $\mathscr W$ & $\mathbb R$ & 0 & $ [- \frac{\pi}{2}, \frac{\pi}{2}] $ & 0 & & \\
\hline Past Null Infinity $\mathscr I^-$ & $-\infty $ & $ \mathbb R_+ $ & $ - \frac{\pi}{2} $ & $[0, \frac{\pi}{2}]$ & $-\infty$ & $\mathbb R$ \\
\hline Future Null Infinity $\mathscr I^+$ & $\mathbb R_- $ & $\infty$ & $[- \frac{\pi}{2}, 0] $ & $ \frac{\pi}{2} $ & $\mathbb R$ & $\infty$ \\
\hline Future Time-like Infinity $i^+$ & 0, $\infty$ & $\infty$, 0 & 0, $\frac{\pi}{2}$ & $\frac{\pi}{2}$, 0 & $t = \infty$ & $r \in \mathbb R$ \\
\hline Past Time-Like Infinity $i^-$ & 0, $- \infty$ & $- \infty$, 0 & 0, $- \frac{\pi}{2}$ & $ - \frac{\pi}{2}$, 0 & $t = - \infty$ & $r \in \mathbb R$ \\
\hline
\end{tabular}
\caption{Relevant boundaries in the Schwarzschild spacetime and their location in the different variables}
\end{table}

\subsection{Event Horizon and Symmetries}

In Kruskal coordinates, the event horizon $r = r_S = 2m$ is a hypersurface at $U=0$; in the Penrose diagram it is a segment at $45^{\circ}$. This means that ingoing light-rays toward the black hole can cross it (and, in fact, cannot avoid to) but outgoing light-rays from the black hole interior toward the exterior cannot escape. This one-way property of the horizon means that an observer in the interior of the black hole cannot send signals to the outside world: the interior has no way to influence the exterior. This is what is called a \textit{causal boundary}. Moreover, any future cone inside the event horizon terminates at the singularity: observers inside the event horizon are doomed to encounter the singularity in a finite amount of proper time.

The event horizon is the defining property of a black hole, according to a large segment of literature. It is, as we said, the boundary between the causal paths which can escape to infinity, and those which will remain confined in the black hole interior. The rigorous definition of a black hole is therefore \textit{the boundary of the chronological past of the future}:
\begin{equation}
\mathscr B = \mathscr M - I^-(\mathscr I^+) \ .
\end{equation}
(We saw the definition of chronological past $I^-$ in section \ref{sec:globally-hyperbolic-spaces}). This definition is the most used to study the mathematical structure of black hole, but for real black holes it can be controversial: for once, it requires the spacetime to be asymptotically flat in order to locate an event horizon. Since our Universe may be not asymptotically flat, and nonetheless black holes have been observed, an alternative definition for realistic black hole is needed. In subsection \ref{ssec:AH-first-law} we will discuss other problems related to the event horizon and a different definition for a black hole horizon.

It is possible to identify the symmetries of a Schwarzschild black hole by simple inspection of the metric. For example, the metric is independent on time $t$; this means that it is invariant under translation along the $t$ direction:
\begin{equation}
\mathcal L_\xi g = 0 \ .
\end{equation}
$\mathcal L_\xi$ is the Lie derivative along $t$, that is, $\xi = \partial_t$. In general, for any isometry of the metric, that is, any vanishing Lie derivative along some parameter $t$, it is possible to find an equation for the generator of the $t$ translations, that is, the vector field $\xi = \partial_t$, if the connection is metric (which is always our case). The Lie derivative of the metric along $t$ is
\begin{equation}
\xi^\rho \partial_\rho g_{\mu \nu} + g_{\mu \rho} \partial_\nu \xi^\rho + g_{\rho \nu} \partial_\mu \xi^\rho = 0 \ .
\end{equation}
On the other hand, the covariant derivative of the metric vanishes, $\nabla_\rho g_{\mu \nu} = \partial_\rho g_{\mu \nu} - g_{\lambda \mu} \Gamma^\lambda_{\rho \nu} - g_{\lambda \nu} \Gamma^\lambda_{\rho \mu} = 0 $, while the covariant derivative of the Killing field is $\nabla_\mu \xi^\nu = \partial_\mu \xi^\nu + \Gamma^\nu_{\mu \rho} \xi^\rho$, so
\begin{equation}
0 = \xi^\rho \partial_\rho g_{\mu \nu} + g_{\mu \rho} \nabla_\nu \xi^\rho - g_{\mu \rho} \Gamma^\rho_{\mu \alpha} \xi^\alpha + g_{\rho \nu} \nabla_\mu \xi^\rho - g_{\rho \nu} \Gamma^\rho_{\mu \alpha} \xi^\alpha = \nabla_\nu \xi_\mu + \nabla_\mu \xi_\nu \ .
\end{equation}
This is called \textit{Killing equation}; a vector satisfying the Killing equation is a Killing field, and it is the generator of an isometry.

The Schwarzschild spacetime admits a Killing field
\begin{equation}
\xi_t = \partial_t = \partial_u + \partial_v = - \frac{U}{2 r_S}\partial_U + \frac{V}{2r_S}\partial_V \ .
\end{equation}
This vector is time-like in the Schwarzschild exterior and thus defines a preferred flow of time, but becomes space-like in the black hole interior. The set of points where $\norm{k} = 0$ is called \textit{Killing horizon}; in the Schwarzschild case, the event and Killing horizons coincide. The existence of a Killing vector field, such that $\norm{\xi}^2 \to -1 $ at infinity, means that the spacetime is \textit{stationary}. Moreover, the Schwarzschild spacetime is \textit{static}, because the Killing field is hypersurface orthogonal; this means that it is possible to foliate the spacetime in 3-hypersurfaces defined by $t = const.$

The Schwarzschild black hole has also spherical symmetry, and therefore there are additional Killing fields related to rotations:
\begin{enumerate}
\item $\xi_\varphi = \partial_\varphi$;
\item $\xi_x = - \sin \varphi \partial_\theta - \cot \theta \cos \varphi \partial_\phi $;
\item $ \xi_y = \cos \varphi \partial_\theta - \cot \theta \sin \varphi \partial_\phi $ 
\end{enumerate}
Any isometry of the metric is associated to a conserved current. In fact, for any Killing vector $\xi$, one has
\begin{equation}
\nabla_\mu ( \xi_\nu T^{\mu \nu} ) = \xi_\nu \nabla_\mu T^{\mu \nu} + T^{\mu \nu} \nabla_\mu \xi_\nu = 0 \ .
\end{equation}
The first term on the right-hand side vanishes because the stress-energy tensor is conserved by construction, the second term vanishes because it is a symmetric tensor times an antisymmetric tensor, since the Killing equation imposes$\nabla_{(\mu} \xi_{\nu)} = 0$ , that is, $\nabla_\mu \xi_\nu$ is a tensor with no symmetric part. Therefore, $J^\mu = \xi_\nu T^{\mu \nu} $ is a conserved current. We can associate the conserved current generated by the time isometry with the energy density, and that associated with the angular isometries with the different components of the angular momentum.

\section{Black Hole Mechanics} \label{sec:BH-mechanics}

The story of Black Hole Thermodynamics started in the mid-'70s, when the interest in General Relativity renewed and researchers in both American and English universities began studying the simplest solutions of Einstein's equations, black holes. The new researches were mainly driven by Dennis Sciama, in Cambridge, and John A. Wheeler, at the Princeton University, who popularised the very term "black hole" and defended it against the scandalised protests of American academia \cite{Herdeiro18}. Few years later, when the controversial name for black holes settled, he and his student Jacob Bekenstein proposed that black holes can be described by three parameters only, their mass, electrical charge and angular momentum, being featureless objects for the rest: in their words, "black holes have no hair".

However, a tension between the no hair theorem and the laws of thermodynamics grew, eventually leading to the breakthrough paper by Hawking on black hole radiation. In fact, Hawking himself, together with Bardeen and Carter \cite{Bardeen73} published a paper presenting a close analogy between the laws of black hole mechanics and the four laws of thermodynamics. Moreover, the gedanken experiments of Bekenstein, of which we talked in the introduction, showed that black holes must carry some amount of entropy, in contradiction with the no-hair theorem. On the other hand, there seemed to be no way in which the classical theory could explain a non-zero temperature for black holes, since for their very definition black holes cannot radiate while they can always absorb radiation from any source at finite temperature. In order to find the origin of black hole radiation, and therefore of their temperature, one needs to take into account the propagation of quantum fields on a black hole background.

Before discussing the Hawking temperature of a Schwarzschild black hole, we briefly review the classical laws of black hole mechanics. It is known that the laws of thermodynamics were discovered in a very different order from that in which they are usually presented: the first established became the second law, the first law was the second to be discovered, and the zeroth law has been the last. Maybe to continue the analogy, Bardeen, Carter, and Hawking \cite{Bardeen73}, enunciated the four laws in an idiosyncratic order. We will follow here the order of their presentation.

\paragraph{Second law}

The second law is arguably the most important one for the thermodynamic analogy. It was first proved \cite{HawkingEllis73} by Hawking and Ellis in 1973 and then disproved \cite{Hawking74} (or, rather, it was found that its failure was even more interesting than its statement) by Hawking himself in 1974.

The second law, also called \textit{Area Theorem}, states that the area of a black hole in a globally hyperbolic spacetime with classical matter present cannot decrease with time. The important adjective here is \textit{classical}: in fact, the whole point is that quantum matter can violate any local energy condition, and in the presence of quantum fields the area of a black hole can in fact decrease, leading to the idea of \textit{black hole evaporation}.

The two key ingredients of the area theorem are the Focussing theorem, which we discussed in section \ref{sec:geodesic-congruences}, and an observation by Penrose on the fact that the null generators of a black hole event horizon can never form caustics, that are, points in which the generators meet. At such points, the expansion of the generators' congruence diverge. If such points cannot exist, then the expansion cannot be negative, or the expansion would diverge in a finite amount of the affine parameter, as stated by the Focussing theorem. Since the expansion is the logarithmic derivative of the event horizon area, this in turn implies that the horizon area cannot decrease with time. The full technical proof of the theorem can be found in \cite{HawkingEllis73}.

Although the area theorem is violated by Hawking radiation, it has been the first hint that black holes can carry entropy. In fact, two things never decrease in Nature: entropy, and the area of black holes (and \textit{time}, which points to some profound, mysterious connection - but this is way outside our scope). Bekenstein's gedanken experiments proved that a decrease in the entropy outside a black hole is always compensated by an increase in the area of the horizon. Therefore, to save the second law of thermodynamics he proposed to identify the entropy of a black hole with some multiple of its area. This way, any entropy decrease in the black exterior would be compensated by an increase in the black hole entropy, and the entropy of the Universe would still be increasing.

Interestingly, it was the violation of the area theorem which eventually justified the identification of the horizon area with a measure of black hole entropy. It is clear, that a classical black hole is rigorously at zero temperature, and therefore, from the Planck-Nerst version of the third law of thermodynamics, it should also have zero entropy. Indeed, the no-hair theorem states that black holes in GR are objects with no internal structure, similar in this regard to elementary particles. There seems to be no way in which a black hole can radiate, but in fact, when Hawking took into account the effects of quantum fields on a black hole background, he discovered that black holes emit a steady flux of thermal particles toward future infinity at a finite temperature, the so-called \textit{Hawking temperature}. Therefore, it is indeed possible to associate an entropy to the black hole, and via a Legendre transformation Hawking confirmed Bekenstein's intuitions finding that the correct proportionality coefficient between the entropy and the area of a black hole is one-fourth.

\paragraph{Zeroth law}

We have seen that, in the case of a Schwarzschild black hole, the event horizon is also a Killing horizon. Actually, there are two results, generally known as \textit{rigidity theorems}, stating that an event horizon \textit{must be} a Killing horizon too. We are interested here in the theorem by Carter \cite{Carter73}. In the case of a static black hole, the theorem states that the vector $\xi_t = \partial_t$ must be normal to the event horizon. If the black hole is not static, the result can be generalised in the following way. Consider a stationary, axisymmetric black hole, that is, a black hole which admits two isometries $\xi_t = \partial_t$ (stationarity) and $\xi_\varphi = \partial_\varphi $ (symmetry under rotations along an axis). Then, there exists a linear combination of the two,
\begin{equation}
\zeta = \partial_t + \Omega_{\mathscr H} \partial_\varphi
\end{equation}
which is normal to the event horizon. However, Hawking showed \cite{Hawking71} that such a combination is also null on the horizon, and therefore (since $\zeta$ is a Killing field) the event horizon of a stationary, axisymmetric black hole coincides with a Killing horizon, and the vector $\zeta$ is therefore tangent to the generators of the horizon.

Now, consider a black hole event horizon, and its normal Killing vector $\zeta$. Since $ \zeta^\beta \zeta_\beta = 0$ along the horizon, its gradient must be normal to the horizon, and therefore it is proportional to $\zeta$:
\begin{equation} \label{eq:surface-gravity-1}
\nabla^\alpha \zeta^\beta \zeta_\beta = - 2 \kappa \zeta^\alpha \ .
\end{equation}
The function $\kappa$ is called \textit{surface gravity}. It can be defined in an equivalent way, since $\zeta$ is tangent to the generators of the horizon, as the inaffinity function in the geodesic equation for $\zeta$. From \eqref{eq:surface-gravity-1} and the Killing equation, we see that
\begin{equation}
\nabla^\alpha \zeta^\beta \zeta_\beta = 2 \zeta^\beta \nabla^\alpha \zeta_\beta = - 2 \zeta^\beta \nabla_\beta \zeta^\alpha = 2 \kappa \zeta^\alpha \ ,
\end{equation} 
from which we conclude that
\begin{equation} \label{eq:surface-gravity-2}
\zeta^\alpha \nabla_\alpha \zeta^\beta = \kappa \zeta^\beta \ .
\end{equation}
We can find an explicit expression for the surface gravity by first noting that, since the $\zeta$ is normal to the event horizon, from the Frobenius theorem (see section \ref{sec:geodesic-congruences}) we know that
\begin{equation}
\omega_{\alpha \beta} = \nabla_{[\alpha} \zeta_{\beta]} = 0 \ ,
\end{equation}
and it immediately follows that $\zeta_{[\gamma} \nabla_{\alpha} \zeta_{\beta]} = 0 $. This is an equation for a sum of six terms; however, using the Killing equation we get
\begin{equation}
\zeta_\gamma \nabla_\alpha \zeta_\beta + \zeta_\alpha \nabla_\beta \zeta_\gamma + \zeta_\beta \nabla_\gamma \zeta_\alpha = 0 \ .
\end{equation}
Contracting with $\nabla^\alpha \zeta^\beta$, we find
\begin{equation}
\begin{split}
\nabla^\alpha \zeta^\beta \nabla_\alpha \zeta_\beta \zeta_\gamma &= - \zeta_\alpha \nabla^\alpha \zeta^\beta \nabla_\beta \zeta_\gamma - \zeta_\beta \nabla^\alpha \zeta^\beta \nabla_\gamma \zeta_\alpha \\
&= - \kappa \zeta^\beta \nabla_\beta \zeta_\gamma + \kappa \zeta^\alpha \nabla_\gamma \zeta_\alpha \\
&= - 2 \kappa^2 \zeta_\gamma
\end{split}
\end{equation}
and so
\begin{equation}
\kappa^2 = -\frac{1}{2} \nabla^\alpha \zeta^\beta \nabla_\alpha \zeta_\beta \ .
\end{equation}

 The zeroth law states that \textit{the surface gravity $\kappa$ is uniform on the event horizon of a stationary black hole}. It follows from proving first that the surface gravity is constant along each generator, and then that it does not change from a generator to another. Details of the proof can be found in \cite{Poisson09}. This law is in analogy with the zeroth law of thermodynamics, and can be considered as a characterisation of gravitational equilibrium, as the constancy of temperature in a system characterises thermal equilibrium in thermodynamics.
 
As a final remark, we compute the surface gravity for a Schwarzschild black hole. In this case, the time-like Killing vector tangent to the horizon is simply $\xi_t = \frac{V}{2 r_S} \partial_V$. On the event horizon, equation \eqref{eq:surface-gravity-1} becomes, for the only non-vanishing component $\beta = V$,
\begin{equation}
\xi^\alpha_t \nabla_\alpha \xi^V_t = \partial_V \xi^V_t + \Gamma^V_{V V} = \frac{1}{2 r_S} = \frac{1}{4m} \ .
\end{equation}
In fact, it can be seen (see, for a completely analogous computation on the past horizon, subsection \ref{ssec:quantization-in-Schwarzschild}) that $\Gamma^V_{V V} = 0$ on the event horizon. Therefore, in the case of a Schwarzschild black hole it is immediate to verify the zeroth law. Moreover, this computation closes a loop we opened in subsection \ref{ssec:interplay-vacuum-KMS}, in which we performed a change of coordinates $V = e^{\kappa v}$: we now see that the surface gravity is the same parameter of the Kruskal-Szekeres change of coordinates, $V = e^{v/4m}$:
\begin{equation}
\kappa = \frac{1}{4m} \ .
\end{equation}

\paragraph{First law}
Consider a stationary black hole, with Killing vector $\xi_t = \partial_t$. We have seen that a stationary black hole is either static or \textit{axisymmetric}, that is, it is symmetric with respect to rotations along an axis and thus admits a Killing vector $\xi_\varphi = \partial_\varphi$. We consider the latter case for generality. In this case, the linear combination of $\partial_t$ and $\partial_\varphi$,
\begin{equation}
\zeta = \partial_t + \Omega_{\mathscr H} \partial_\varphi
\end{equation} 
is null on the horizon. Therefore, the event horizon is a Killing horizon and $\xi$ is the vector tangent to its null generators. Such a black hole is described by two parameters, its mass $M$ and its angular momentum, $J$. Even without a detailed description of the metric we can find how these two parameters are modified if some matter falls through the horizon. We consider a quasi-stationary process in which the black hole parameters change by an infinitesimal amount before reaching a new stationary state. If $T_{\alpha \beta}$ is the stress-energy tensor of the infalling matter, the change in mass and in angular momentum of the black hole will be
\begin{align}
\delta M = - \int_{\mathscr H} T_{\alpha \beta} \xi_t^\alpha \ \dd \Sigma_{\mathscr H}^\beta \\
\delta J = \int_{\mathscr H} T_{\alpha \beta} \xi_\varphi^\alpha \ \dd \Sigma_{\mathscr H}^\beta \ .
\end{align}
The minus sign in the variation of the mass comes because $\xi_t$ is a time-like vector, and we need the minus to get the mass at a later time minus the mass at an earlier time.

The surface element of the event horizon is given by $\dd \Sigma^\alpha_{\mathscr H} = - \zeta^\alpha \dd t \dd \mathbb S_2 $, where $\dd \mathbb S_2$ is the metric of a 2-sphere. Again, the minus sign comes so that the surface integral is future-directed. Note, however, that we are not using the affine parameter along the horizon. We have
\begin{equation}
    \begin{split}
\delta M - \Omega_{\mathscr H} \delta J &= \int_{\mathscr H} T_{\alpha \beta} ( \xi_t^\alpha + \Omega_{\mathscr H} \xi_\varphi^\alpha ) \zeta^\beta \ \dd t \dd \mathbb S_2 = \\
&=  \int_{\mathscr H} T_{\alpha \beta} \zeta^\alpha \zeta^\beta \dd t \dd \mathbb S_2 \ .
\end{split}
\end{equation}
The integrand can be rewritten using the Raychaudhuri equation along the generators of the horizon. Since we are interested in the first order variation of the parameters, we can neglect the quadratic terms in the equation; however, we do not assume that the generators are affinely parametrised, but rather that there exists a nonvanishing surface gravity. Therefore, we write the Raychaudhuri equation \eqref{eq:Raychaudhuri-inaffine} for the null generators of the horizon as
\begin{equation}
\dv{\theta}{t} = \kappa \theta - 8 \pi T_{\alpha \beta} \zeta^\alpha \zeta^\beta \ .
\end{equation}
Substituting in the integral, we find
\begin{equation}
\begin{split}
\delta M - \Omega_{\mathscr H} \delta J &= - \int_{\mathscr H} \dv{\theta}{t} - \kappa \theta \ \dd t \dd \mathbb S_2 \\
&= - \int_{\mathbb S_2} \eval{\theta}^{+\infty}_{-\infty} + \frac{\kappa}{8 \pi} \int_V \int_{\mathscr H(t)} \dv{\sqrt \sigma}{t} \dd t \dd^2 \theta \ .
\end{split}
\end{equation}
To rewrite the second integral, we used the interpretation of the expansion scalar as the logarithmic derivative of the square root of the determinant of the cross-sectional area, which in this case is a 2-sphere with angular variables $\theta^A$, as we proved in subsection \ref{ssec:geometric-meaning-expansion}. At any fixed $t$, the event horizon is a 2-sphere, which we denoted $\mathscr H(t)$. Now, assuming that the black hole is stationary both before and after the perturbation, $\theta (t = \pm \infty) = 0$ and the first integral vanishes as well. The second integral instead gives
\begin{equation} \label{I-law-EH}
\delta M - \Omega_{\mathscr H} \delta J = \frac{\kappa}{8 \pi} \int_{\mathscr H(t)} \eval{\sqrt{\sigma}}^\infty_{- \infty} \dd^2 \theta = \frac{\kappa}{8 \pi} \delta A \ ,
\end{equation}
where $\delta A$ is the variation in the area before and after the perturbation. Finally, we find the first law of black hole mechanics:
\begin{equation}
\delta M = \frac{\kappa}{8 \pi} \delta A + \Omega_{\mathscr H} \delta J \ .
\end{equation}
We can see the close resemblance of this equation with the first law of thermodynamics for a system with internal energy $E$, at temperature $T$, with pressure and volume $P$ and $V$:
\begin{equation}
\dd E = T \dd S + P \dd V \ .
\end{equation}
In fact, it is natural to identify the mass of the black hole with its internal energy, $M= E$. This again suggests that the area of the black hole is some multiple of its entropy, $S \sim A$, and the surface area is proportional to the temperature of the black hole, $\kappa \sim T$; clearly, from this equation there is no way to identify the correct proportionality coefficients.

Our derivation of the first law of black hole mechanics does not rely on the particular black hole model, but it is based on General Relativity and the first order approximation. Remarkably, a version of the first law holds for any field equation derived from a diffeomorphism covariant Lagrangian (\cite{Iyer94}, and see \cite{Wald99} for a review).

\paragraph{Third law}
The third law of black hole mechanics states that \textit{it is not possible to reduce the surface gravity of a black hole to zero in a finite amount of advanced time in any continuous process in which the stress-eenrgy tensor remains bounded}. A derivation of this law was provided by W. Israel in 1986 \cite{Israel86}, and we will not reproduce it here.
The third law of black hole mechanics is in analogy with the Nernst version of the third law of thermodynamics, which states that it is not possible to reduce the entropy of a system to zero in any continuous process, no matter how idealised. However, the Planck-Nernst version of the thermodynamics third law, which states that the entropy of an object approaches zero as the temperature approaches zero, is violated in black hole mechanics, breaking the analogy between black holes and thermodynamics. In fact, it is possible to choose the parameters of a Kerr black hole (the mass and its angular momentum) so that its surface gravity vanishes while the area remains finite. However, the Planck-Nernst statement is violated by ordinary quantum systems too \cite{Wald97}, and therefore it should be considered not a fundamental law but rather a property of common studied materials.

\hfill

The laws of black hole mechanics shows a remarkable resemblance with the laws of thermodynamics, a resemblance that appears rather mysterious: black hole laws are statements in differential geometry, derived from classical field equations and from the properties of the spacetime; the laws of thermodynamics arise as the statistical limit of the microscopic nature of the constituents of matter. Indeed, in General Relativity this is nothing more than an analogy: in fact, the temperature of any black hole is absolute zero, since, as we said, a black hole is in thermal equilibrium only in vacuum, at zero temperature. If we take into account quantum effects, however, things change dramatically, and a black hole becomes an object with a finite, well-defined temperature (this is the famous \textit{Hawking effect}), which in vacuum radiates and eventually evaporates. From a mathematical point of view, the main reason is that quantum matter violates almost any local energy condition, and in particular the weak energy condition, which is a crucial assumption in many theorems on black holes. From a physical point of view, the curvature in spacetime provides an enhancement of the vacuum energy fluctuations, bringing the virtual particle-antiparticle pairs into life at the expense of black hole energy.

We will not provide a thorough discussion of the quantum effects in the presence of a black hole, since they would easily take an entire book. Rather, we will show that the quantization of a scalar field in the exterior of a Schwarzschild black hole naturally give rise to thermal properties for the modes incoming from the white hole horizon.

\section{Hawking Temperature} \label{sec:Hawking-temperature}

\subsection{Null Quantization in Schwarzschild and Unruh State} \label{ssec:quantization-in-Schwarzschild}

We apply the quantization procedure we discussed in chapter \ref{ch:QFTCS} to the case of a scalar, neutral, free quantum field propagating on a Schwarzschild black hole background, the spacetime $\mathscr M$ formed by the union of regions I and II in the Kruskal extension we introduced in section \ref{sec:Schw-geometry}. As usual, we start with a classical solution $\nabla_a \nabla^a \psi = 0 $ of the massless Klein-Gordon equation with minimal coupling $\xi = 0$, with smooth, compactly supported Cauchy data on some hypersurface $\Sigma $. Now, we want to deform this hypersurface to the surface $\mathscr W \cup \mathscr I^-$, that is, on the union of the past horizon and the conformal past infinity. This is the most physically sensible Cauchy surface since we are giving the initial data at past infinity and we study the evolution of the field as it propagates through the spacetime. However, it might occur that the symplectic form defined on $\Sigma$ in this limit diverges, since we are computing a product of the field and its derivatives at infinity, or that we must take into account some contribution from $i^-$. In \cite{DMP11} it was shown that such a construction is possible. In fact, if we take the restriction of the scalar field to $\mathscr W$, and the restriction of $\tilde \psi = r \psi$ to $\mathscr I^-$, one can prove that it is possible to write the symplectic form given in \eqref{eq:symplectic-form} using data on $\mathscr I^- \cup \mathscr W$:
\begin{equation} \label{eq:symplectic-form-Schwarzschild}
w(\psi_1, \psi_2) = \int_{\mathscr W} ( \psi_2 \nabla_\mu \psi_1 - \psi_1 \nabla_\mu \psi_2 ) n^\mu_{\mathscr W} \ \dd U \dd \mathbb S_2 +
 \int_{\mathscr I^-} ( \tilde \psi_2 \nabla_\mu \tilde \psi_1 - \tilde \psi_1 \nabla_\mu \tilde \psi_2 ) n^\mu_{\mathscr I^-} \ \dd v \dd \mathbb S_2
\end{equation}
for the classical field propagating on the black hole spacetime. The causal development of $\mathscr I^- \cup \mathscr W$ covers the Schwarzschild spacetime $\mathscr M$, the interior and the exterior of the black hole with its boundaries, and so we can define a quantum field theory in this region. As we have seen, the symplectic form is related to the causal propagator $E$, and we can construct the Weyl C*-algebra $\mathcal A(\mathscr M) $ with relations \eqref{eq:Weyl-relations}.

In chapter \ref{ch:QFTCS} we have seen two possible approaches to the quantization of the scalar field, the one based on quotient of the space of test functions, $(\mathcal E, E)$, and the one based on the space of classical solutions, $(\mathsf{Sol}, w)$, and we have seen that they are related via the computation \eqref{eq:symplectic-form-causal-propagator-relation}. Here, we adopt the approach to smear the generators of the Weyl algebra with the set of smooth, compactly supported initial data for the Klein-Gordon equation, given on $\mathscr I^- \cup \mathscr W $. We recall that the Weyl generators smeared with test functions are related to the smeared field via \eqref{eq:Weyl-operators-fields}. If we introduce a classical solution $\psi_f = E f$, where $E$ is the causal propagator and $f$ a test function, \eqref{eq:Weyl-operators-fields} is equivalent to
\begin{equation}
W(\psi) = e^{i w(\phi, \psi_f)} \ .
\end{equation}
We will denote $W(\psi)$ the generators of the Weyl algebra, leaving implicit the restriction of $\psi$ to $\mathscr I^- \cup \mathscr W $.

Now, we need to select a suitable state on the Schwarzschild space. We choose a quasifree, Hadamard state, since as we discussed in \ref{ssec:hadamard-states} this is the class of the physically sensible states in which the fluctuations of the expectation value of normal ordered fields are always finite, as is the case for the Minkowski vacuum. Since the state is quasifree, we can uniquely define it giving its 2-point functions; however, we do not define an explicit expression for the 2-point function in $\mathscr M$, but rather we give its restriction to the initial data surface $\mathscr I^- \cup \mathscr W $; it can be proved that this is sufficient to uniquely select a quasifree state in $\mathscr M$. The observation in this case is that we can see \textit{a posteriori} that the symplectic form \eqref{eq:symplectic-form-Schwarzschild} splits into two symplectic forms for the two surfaces $\mathscr W$ and $\mathscr I^-$, which we can respectively denote $w_{\mathscr W}$ and $w_{\mathscr I^-}$. Therefore, we can quantize the field as two copies of the Weyl algebra based on the two symplectic spaces $\mathsf{Sol}(D^+(\mathscr W), w_{\mathscr W})$ and $\mathsf{Sol}(D^+(\mathscr I^-), w_{\mathscr I^-})$. Their restriction on the respective Cauchy surface, $\mathscr W $ and $\mathscr I^- $, gives a Weyl algebra for a theory localised on the boundary. Such a restriction can be interpreted as a symplectomorphism between algebras. If we repeat the argument in reverse, we can construct a Weyl algebra on the boundaries $\mathscr W$ and $\mathscr I^-$, select a boundary state, and then define a bulk state as the pull-back of the symplectomorphism from the bulk to the boundary algebras. The details for such a construction are given in \cite{DMP11}. For our purposes, it is sufficient to say that the restriction of the 2-point function of a quasifree, Hadamard state on the surface $\mathscr I^- \cup \mathscr W$ is sufficient to define the 2-point function on $\mathscr M$. Moreover, this 2-point function will be written as a sum of two surface integrals, over $\mathscr I^-$ and over $\mathscr W$: thus, we can define a \textit{class} of states giving the 2-point funcion expliclty for one surface only, and leaving the other integral arbitrary.

This being said, we proceed defining the state which we will use throughout the chapter, the \textit{Unruh state}. The Unruh state has been originally proposed by W. Unruh \cite{Unruh76}, via the mode decomposition approach, as the state which is a vacuum state with respect to modes coming from $\mathscr I^-$, and a vacuum state with respect to modes coming from the full past horizon $\mathscr W$. From this prescription, we can fix the form of the 2-point function on the conformal past infinity and on $\mathscr W$. We start from the conformal past infinity, $\mathscr I^-$. From the metric \eqref{eq:BH-lightlike-coordinates} we see that the conformal past infinity is defined as the surface $S: u = u_0 $ for $u_0 \to - \infty$. The gradient of this hypersurface is $\nabla S = \partial_u$, and the normal vector is
\begin{equation}
\tilde n_{\mathscr I^-} = g^{\mu \nu} \partial_\mu S = - 2 \bigg ( 1 - \frac{2m}{r} \bigg )^{-1} \partial_v = - 2 \partial_v \ .
\end{equation}
Since $\mathscr I^-$ is a null hypersurface, there is not a canonical normalisation, but one can check that $v$ is the affine parameter on the hypersurface defining $n_{\mathscr I^-} = \partial_v $ and showing that $n^\mu_{\mathscr I^-} \nabla_\mu n^\nu_{\mathscr I^-} = 0$. This is easily done: since the vector component is constant, the covariant derivative has  only the Christoffel symbol contribution, and the only relevant component is $v$, so
\begin{equation}
n^\mu_{\mathscr I^-} \nabla_\mu  n^v_{\mathscr I^-} = \Gamma^v_{v v} = \frac{1}{2} g^{v k} (2 \partial_v g_{v k} - \partial_k g_{v v} ) \ ,
\end{equation}
but $g_{v v } = 0$, and if $k = u$ $g_{u v} \to -1/2$ at infinity, that is, it is independent on $v$. Since $n_{\mathscr I^-} $ satisfies the geodesic equation it is also the vector tangent to the geodesic generators of $\mathscr I^-$.

The vector $n_{\mathscr I^-}$ is nothing but the restriction of the Killing vector $\partial_t$ to $\mathscr I^-$, and one can see that the conformal past infinity is translationally invariant in the $v$ direction. This is analogous to say that that Minkowski space is translationally invariant in the $t$ direction, and therefore admits a Killing vector $\partial_t$; it makes sense, then, that in order to construct a vacuum state on $\mathscr I^-$ we choose a vacuum with respect to translations in $v$. But we already studied the form of the vacuum 2-point function with respect to some parameter in subsection \ref{ssec:interplay-vacuum-KMS}: it is given by
\begin{equation} \label{eq:Unruh-2PF-I}
\eval{W_2( \psi_1, \psi_2)}_{\mathscr I^-} = \frac{1}{\pi} \int_{\mathbb R \times \mathbb R \times \mathbb S_2} \frac{\psi_1(v_1) \psi_2(v_2)}{(v_2 - v_1 - i \epsilon)^2} \ \dd v_1 \dd v_2 \dd \Omega_2 \ ,
\end{equation}
where $\dd \Omega_2 = \sin \theta \dd \theta \dd \varphi$ is the metric of a unit 2-sphere. This is the component of the 2-point function of the Unruh state restricted to the conformal past infinity.

One could now follow the same reasoning for the 2-point function restriction to the past horizon $\mathscr W$. Since the spacetime admits a globally time-like Killing vector, it would make sense to choose a vacuum state with respect to $\partial_t$. Therefore, the restriction of its 2-point function on $\mathscr I^-$ would be \eqref{eq:Unruh-2PF-I}, while its restriction on the past horizon would originate a 2-point function which is the vacuum with respect to $u$. The state defined this way is called \textit{Boulware vacuum}, and it is the most natural choice at first sight. The Boulware vacuum corresponds to our familiar concept of vacuum at large distances from the black hole, but it shows a pathological behaviour at the event horizon, in the sense the expectation value of the stress-energy tensor, evalutated for a free falling observer, diverges \cite{Candelas80}. This is in contradiction with the Equivalence Principle, for which a free falling observer in a gravitational potential should be indistinguishable from an accelerated observer in the vacuum: if the observer gets annihilated from an infinite amount of energy at the event horizon, it would be hard to ignore the presence of a gravitational field. From a rigorous point of view, we can say that the Boulware vacuum is in Hadamard form in the exterior of the black hole, but not in the interior nor on the event horizon, and therefore cannot be considered as a physically sensible candidate for a state in a black hole background. Therefore, we need to find another candidate, one which is regular enough at the horizon (which does not annihilates any free-falling observer in a wall of fire while crossing the horizon). A possibility is to note that $v$ is the affine parameter on $\mathscr I^-$, and for symmetry to choose the vacuum on $\mathscr W$ with respect to its affine parameter. The past horizon is the surface defined by $ V = 0$. Since the metric in light-like coordinates is ill-defined, we use Kruskal coordinates here. The function defining this surface is $S: V = 0$, and so the gradient of the surface function $S$ is $\nabla S = \partial_V $, so that the normal vector is $\tilde n_{\mathscr W} \propto \partial_U $. We will fix the normalisation in a moment; indeed, one can check that $U$ is the affine parameter on $\mathscr W$, by noting that the only non-trivial contribution to the covariant derivative comes from the Christoffel symbol in the $U$ component:
\begin{equation}
\Gamma^U_{U U} = \frac{1}{2} g^{U k} ( 2 \partial_U g_{U k} - \partial_k g_{U U} ) \ ,
\end{equation}
but $g_{ U U } = 0$, and if $k = V$, one gets
\begin{equation}
\Gamma^U_{U U} = g^{U V} \partial_U g_{U V} = \frac{r}{16 m^3} e^{r/2m} \partial_U \bigg ( \frac{16 m^3}{r} e^{- r/2m} \bigg ) = - \bigg (\frac{1}{r} + \frac{1}{2m} \bigg ) \partial_U r \ .
\end{equation}
From the relation
\begin{equation}
UV = \bigg ( 1 - \frac{r}{2m} \bigg ) e^{r/2m} \ ,
\end{equation}
taking the derivative with respect to $U$ on both sides, we see that
\begin{equation}
V = - \frac{r}{4m^2} e^{r/2m} \partial_U r \ ,
\end{equation}
and, therefore, $\partial_U r \propto V$. But $V = 0$ on the past horizon by definition, and therefore $\Gamma^U_{U U} = 0$ on the past horizon. This in turn implies that $U$ is the affine parameter on $\mathscr W$. It makes sense, then, to choose a state such that its restriction to the conformal past infinity is \eqref{eq:Unruh-2PF-I}, and its restriction on the past horizon is the vacuum with respect to $U$, that is,
\begin{equation} \label{eq:Unruh-2PF-W}
\eval{W_2(\psi_1, \psi_2)}_{\mathscr W} = \frac{r_S^2}{\pi} \int_{\mathbb R \times \mathbb R \times \mathbb S_2} \frac{\psi_1(U_1) \psi_2(U_2)}{(U_2 - U_1 - i \epsilon)^2} \ \dd U_1 \dd U_2 \dd \Omega_2 \ .
\end{equation}
This defines the \textit{Unruh state} in $\mathscr M$. In \cite{DMP11} it has been proved that the Unruh state is Hadamard on $ \mathscr M$, and therefore it is a ground state for a free massless field propagating on the black hole.

We remark that this construction of the Unruh state holds for black hole in formation or evaporation, provided that spherical symmetry is preserved. In fact, in this case the spacetime outside the distribution of matter in accretion or evaporation is still a Schwarzschild solution, thanks to the Birkhoff theorem, and therefore one can still construct the Unruh state. Therefore, we will use the same state also when dealing with dynamical black holes, perturbed by some matter content for a finite amount of time.

\subsection{Hawking Temperature} \label{ssec:hawking-temperature}

The Hawking effect is associated to two conceptually distinct effects: one is the fact that, in an eternal black hole, the restriction of the Unruh state to the exterior of the black hole is a thermal state at Hawking temperature with respect to modes incoming from the white hole horizon. The second effect is the fact that a black hole formed by gravitational collapse radiates and eventually evaporates, emitting a thermal flux of particles which reaches future infinity as a black-body radiation at Hawking temperature.

From the form of the 2-point function of the Unruh state on the conformal past infinity and the past horizon, we can easily understand the thermal nature of the modes incoming from the white hole horizon, $\mathscr W^-$. In fact, the relationship between $U$ and $u$ (and, therefore, $t$) is the same relationship we studied in subsection \ref{ssec:interplay-vacuum-KMS} between $V$ and $v$. For an observer moving along a time-like geodesic, no signal can arrive from the interior of the black hole, and therefore its algebra is not the whole algebra $\mathcal A(\mathscr M)$, but its restriction to the black hole exterior $\mathcal A(\mathscr E)$. For the black hole exterior, the past horizon reduces to the white hole horizon $\mathscr W^-$, where $U < 0$, and we have the change of variables \eqref{Kruskal-light-like-coordinates}, that is, $U = - e^{-u/4m} $. The minus sign is the only difference with the change of variables we considered in subsection \ref{ssec:interplay-vacuum-KMS}, but the 2-point function \eqref{eq:Unruh-2PF-I} written in terms of $u$ coincides with the 2-point function \eqref{eq:2PF-vacuum-v}, where now $\kappa = 1/4m$:
\begin{equation}
\eval{W_2(\psi_1, \psi_2)}_{\mathscr W^-} = \lim_{\epsilon \to 0^+} \frac{\kappa^2}{4 \pi} \int_{\mathbb R_+ \times \mathbb R_+ \times S^{2}} \frac{\psi_1(u_1) \psi_2(u_2)}{\bigg( \sinh{\frac{\kappa}{2}(u_1-u_2)} + i\epsilon \bigg)^2} \dd u_1 \dd u_2 \dd \Omega_2 \ .
\end{equation}
We have seen however that this is the 2-point function for a thermal state with respect to translations in $u$, at inverse temperature $\beta = 2 \pi / \kappa  $; and therefore we can conclude that the Unruh state exhibits a thermal flow of radiation for modes incoming from the white hole horizon, at the \textit{Hawking temperature}
\begin{equation} \label{Hawking-temperature}
T = \frac{\kappa}{2 \pi} = \frac{1}{8 \pi m} \ .
\end{equation}

The fact that the Unruh vacuum is a thermal state for modes from $\mathscr W^-$ in the black hole exterior is similar to the Unruh effect, since both boil down to the fact that some observer has access to a restricted algebra of observables only. However, a black hole formed in gravitational collapse radiates a thermal flux of particles toward future infinity. This effect, which is the genuine Hawking effect that black hole radiate, can be understood considering the propagation of a massless field in the Unruh state in the black hole exterior. On the initial surface $\mathscr W^- \cup \mathscr I^-$, the field modes can be decomposed in those coming from past infinity, the so called \textit{in-modes}, and those coming from the white hole horizon, the \textit{out-modes}. The in-modes start with positive frequencies only, since the Unruh state restricted on past infinity is a vacuum state, while the out-modes are thermal states. Due to the propagation of the field in the black hole exterior, however, the modes are mixed, and, at future infinity, one sees a radial thermal flux of particles coming from the black hole event horizon. The details of this computation can be found in \cite{FredenhagenHaag89}.

The result of Hawking's paper \cite{Hawking74} showed that the relationship between black hole mechanics and thermodynamics is not a useful analogy but a true physical fact: black holes do have a finite temperature, and therefore a finite entropy. Knowing the temperature of a black hole, the entropy is easily found, using the first law of black hole mechanics and its analogy with the first law of thermodynamics: if the temperature is $T = \frac{\kappa}{2 \pi} $, it is immediate to see that a black hole carries an entropy
\begin{equation} \label{BH-entropy}
S_{BH} = \frac{A}{4}
\end{equation}
where $BH$ stand both for black hole or \textit{Bekenstein-Hawking}, according to one's taste. Thus, Hawking's computation fixes the proportionality coefficient between the area and the entropy, the bit of information Bekenstein couldn't derive with its gedanken experiments.

It is important to stress that the original derivation of Hawking radiation did not assume a quantum state, but rather used a semi-classical approximation, showing that the vacua before and after the formation of a black hole do not coincide but rather are related by a Bogoliubov transformation. The relation implies that because of the formation of a black hole we can detect a flux of thermal radiation.

The relevance of the Unruh state for the Hawking radiation is twofold. First, it can be used to compute the expectation value for the regularised stress-energy tensor of a massless scalar field in $\mathscr M$. It was shown in \cite{Candelas80} that the expectation value is regular on the event horizon (in contradistinction to the result obtained with the Boulware vacuum) but shows a steady flux of black-body radiation at future null infinity at Hawking temperature. This result shows that the Unruh vacuum is a suitable candidate quantum state, since it correctly reproduce the Hawking radiation. Second, it was shown in \cite{FredenhagenHaag89} that the existence of Hawking radiation after a spherically symmetric gravitational collapse is always implied whenever the state is of Hadamard form, in particular near the event horizon, and approximate the ground state near past null infinity. In particular, there is no need to assume any notion of particle, nor any assumption is made on the physics near the singularity: rather, the Hawking radiation is a consequence of the propagation of a quantum field in the outside of a black hole, provided we choose a physically sensible state. In terms of its 2-point function, the only assumption we need is that it reduces to \eqref{eq:Unruh-2PF-I} on the conformal past infinity, while it needs to be of Hadamard form on the whole past horizon $\mathscr W$.

We are now in the position to derive a formula for the entropy of a black hole using the same framework needed for the Hawking radiation. In particular, we consider a quantum state with 2-point function \eqref{eq:Unruh-2PF-I} on the past horizon, while we will assume only the Hadamard condition for the 2-point function on the past horizon, which we will call a \textit{generalised Unruh state}, or Unruh state for short. This is the class of states which exhibits Hawking radiation. We consider the propagation of a free, massless field in the outside of a black hole, and we compute the relative entropy between the Unruh state and a coherent perturbation. We will show that a variation of the relative entropy for the quantum field is accompanied by a variation of one fourth of the area of the horizon, not only in the Schwarzschild case, but also for the apparent horizon of a dynamical black hole with spherical symmetry. In fact, the computation of relative entropy relies on the generalised Unruh state, given by the 2-point function at conformal past infinity, which is a viable state whenever the spacetime is asymptotically flat.

\section{Relative Entropy for the Generalised Unruh State and Its Perturbations} \label{sec:RE-Unruh}

In order to discuss the variation of relative entropy we need to introduce one last modification to our model. In fact, instead of considering the whole past conformal infinity, we can give initial data (and the symplectic form) only on a portion of it, given by the condition $v < v_0 $. We will denote this portion with $\mathscr I^-(v_0) $. This implies that the quantum theory we construct is not defined on the whole Schwarzschild spacetime $\mathscr M$, but rather on the causal development of $\mathscr W \cup \mathscr I^-(v_0)$.

Now, we consider the propagation of a free massless scalar field in $D(\mathscr W \cup \mathscr I^-(v_0))$, and we give initial data such that $\eval{\psi}_{\mathscr W} = \eval{\nabla_n \psi}_{\mathscr W} = 0 $, with the covariant derivative in the direction normal to $\mathscr W$, that is, no matter is coming from the white hole.
 The idea is to perturb the Schwarzschild metric with a classical field, $\psi$, propagating over a fixed Schwarzschild background. We will consider this classical wave as a perturbation: we can consider a one-parameter family of metrics, $g(\zeta)$, such that $g(0)$ is the Schwarzschild metric. The stress-energy tensor of the scalar field is not a source of Einstein's equation, but it acts as a perturbation to some order in $\lambda$; in other words, for any function $f$, we will have the expansion
\begin{equation}
f(x) = f + \zeta \delta f + \zeta^2 \delta^2 f + ...
\end{equation}
so $\delta f = \dv{f}{\zeta}$ and $\delta^2 f = \frac{1}{2} \dv{f}{\zeta}$. To avoid notational clutter, we will denote simply with $f$ the value referred to the fixed background, so, for example, the area $A$ is the area of the Schwarzschild black hole $A = 4 \pi m^2$ and $\delta A$ its first order correction due to the scalar perturbation. 
 
The symplectic form \eqref{eq:symplectic-form-Schwarzschild}, over the space of the classical wave solutions, in the region $D(\mathscr W \cup \mathscr I^-(v_0))$, is
\begin{equation} \label{eq:symplectic-form-Schw-subregion}
w(\psi_1 , \psi_2) = \int_{\mathscr I^-(v_0)} ( \tilde \psi_2 \partial_\mu \tilde \psi_1 - \tilde \psi_1 \partial_\mu \tilde \psi_2 ) n^\mu_{\mathscr I^-} \ \dd v \dd \Omega_2 \ .
\end{equation}
The integral on the past horizon vanishes because of our choice of initial data, while we have already seen that $n_{\mathscr I^-} = \partial_v $. As we've said, we use a Weyl C*-algebra approach, and we select the generalised Unruh state, which is the vacuum with respect to $v$ translations at conformal past infinity. Since the Unruh state is quasifree we can construct the one-particle structure and the GNS-representation over a Hilbert space, and therefore the Unruh state is represented as a vector $\ket U \in \mathcal H$. Next, we introduce a coherent perturbation of the Unruh state, that is, a state of the form $W(\psi) \ket U$. We call the coherent state $\omega_\psi$. The stress-energy tensor of a classical solution of the Klein-Gordon equation is
\begin{equation} \label{SET-scalar-field}
T_{\mu \nu} = \partial_\mu \psi \partial_\nu \psi - \frac{1}{2} g_{\mu \nu} \partial_\alpha \psi \partial^\alpha \psi \ .
\end{equation} 

\begin{figure}
\centering
\includegraphics[scale=0.3]{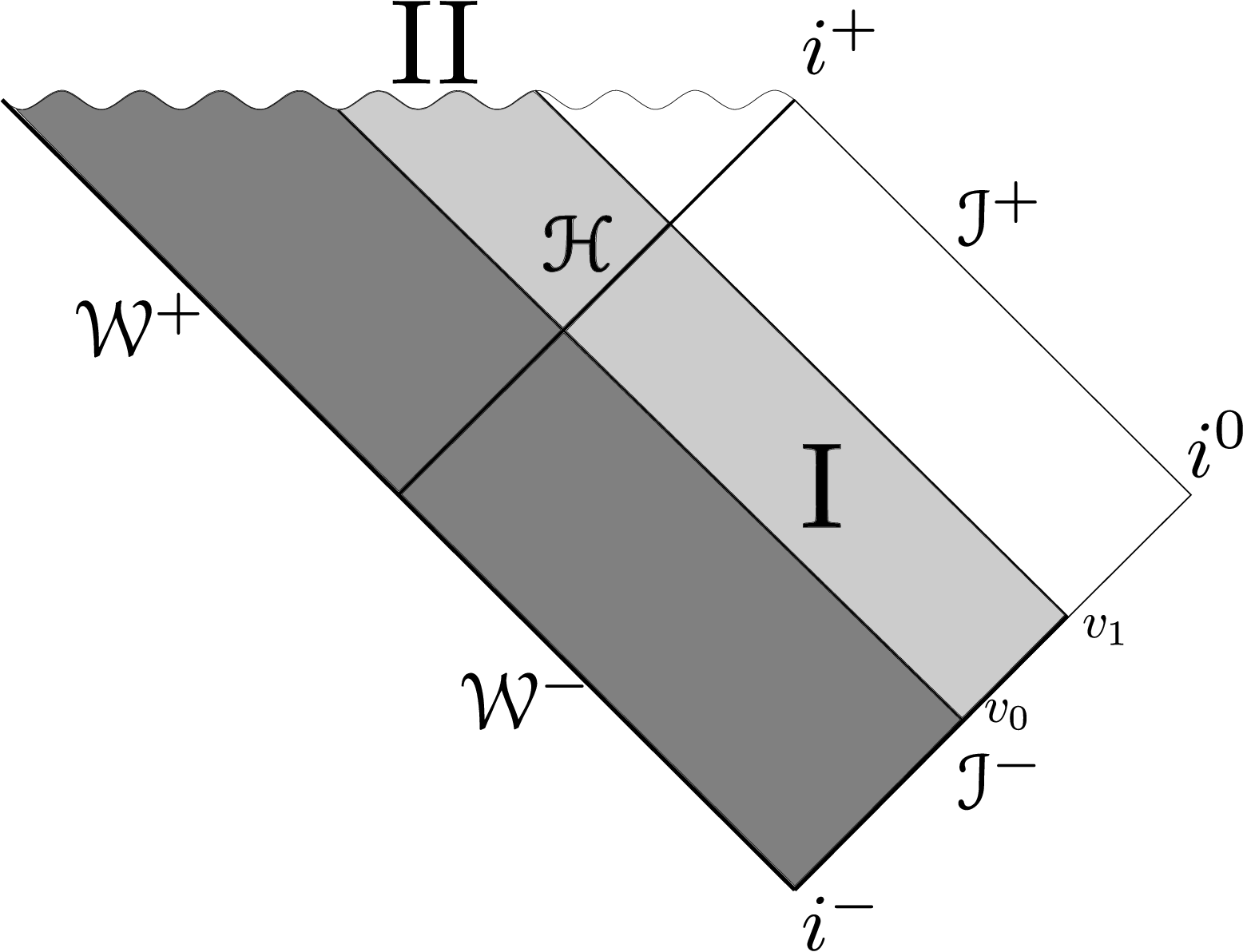}

\caption{The spacetime $\mathscr M$ we are considering, given by the union of regions I and II and their boundaries. The dark grey region is the causal development $D(\mathscr W \cup \mathscr I^-(v_0))$, while the union of the dark grey and light grey regions is the causal development $D(\mathscr W \cup \mathscr I^-(v_1))$.}
\end{figure}

Now, we want to compute the relative entropy between the Unruh state and its coherent perturbation, using \eqref{eq:EE-foundamental-formula}. Since the symplectic form is given at past infinity, where the Schwarzschild spacetime is flat, we can make use of the Bisognano-Wichmann theorem \eqref{th:bisognano-wichmann-theorem}, which states that $\alpha_s$ acts as a boost in the $t-x$ plane. From the change of coordinates \eqref{eq:change-of-variables-Lass}, we know that a boost is a translation $\tau \to \tau + s$, and we see that the coordinate $v$ transforms as
\begin{equation}
v = t + x = \frac{e^{a \xi}}{a} (\cosh a \tau + \sinh a \tau) = \frac{1}{a} e^{a ( \xi + \tau)} \to e^{a s}v \ ,
\end{equation}
and, therefore, the modular flow acts on $\psi$ as a dilation in the $v$ direction:
\begin{equation}
\alpha_s(\psi(u,v,\theta, \varphi) ) = \psi(u, v_0 + e^{- 2 \pi s} (v_0 - v ), \theta, \varphi) \ .
\end{equation}
Therefore, we have
\begin{equation} \label{eq:coherent-state-Unruh-alpha}
\eval{\dv{\alpha_s(\psi)}{s}}_{s = 0} = - 2 \pi (v_0 - v) \partial_v \psi \ .
\end{equation}
Now, we have all the ingredients to compute the relative entropy, in the region $ D(\mathscr W \cup \mathscr I^-(v_0))$, between the Unruh state and a coherent perturbation caused by a classical wave $\psi$. We have seen in section \ref{sec:E-EE}, subsection \ref{ssec:EE-coherent-states}, that the relative entropy between two coherent states is given by \eqref{eq:EE-foundamental-formula}. Using \eqref{eq:coherent-state-Unruh-alpha}, we can compute the relative entropy between the Unruh state and a coherent perturbation, $S_{v_0}(\omega | \omega_\psi) $:
\begin{multline}
S_{v_0}(\omega | \omega_\psi) = \frac{1}{2} w \bigg (  \eval{\dv{\alpha_s(\psi)}{s}}_{s=0}, \psi \bigg ) = \\ 
= - \frac{1}{2} 2 \pi \int_{-\infty}^{v_0} \int_{\mathbb S_2} \psi \partial_v [(v_0 - v) \partial_v \tilde \psi] - (v_0 - v) (\partial_v \tilde \psi)^2  \ \dd v \dd \Omega_2 \ .
\end{multline}
Integrating by parts the first term, one sees that the boundary contribution vanishes, at $v_0$ because the integrand is proportional to $v_0 - v$, and at $-\infty$ because $\eval{\tilde \psi}_{-\infty} = 0$ since it is evaluated on $\mathscr W$, where the initial data are zero. So one gets
\begin{equation} \label{eq:RE-Schw-v0}
S_{v_0}(\omega | \omega_\psi) = 2 \pi\int_{-\infty}^{v_0} \int_{\mathbb S_2} (v_0 - v) (\partial_v \tilde \psi)^2  \ \dd v \dd \Omega_2 \ .
\end{equation}
This is the same form of the result obtained in \cite{HollandsIshibashi19} for the relative entropy between coherent states of a free scalar field.

Now, we compute the derivative of the relative entropy as we translate the boundary by some parameter, $v_0 \to v_0 + t$. If we take the derivative of $S_{v_0 + t}$ with respect to $t$, we find the variation of entropy when we modify the region in which the theory is defined:
\begin{equation}
\dv{t} S_{v_0}(\omega | \omega_\psi) = 2 \pi\int_{-\infty}^{v_0} \int_{\mathbb S_2}(\partial_v \tilde \psi)^2  \ \dd v \dd \Omega_2 \ .
\end{equation}

We can repeat the same computation choosing another portion of conformal past infinity $\mathscr I^-(v_1)$, where $v_0 < v_1$, getting the same result. Now, we can rigidly translate both $v_0$ and $v_1$ of the same parameter $t$.  We take the difference between the derivatives of $S_{v_1}$ and $S_{v_0}$ with respect to $t$, and finally, noting that a derivative with respect to $t$ can be written as a derivative with respect to $v$, we find 
\begin{equation} \label{RE-variation-Unruh-state}
\frac{1}{2 \pi}\dv{v}\big ( S_{v_1}(\omega | \omega_\psi) - S_{v_0}(\omega | \omega_\psi) \big ) = \int_{v_0}^{v_1} \int_{\mathbb S_2} (\partial_v \tilde \psi)^2  \ \dd v \dd \Omega_2
\end{equation}
The notation $\dv{v}S_{v_0}(\omega | \omega_\psi) $ should be read as the derivative of $S_{v}(\omega | \omega_\psi)$, that is, the relative entropy for a region extended up to $v$, with respect to the location of the boundary of that region, and then evaluated for $v = v_0$.

We stress on the fact that this is not the relative entropy associated to the region $v_0 < v < v_1$, since the entropy is not localised but is rather a property of the theory. The fact that the integral takes contributions only on this region is simply a consequence of the computation. We can say that the relative entropy is associated with the region in which the theory is defined, that is, the causal development of the surface on which we give initial data. Therefore, this quantity expresses the variation of the difference between the relative entropies for a theory defined in two different regions, when we rigidly translate their boundaries of the same amount $t$.

If we had another quantum state, how would the entropy change? Suppose we are given a state $\omega'$ in which the modular flow acts not as dilation in the $v$ direction, but rather as dilation of some different parameter which is a smooth function of it, $w(v)$, that is,
\begin{equation}
\alpha_s(\psi(u, w(v), \theta, \varphi)) = \psi(u, w_0 + (w - w_0)e^{- 2 \pi s}, \theta, \varphi) \Rightarrow \eval{\dv{\alpha_s(\psi)}{s}}_{s=0} = - 2 \pi (w - w_0 ) \partial_w \psi \ .
\end{equation}
This can happen if, for example, the state we choose is invariant under translations in $w$. Then, calling $w_0 = w(v_0)$, the relative entropy between $\omega'$ and a coherent perturbation of it, $\omega'_\psi$, is given by expression \eqref{eq:EE-foundamental-formula},
\begin{multline}
S(\omega' | \omega'_\psi) = \frac{1}{2} w \bigg( \eval{\dv{\alpha_s(\psi)}{s}}_{s= 0}, \psi \bigg) = \\ 
= - \frac{1}{2} 2 \pi \int_{-\infty}^{v_0} \int_{\mathbb S_2} \tilde \psi \partial_v [(w_0 - w) \partial_w \tilde \psi] - (w_0 - w) \partial_w \tilde \psi \partial_v \tilde \psi  \ \dd v \dd \Omega_2 \ .
\end{multline}
We can again integrate by parts the first term, noting that the boundary term once again vanishes, to get
\begin{equation}
S(\omega' | \omega'_\psi) = 2 \pi \int_{-\infty}^{v_0} \int_{\mathbb S_2} (w_0 - w) \partial_w \tilde \psi \partial_v \tilde \psi  \ \dd v \dd \Omega_2 \ .
\end{equation}
Now, the derivative with respect to $w$ can be rewritten as $\dv{w} = \dv{v}{w} \dv{v}$. If we expand the relation between $w$ and $v$ around $w_0$,
\begin{equation}
w - w_0 = \dot w(v_0) (v - v_0) + \mathcal O((v-v_0)^2) \ ,
\end{equation}
where the dot denotes the derivative with respect to $v$, then we can rewrite the integral at first order in $v-v_0$, obtaining
\begin{equation}
S(\omega' | \omega'_\psi) \simeq 2 \pi \int_{-\infty}^{v_0} \int_{\mathbb S_2}  \dot w(v_0)(v_0 - v) \frac{1}{\dot w(v)} (\partial_v \tilde \psi)^2  \ \dd v \dd \Omega_2  \ .
\end{equation}
Now, we can further expand the derivative of $w$ in
\begin{equation}
\dot w(v)^{-1} = \frac{1}{\dot w(v_0) + (v-v_0)\ddot w(v_0)} + \mathcal O((v-v_0)^2) \simeq \frac{1}{\dot w(v_0)} \bigg (1 - (v-v_0)\frac{\ddot w(v_0)}{\dot w(v_0)} \bigg ) \ .
\end{equation}
The relative entropy then becomes, at linear order in $v-v_0$,
\begin{equation}
S(\omega' | \omega'_\psi) \simeq 2 \pi \int_{-\infty}^{v_0} \int_{\mathbb S_2} \bigg(1 + (v_0 - v) \frac{\ddot w(v_0)}{\dot w (v_0)} \bigg) (v_0 - v) (\partial_v \tilde \psi)^2 \ \dd v \dd \Omega_2 \ . 
\end{equation}
We see that the first term reproduce the relative entropy in the Unruh state, while the second term is a correction proportional to $v_0 - v$. Therefore, the leading term in the relative entropy between coherent states is independent on the choice of the state, as long as the state is a vacuum with respect to some parameter $w(v)$. Our results on the relative entropy therefore holds also at first order in a perturbation expansion in $v_0 - v$ for states which are perturbations to the Unruh state, in the sense that the parameter $w(v)$ can be expanded in $v$.

\section{Variation of Generalised Entropy in Schwarzschild}
The idea is now to follow the conceptual steps of Bekenstein gedanken experiments \cite{Bekenstein73} in order to identify the correct expression for the entropy of a black hole. In this section we closely follow the work of Hollands and Ishibashi \cite{HollandsIshibashi19}, which showed how the relative entropy between coherent states can be used in the context of black hole thermodynamics, with some simple modifications.

Here, we considered the two initial data surfaces $\mathscr W \cup \mathscr I^-(v_0)$ and $\mathscr W \cup \mathscr I^-(v_1)$ (that are, the union of the past horizon with a portion of past null infinity with $v < v_{0,1}$), from which we let a free, massless scalar field to propagate in their causal developments, $D(\mathscr W \cup \mathscr I^-(v_0))$ respectively $D(\mathscr W \cup \mathscr I^-(v_1))$, which cover a portion of the black hole exterior. As we have seen in the last section, we can compute the relative entropy associated to each region between the Unruh state and a coherent perturbation. We remark that the relative entropy is not localised, but it is associated to a theory, and as such is dependent on the region in which the theory is defined. We then compute the variation of relative entropy as we move the boundary of its associated region by some parameter $t$, $v_{0,1} \to v_{0,1} + t$, and finally we take the difference between the derivatives of the relative entropies associated to the two regions $D(\mathscr W \cup \mathscr I^-(v_0))$ and $D(\mathscr W \cup \mathscr I^-(v_1))$.

We then want to link this quantity to a variation in the horizon area of the black hole. In order to do so, we use a conservation law associated to the field stress-energy tensor. In the case of Schwarzschild, we use the conserved current associated to the Killing vector $\xi_t = \partial_t$. In the case of dynamical horizons, we will see that it is possible to construct a different vector, the Kodama vector \cite{Kodama79} (see section \ref{ssec:kodama-miracle}) which defines a different conserved current, but the conceptual setup is exactly the same. In the black hole exterior, we consider the region enclosed by four null hypersurfaces: two null hypersurfaces at $v = v_0$ and $ v = v_1$, the corresponding portion at past infinity, $\delta \mathscr I^-$, defined by the condition $v_0 < v < v_1$, and the portion of event horizon $\delta \mathscr H$ with $ e^{\kappa v_0} < V < e^{\kappa v_1} $. We then translate the Killing conservation law in a integral of a flux across the four boundaries using the Gauss-Stokes theorem, \eqref{th:gauss-stokes-theorem}. Finally, the flux term across the event horizon is linked to a variation of the horizon area, using the Raychaudhuri equation (equation \eqref{eq:Raychaudhuri}, in subsection \ref{ssec:Raychaudhuri-eq}) for the null generators of the horizon; the flux term at past infinity is written in terms of the variation of the derivative of relative entropy, as we constructed in the last section, and this finally connects the variation of the horizon area with a variation of relative entropy. The result is given in equation \eqref{eq:result-RE-Schw}.

In \cite{HollandsIshibashi19}, Hollands and Ishibashi take the Cauchy surface for the scalar field as the union of a portion of event horizon and of future null infinity. This way, they are able to compute the relative entropy directly as a variation of the horizon area, plus a flux term given by the field stress-energy tensor at future infinity, without the need of a conservation law. We have chosen to modify their setup for a couple for reasons. The first is motivated on physical grounds: we wanted to give initial data at past infinity, and see how the field would propagate in the spacetime, rather than fix the field's data at future infinity. Moreover, giving initial data at past infinity we could choose the Unruh 2-point function, which is the most physically sensible state for black hole spacetimes. The second reason is that, in order to generalise to dynamical black holes, we want to compute the variation in horizon area in an arbitrary, finite interval. This is because in section \ref{sec:RE-vaidya} we will consider a perturbation in a Schwarzschild spacetime caused by the infalling of some matter in a finite interval into the black hole. While the matter falls into the black hole, the distinction between apparent horizons (see subsection \ref{ssec:AH-first-law}) and event horizons become important, and we want to associate the variation in relative entropy to the variation of the apparent horizon area. Therefore, we need to compute the area variation in a finite interval, and not between an instant $v_0$ and infinity, when the apparent and event horizons coincide again. We can then compute the limit in which one of the two extremes of our region, $v_0$ or $v_1$, goes to infinity. 

\begin{figure}
\centering
\includegraphics[scale=0.3]{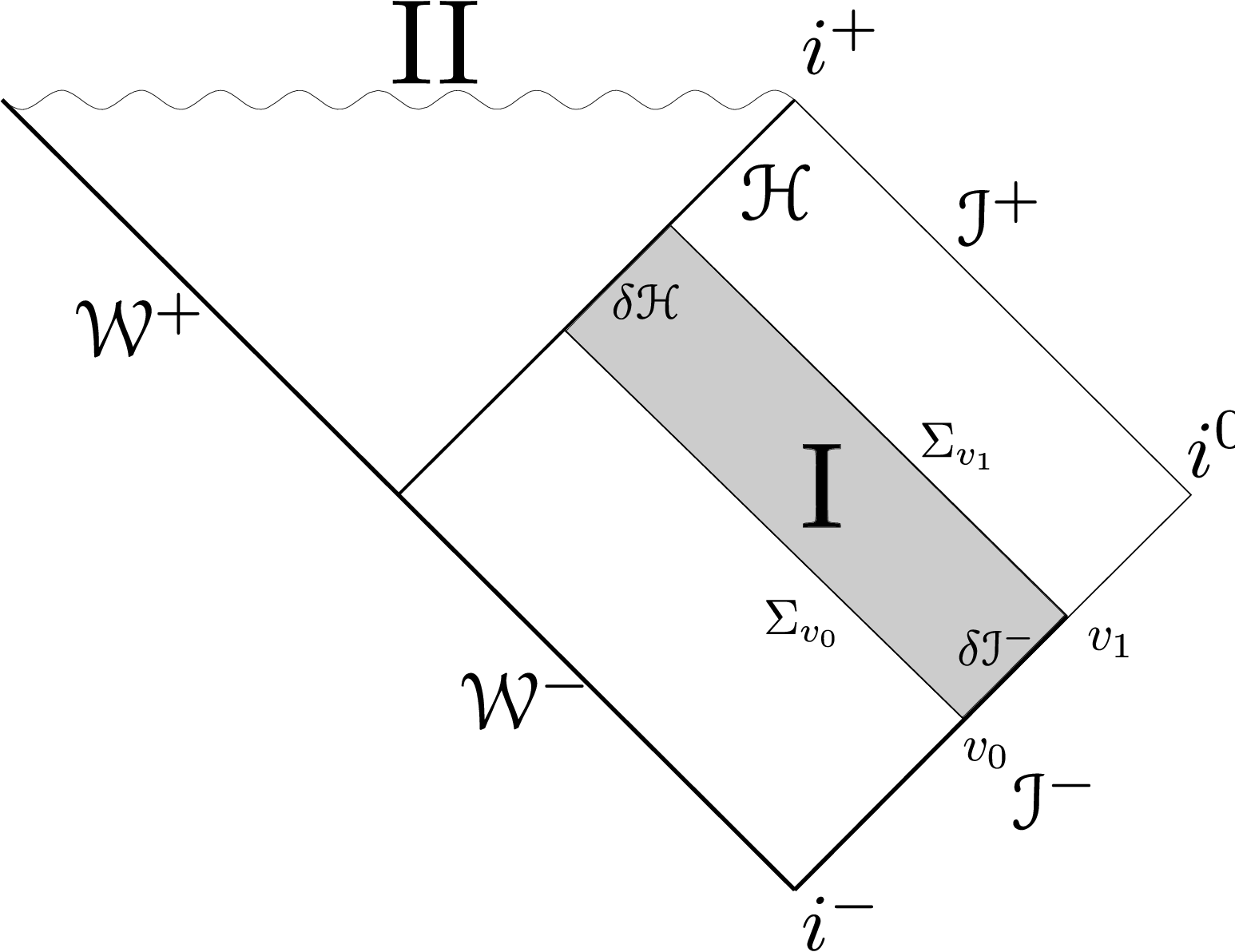}

\caption{In grey, the interior of the region enclosed by the four boundaries over which we compute the flux of the conservation law.}
\end{figure}

    \subsection{Killing Conservation Law} \label{ssec:killing-conservation-law}
    
In order to connect the variation of relative entropy \eqref{RE-variation-Unruh-state} to a variation in the horizon area, we make use of the conservation law for the Killing current associated with $\xi_t = \partial_t$, $\nabla_\mu ( \xi_{t \nu} T^{\mu \nu} ) = 0$. Since $\partial_t$ will be the only Killing vector field we will consider in this and the following sections, we drop the subscript $t$ and we write simply $\xi = \partial_t$. We now rewrite the Killing conservation law as a flux integral, using the Gauss-Stokes theorem, \eqref{th:gauss-stokes-theorem}.  Consider a region in the Schwarzschild exterior defined by four, light-like boundaries: two hypersurfaces at $v = v_0$ and $v = v_1$, denoted $\Sigma_{v_0}$ and $\Sigma_{v_1}$, the portion of conformal past infinity between the two, which we denote $\delta \mathscr I^-$, (that is, the portion of conformal past infinity with $v_0 < v < v_1 $) and the corresponding portion of event horizon, $\delta \mathscr H$, with $e^{\kappa v_0} < V < e^{\kappa v_1}$. These four hypersurfaces define what is called a double cone in the Schwarzschild exterior, and form its boundary. Therefore, we can use the Gauss-Stokes theorem to turn the conservation law for $J^\mu$ into a flux integral over the boundary:
\begin{equation} \label{flux-integral-Schw}
\begin{split}
0 &= \int_{\delta \mathscr H} J_\mu n^\mu_{\mathscr H} \ \dd V \dd \mathbb S_2 + \\
&- \int_{\Sigma_{v_0}} J_\mu n^\mu_{\Sigma_{v_0}} \ \dd u \dd \mathbb S_2 + \\
&+ \int_{\Sigma_{v_1}} J_\mu n^\mu_{\Sigma_{v_1}} \ \dd u \dd \mathbb S_2 + \\
&- \int_{\delta \mathscr I^-} J_\mu n^\mu_{\mathscr I^-} \ \dd v \dd \mathbb S_2
\end{split}
\end{equation}
The signs come from the chosen convention, that the normal vector is always future-directed. From this equation we will rewrite each term in order to find a law for the variation of the area of the horizon and of the relative entropy of the scalar field.

We briefly comment that the first surface integral in the above equation, the one on the event horizon, is slightly different from the one we computed in the context of the first law, \eqref{I-law-EH}. There, we computed the variation in the black hole mass along a Killing vector \textit{normal} to the horizon. In that case too we had a conserved current constructed from the stress-energy tensor and a Killing vector, but the two Killing fields are different. In particular, confusion may arise because the first component of the Killing vector we used in the first law coincides with the Killing vector $\xi = \partial_t$, which we use here and hereafter. However, the parametrisation in the two cases is different: in the case of the first law, we computed the component of the conserved current along the Schwarzschild time, $t$. Here, we compute the integral of the conserved current along the horizon itself, with affine parametrisation. The result of the integration therefore will not be the variation of the area in time, times a multiple of the surface gravity, as we found in the first law, but, as we will see in the next subsection, the variation of the \textit{derivative} of the area.

    \subsection{Area term} \label{ssec:area-term-Schw}

On the event horizon, the Killing vector reduces to $\xi = \frac{V}{2 r_S} \partial_V $. On the other hand, we computed in the previous section the normal vector to the past horizon, $\mathscr W$, finding $n_{\mathscr W} = \partial_U $. The computation for the normal vector to the event horizon goes along the same lines, and in particular one finds that $n_{\mathscr H} = \partial_V$, where $V$ is the affine parameter along the horizon generators. Finally, since the event horizon is a light-like hypersurface, the induced metric is singular, as we have seen in subsection \ref{ssec:hypersurfaces}, equation \eqref{eq:surface-element-null}. The non-vanishing components are the angular part of the metric, so the square root of the determinant of the metric is $\sqrt \sigma = r_S^2 \sin \theta \dd \theta \dd \varphi$, while the induced coordinates are $(V, \theta, \varphi)$. Therefore, the flux term on the horizon becomes
\begin{equation}
\int_{\delta \mathscr H} J_\mu n^\mu_{\mathscr H} \ \dd V \dd \mathbb S_2 = \int_{\delta \mathscr H} \frac{V}{2 r_S} T_{V V} \ \dd V \sqrt{\sigma} \dd \theta \dd \varphi \ .
\end{equation}       
We can rewrite the integrand using the Rayhaudhuri equation \eqref{eq:Raychaudhuri}, that we discussed in subsection \ref{ssec:Raychaudhuri-eq}, for the horizon generators congruence, which has tangent vector $u = \xi$.  Since we are working with an affine parameter, $\kappa \theta = 0$, and since the geodesics are light-like we can use the expression for the Raychaudhuri equation \eqref{eq:Raychaudhuri-SET}. The congruence is hypersurface orthogonal, so the rotation tensor vanishes. In fact, since the perturbation is a scalar, no tensor perturbations are possible, and therefore $\sigma^{ab} = \omega^{ab} = 0$. Finally, we can compute the expansion:
\begin{equation}
\theta = \nabla_\mu \xi^\mu = \partial_\mu \xi^\mu + \Gamma^\mu_{\mu \nu} \xi ^\nu = \frac{V}{2 r_S} \Gamma^\mu_{\mu V} \ .
\end{equation}       
The partial derivatives vanish because $\partial_\mu \xi^\mu = - \frac{1}{2 r_S} + \frac{1}{2 r_S}$, while
\begin{equation}
\Gamma^U_{U V} = \frac{1}{2} g^{UV} ( \partial_U g_{V V} + \partial_V g_{U V} - \partial_V g_{U V}) = 0 \ .
\end{equation} 
On the other hand, $\Gamma^V_{V V}$ vanishes for the same reason why $\Gamma^U_{U U} = 0$ on the past horizon, as we saw in the last subsection; in fact, the computation is exactly the same if we make the exchange $U \leftrightarrow V$. Therefore, $\theta = 0$ on the horizon. The same conclusion follows from our physical assumption of a stationary black hole, which does not grow. Finally, in \cite{HollandsWald12}, it is shown that a scalar perturbation acts only at second order in the perturbation expansion. This can be seen directly in the Raychaudhuri equation, because if we take a derivative with respect to the perturbative parameter $\zeta$, we see that the geometric quantities in the right-hand side of \eqref{eq:Raychaudhuri-SET} (that are, the expansion scalar and the shear and rotation tensors) vanish identically, because each term is quadratic and therefore the derivative is proportional to the zeroth quantity: for example, for the expansion term,
\begin{equation}
\delta \theta^2 \propto \theta \delta \theta = 0 \ .
\end{equation}
The same holds for the second-order variation; in fact, we can see that
\begin{equation}
\delta^2 (\theta^2) = 2 \delta(\theta \delta \theta) = 2 (\delta \theta)^2 + 2 \theta \delta^2 \theta = 0 \ ,
\end{equation}
with analogous computations for the shear and rotation tensors.

On the left-hand side, the first-order variation vanishes because we can exchange the derivative with respect to the affine parameter, $V$, and the derivative with respect to $\zeta$. The second order derivative, however, it is not necessarily zero, and therefore the scalar pertubation acts at second order in the geometric quantities appearing in the Raycaudhuri equation. Thus, the Raychaudhuri equation for the perturbation of the scalar expansion caused by the scalar field is
\begin{equation}
\dv{\delta^2 \theta}{V} = - 8 \pi T_{V V} \ .
\end{equation}

We can then substitute the left-hand side in the flux integral on the horizon, finding
\begin{equation}
 \int_{\delta \mathscr H} \frac{V}{2 r_S} T_{V V} \ \dd V \sqrt{\sigma} \dd \theta \dd \varphi = - \frac{1}{16 \pi r_S } \int_{\delta \mathscr H} V \dv{\delta^2 \theta}{V} \ \dd V \sqrt{\sigma} \dd \theta \dd \varphi \ .
\end{equation}

We integrate by parts to get
\begin{equation} \label{eq:area-integral-S}
   \frac{1}{16 \pi r_S} \int_{S^{d-2}} \eval{V \delta^2 \theta(V)}^{V_1}_{V_0} \sqrt \sigma \dd \theta \dd \varphi - \delta^2 \int_{V_0}^{V_1} \int_{S^{d-2}}  \theta \dd V  \sqrt \sigma \dd \theta \dd \varphi \ .
\end{equation}

The second term is vanishing, as the expansion is zero on the horizon. Now, we use the geometric interpretation of the scalar expansion \eqref{eq:expansion-area-light-like}. From this, we find a relation for the first order perturbation
\begin{equation}
   0 = \delta \theta = \frac{1}{\sqrt \sigma} \bigg( \dv{\delta \sqrt \sigma}{V} - \theta \frac{\delta \sqrt \sigma}{\sqrt \sigma} \bigg) = \frac{1}{\sqrt \sigma} \dv{\delta \sqrt \sigma}{V} \ ,
\end{equation}
and the second order:
\begin{equation}
    \delta^2 \theta = \frac{1}{\sqrt \sigma} \bigg ( \dv{\delta^2 \sqrt \sigma}{V} - \delta \theta \frac{\delta \sqrt \sigma}{\sqrt \sigma} \bigg)  = \frac{1}{\sqrt \sigma} \dv{\delta^2 \sqrt \sigma}{V} \ .
\end{equation}
We can substitute it in the integral, finding

\begin{equation} \label{eq:area-term-Schw-intermediate}
\begin{split}
    \frac{1}{16 \pi r_S} \int_{S^{d-2}} \eval{V \delta^2 \theta(V)}^{V_1}_{V_0} \sqrt \sigma \dd \theta \dd \varphi  &=
    \frac{1}{16 \pi r_S} \int_{S^{d-2}} \eval{V \dv{\delta^2 \sigma}{V}}^{V_1}_{V_0} \dd \theta \dd \varphi = \\
    = & \frac{1}{16 \pi r_S} \bigg( V_1 \dv{\delta^2 A(V_1)}{V} - V_0 \dv{\delta^2 A(V_0)}{V} \bigg) \ .
\end{split}
\end{equation}

Using $\dv{ v} = \frac{ V}{2 r_S}\dv{ V}$ we finally have
\begin{equation} \label{eq:flux-term-area-Schw}
    \int_{\delta \mathscr H} J_\mu n^\mu_{\mathscr H} \ \dd V \dd \mathbb S_2 = - \frac{1}{8 \pi} \bigg( \dv{\delta^2 A(v_1)}{v} - \dv{\delta^2 A(v_0)}{v} \bigg) \ .
\end{equation}

    \subsection{Variation of Generalised Entropy} \label{ssec:generalised-entropy-Schw}

We now consider the flux term at conformal past infinity. This is given by
\begin{equation}
\Phi_{\delta \mathscr I^-} = - \int_{\delta \mathscr I^-} J_\mu n^\mu_{\mathscr I^-} \ \dd v \dd \mathbb S_2 \ .
\end{equation}
It has been shown in \cite{DMP11} that the scalar field decays sufficiently fast at infinity, so that this integral is in fact convergent.

The integrand is $n^\mu_{\mathscr I^-} J_\mu = T_{uv} + T_{v v}$; however, the off-diagonal term is
\begin{equation}
T_{u v} = \partial_u \psi \partial_v \psi - \frac{1}{2}g_{uv} ( 2 g^{uv} \partial_u \psi \partial_v \psi + \sigma^{AB} \partial_A \psi \partial_B \psi) = -\frac{1}{2} g_{uv} \sigma^{AB} \partial_A \psi \partial_B \psi \ .
\end{equation}
But the angular part of the metric has an asymptotic behaviour $\sigma^{AB} \sim \frac{1}{r^2}$, and therefore this term does not contribute at infinity. So we are left with
\begin{equation}
\Phi_{\delta \mathscr I^-} = - \int_{\delta \mathscr I^-} (\partial_v \psi)^2 \ r^2 \sin \theta \dd v \dd \theta \dd \varphi \ ,
\end{equation}
which can be rewritten in terms of $\tilde \psi = r \psi$,
\begin{equation}
r \partial_v \psi = \partial_v (\tilde \psi) - \frac{\partial_v r}{r} \tilde \psi \ .
\end{equation}
The second term vanishes at infinity, so the flux term is
\begin{equation}
\Phi_{\delta \mathscr I^-} = - \int_{\delta \mathscr I^-} J_\mu n^\mu_{\mathscr I^-} \ \dd v \dd \mathbb S_2 = - \int_{\delta \mathscr I^-} (\partial_v \tilde \psi)^2 \ \sin \theta \dd v \dd \theta \dd \varphi \ .
\end{equation}
This is nothing but the expression we computed for the variation of relative entropy in \eqref{RE-variation-Unruh-state}:
\begin{equation}\label{eq:flux-term-infinity-RE-Schw}
- \int_{\delta \mathscr I^-} J_\mu n^\mu_{\mathscr I^-} \ \dd v \dd \mathbb S_2 = - \frac{1}{2 \pi}\dv{v}\big ( S_{v_1}(\omega | \omega_\psi) - S_{v_0}(\omega | \omega_\psi) \big ) \ ,
\end{equation}
where $\dv{v}$ is the derivative with respect to the translation in the $v$ direction of the two null hypersurfaces at fixed $v$, $\Sigma_{v_{0,1}} \to \Sigma_{v_{0,1} + t} $, as we discussed in section \ref{sec:RE-Unruh}.

We can therefore rewrite the flux integral of the conserved current, \eqref{flux-integral-Schw}, substituting \eqref{eq:flux-term-area-Schw} and \eqref{eq:flux-term-infinity-RE-Schw}, reorganizing terms and noting that the derivative with respect to $v$ and $t$ coincide, in order to find
\begin{equation}
\frac{1}{2 \pi} \dv{v} \bigg [ S_{v_1} + \frac{1}{4} \delta^2 A(v_1) - \bigg( S_{v_0} + \frac{1}{4} \delta^2 A(v_0) \bigg )  \bigg] = - \int_{\Sigma_{v_0}} J_\mu n^\mu_{\Sigma_{v_0}} \ \dd u \dd \mathbb S_2 + \int_{\Sigma_{v_1}} J_\mu n^\mu_{\Sigma_{v_1}} \ \dd u \dd \mathbb S_2 \ .
\end{equation}
It is natural, then, to introduce a \textit{generalised entropy} 
\begin{equation} \label{eq:generalised-entropy-Schw}
S_{gen} = S(\omega | \omega_\psi) + \frac{1}{4}A \ ,
\end{equation}
which takes into account the contribution to the entropy given by the relative entropy between the vacuum and a coherent perturbation, and one-fourth the area of the black hole, which arises as the entropy contribution of the black hole. Since at zeroth and first order the area of the black hole does not change, the equation can be rewritten as
\begin{equation} \label{eq:result-RE-Schw}
\frac{1}{2 \pi} \dv{v}S_{gen}(v_1) - \dv{v} S_{gen}(v_0) = - \int_{\Sigma_{v_0}} J_\mu n^\mu_{\Sigma_{v_0}} \ \dd u \dd \mathbb S_2 + \int_{\Sigma_{v_1}} J_\mu n^\mu_{\Sigma_{v_1}} \ \dd u \dd \mathbb S_2 \ ,
\end{equation}
which holds up to second order in the perturbation expansion in $\zeta$.

To make the expression cleaner, we can consider the limit in which $v_0 \to - \infty$, in which $\Sigma_{v_0} \to \mathscr W^-$. In this limit, if we assume that the black hole is stationary, $\delta^2 A(v_0) = 0$, simply because, in \eqref{eq:area-term-Schw-intermediate}, $\theta(V_0) = 0$. On the other hand, the choice $v_0 = -\infty$ shrinks the portion $\delta \mathscr I^-(v_0)$ to a 2-sphere. Then, the integral which gives the relative entropy, \eqref{eq:RE-Schw-v0}, becomes zero. In other terms, we have $S_{gen}(v_0) \to 0$ if $v_0 \to - \infty$. Finally, since we have chosen initial data such that the field and its derivatives vanish on the white hole horizon, the flux integral on $\Sigma_{v_0}$ vanishes, since it becomes a surface integral of the field on $\mathscr W^-$. Therefore,  equation \eqref{eq:result-RE-Schw} reduces to two terms: the generalised entropy associated to the region $D(\mathscr W^- \cup \mathscr I^-(v_1)$, and the integral of a flux of particles crossing $\Sigma_{v_1}$:
\begin{equation} \label{eq:result-RE-Schw-limit}
\dv{v} S_{gen}(v_1) = 2 \pi \int_{\Sigma_{v_1}} J_\mu n^\mu_{\Sigma_{v_1}} \ \dd u \dd \mathbb S_2 \ .
\end{equation}

\section{Variation of Generalised Entropy in Vaidya spacetimes} \label{sec:RE-vaidya}

\subsection{Vaidya Spacetime} \label{ssec:vaidya-geometry}
In this section we want to apply the same framework we used in the Schwarzschild spacetime to the case of a Vaidya black hole.

A Vaidya black hole is a toy model for a black hole in formation or in evaporation. It is the simplest generalisation of a Schwarzschild black hole, since it replaces the constant mass parameter $m$ with a function of the advanced time only, $m(v)$. The metric in Eddington-Finkelstein coordinates then is
\begin{equation} \label{eq:Vaidya-metric}
\dd s^2 = - \bigg ( 1 - \frac{2m(v)}{r} \bigg) \dd v^2 + 2 \dd v \dd r + r^2 \dd \Omega_2^2 \ .
\end{equation}
We will often denote $- g_{v v} = f(v,r) = 1- 2m(v)/r $.

The stress-energy tensor for the matter source of the Vaidya metric can be computed from Einstein equations, and is given by
\begin{equation} \label{eq:SET-source-Vaidya}
T^{source}_{\mu \nu} = \frac{\dot m(v)}{4 \pi r^2} \delta_{v \mu} \delta_{v \nu} \ .
\end{equation}
This is the stress-energy tensor of shells of null dust, or radiation.

If the weak energy condition holds, that is, if the source matter is classical, one can show using the Einstein equation that $\dot m(v) > 0$: the mass function is increasing in time, and the Vaidya metric \eqref{eq:Vaidya-metric} describes a black hole in accretion. For our purposes, the weak energy condition is not needed: that is, our results hold also if $\dot m(v) < 0$, in which case the black hole evaporates emitting radiation. However, to fix the ideas we will talk of the source matter as infalling matter, considering the classical case.

We choose a mass function so that our black hole is just a perturbation in a finite time of the Schwarzschild black hole. In other words, we consider a black hole in which some matter falls during an interval $v \in [v_0, v_1]$, and that is stationary otherwise. The spacetime is a static black hole for $v< v_0$, while at $v = v_0$ a perturbation is switched on, due to the infalling of shells of null radiation. After $v_1$, the perturbation stops, and the black hole settles to a stationary state again.

We report here for later use the relevant Christoffel symbols
\begin{align*}
    \Gamma^k_{v v} &= \frac{1}{2}g^{v k}\partial_v g_{v v} - \frac{1}{2}g^{k r} \partial_r g_{v v} \\
    \Gamma^k_{r r} &= 0 \\
    \Gamma^k_{r v} &= \frac{1}{2}g^{k v}\partial_r g_{v v} \\
    \Gamma^\theta_{\theta r} &= \Gamma^\varphi_{\varphi r} = \frac{1}{r} \\
    \Gamma^\theta_{\theta v} &= \Gamma^\varphi_{\varphi v} = 0
\end{align*}
and the $t-r$ part of the inverse metric
\begin{equation}
    g^{-1} =
    \begin{pmatrix}
    0& 1 \\
    1 & f \\
    \end{pmatrix} \ .
\end{equation}

Since the perturbation acts only in a finite amount of time, the spacetime is asymptotically flat, and therefore we can construct a Unruh state just as we did in Schwarzschild. Therefore, the expression for the variation of relative entropy between the Unruh state and a coherent perturbation, in Vaidya spacetime, will be the same as that we computed in Schwarzschild, in section \ref{sec:RE-Unruh}.

We consider a scalar perturbation over the Vaidya spacetime, which is not a source for the metric but acts as a perturbation of the quantities at second order, just as the scalar field in Schwarzschild. Now, the area of the horizon grows both because of the infalling matter (the source of Einstein equations) and because of the perturbation of the scalar field. This classical perturbation is a coherent state of the Unruh state, and so we can compute the relative entropy between the two states in the same way as we did in Schwarzschild. We choose as initial data surface  the union $\mathscr W \cup \mathscr I^-(v_0)$, where, as before, $\mathscr I^-(v_0)$ denotes the portion of past null infinity with $v < v_0$. We compute the relative entropy for the scalar field defined in the causal development of the initial data surface, $D(\mathscr W \cup \mathscr I^-(v_0))$, and we compute its derivative with respect to $v$ as we translate $v_0 \to v_0 + t$. If we then take another initial data surface $\mathscr W \cup \mathscr I^-(v_1)$, we compute the derivative of the relative entropy, and we finally take the difference of the derivatives of the two relative entropies, we find \eqref{RE-variation-Unruh-state}. This is the same procedure we applied in Schwarzschild, and since the Vaidya spacetime shares the same structure at infinity, where the relative entropy is computed, the result is the same.

\subsection{Apparent Horizons and the First Law of Apparent Horizons Mechanics} \label{ssec:AH-first-law}
The main reason to study the Vaidya case, is that it is the simplest example in which the event horizon and the apparent horizon are distinct objects. The event horizon is a causal boundary: the light-rays in its interior cannot reach future infinity, because of its very definition. For this reason, to locate an event horizon one needs to wait the end of the Universe: the event horizon has a \textit{teleological} nature, meaning that one needs to know the entire future history of the spacetime in order to locate the horizon.

For the same reason, the event horizon does not interact locally with the environment. This means that, for example, if a shell of dust falls into a black hole, the event horizon grows \textit{before} the shell crosses it, and while the matter passes into the interior of the black hole, the event horizon settles again to a stationary state.

For these reasons, it is useful to introduce a local notion of horizon, which causally interacts with the environment and that can be located using local observables. Such a horizon is called an \textit{apparent horizon}. It has been shown that apparent horizons obey laws similar to the laws of black hole mechanics, and therefore it is possible to assign to apparent horizon a definite temperature and entropy based on geometric reasoning \cite{Hayward97}. Therefore, it is natural to ask if the quantum approach we described in the last section applies also to apparent horizons. 

There are many different notions of apparent horizons in the literature. Here we adopted the notion introduced by Hayward \cite{Hayward93} of future outer, trapping surface. First, we need to define a \textit{trapped surface}. A trapped surface is a closed, 2-dimensional surface such that the expansion of both the future-directed, null geodesic congruences orthogonal to the surface is negative everywhere on the surface. Consider, for example, a spherical surface inside the event horizon of a Schwarzschild black hole, and consider the light-rays outgoing from the surface. In both the two normal directions, the congruence is converging, meaning the cross-sectional area of the congruence is decreasing for geodesics both pointing inward and outward the surface. Now, consider a 3-hypersurface $\Sigma$; the portion of $\Sigma$ formed by trapped surfaces is called \textit{trapped region}, and the boundary of this region is the apparent horizon of $\Sigma$. Now, if the apparent horizon can be extended to the past and future of $\Sigma$, the union of the apparent horizons on each hypersurface is called a \textit{trapping horizon}. In the following we will call apparent horizon both the 3-hypersurface and its 2-surface cross-section, just as we did for the event horizon.

Note that, by definition, the expansion of the future-directed null geodesics normal to the apparent horizon is vanishing. This condition helps in finding the location of the apparent horizon. If we distinguish the two normal directions to be ingoing and outgoing, the trapping horizon is \textit{outer} if $\eval{\mathcal L_-\theta_+}_{AH} < 0$, where the minus denotes the ingoing direction and the plus the outgoing one, and the Lie derivative is evaluated on the apparent horizon. Finally, the trapping horizon is said to be \textit{future} if $\eval{\theta_- }_{AH} < 0$.

It is possible to formulate the laws of black hole mechanics in terms of apparent horizons, instead of event horizons \cite{Hayward93}, \cite{Hayward97}. In particular, the physicist S. A. Hayward showed \cite{Hayward97} that the laws of black hole mechanics can be reformulated in a fully dynamical context, without assuming asymptotic flatness nor stationarity, referring to geometric quantities associated to apparent horizons. Remarkably, he showed that the first law of thermodynamics and the first law of apparent horizon mechanics can be given in a unified expression, as an equation for the variation of a suitably defined notion of energy, thus reinforcing the analogy between geometric and thermodynamic quantities.

 In the absence of asymptotic flatness one needs to find a quasi-local notion of energy, since the ADM or Bondi energies (the most used definitions for the energy of a gravitational system in an asymptotically flat spacetime) are both defined as surface integrals at infinity. In the case of spherical symmetry one can use the Misner-Sharp energy,
\begin{equation}
E = \frac{1}{2} r ( 1 - \nabla_\mu r \nabla^\mu r ) \ .
\end{equation}
In Vaidya spacetime, for example, the Misner-Sharp energy reduces to $m(v)$ on the horizon, thus being the intuitive notion of black hole mass. Then, one can compute the variation of the Misner-Sharp mass along a vector tangent to the apparent horizon, and relate it to the variation of the geometric properties of the horizon via the Einstein equations. Since in spherical symmetry the area of the apparent horizon is a geometrical invariant, one can use it to define the areal radius and volume,
\begin{equation}
A = 4 \pi r^2 \quad V = \frac{4}{3} \pi r^3  \ .
\end{equation}
Finally, one needs to introduce the function
\begin{equation}
w = - \frac{1}{2} \Tr T \ ,
\end{equation}
where the trace is taken with respect to the 2-dimensional space normal to the spheres of symmetry (in the case of Vaidya, the $v-r$ plane).

Given these definitions, if one denote $\dv{z}$ the derivative along a vector tangent to the apparent horizon, the first law of apparent horizon dynamics is
\begin{equation} \label{eq:unified-first-law}
\dv{E}{z} = \frac{\kappa}{8 \pi} \dv{A}{z} + w \dv{V}{z} \ ,
\end{equation}
where $\kappa$ is the analogue of the surface gravity for an apparent horizon. This is a "differential version" of the first law, in contrast with the "integral form" we gave for the first law of black hole event horizons, \eqref{I-law-EH}. The last term is usually interpreted as a work term done by the matter fields on the black hole. For example, in the context of Einstein-Maxwell equations (that is, in systems in which both the gravitational and the electromagnetic fields are present), this work term coincides with the usual electromagnetic work defined in Special Relativity, and thus it can be interpreted as the work done by the electromagnetic field on the black hole.
  
  \begin{figure}
  \centering
  \includegraphics[scale=0.3]{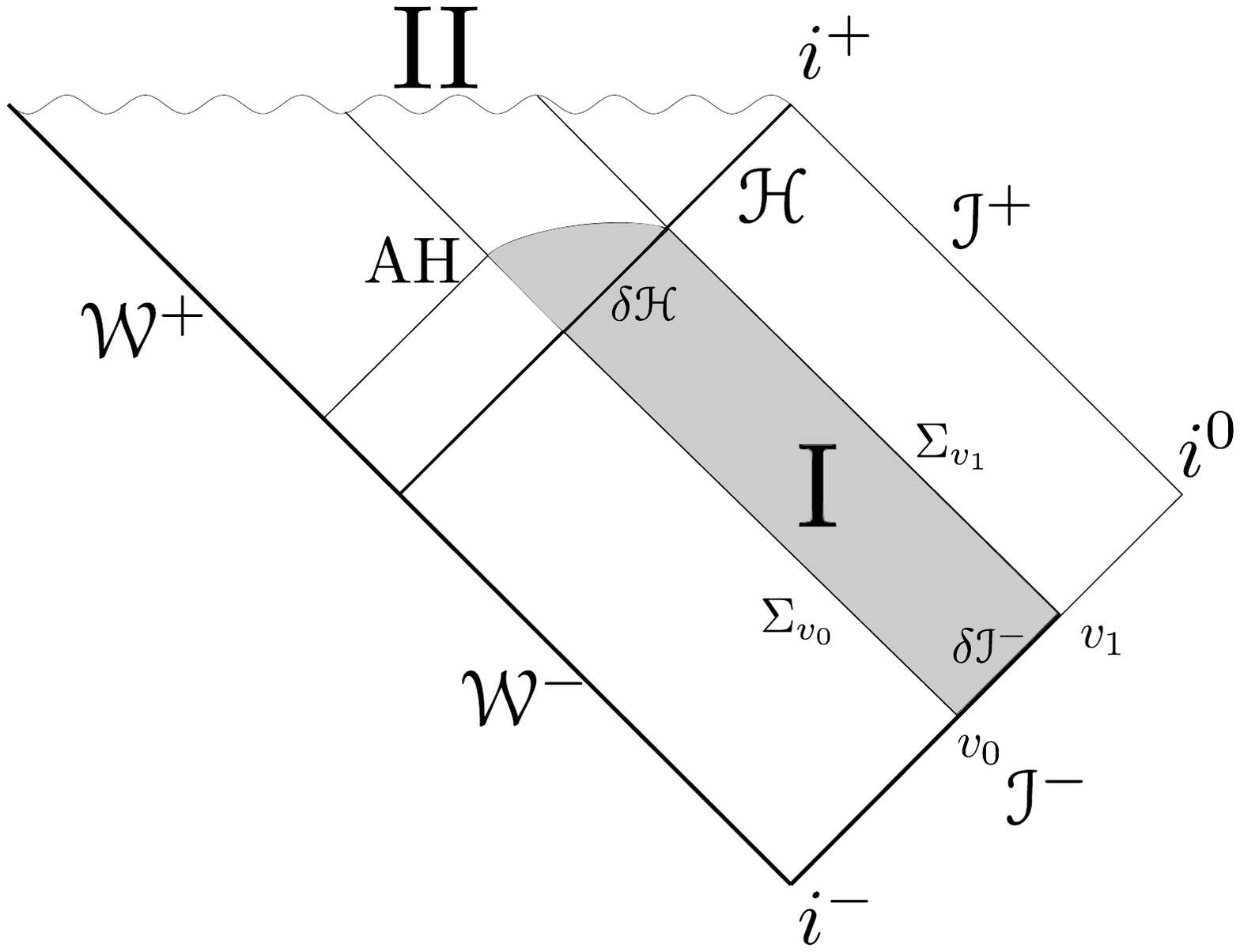}
  
  \caption{Schematic representation of the behaviour of an apparent horizon versus the event horizon, under the infalling of null dust in a time interval $v_0 < v < v_1 $ (a Vaidya model). Since the event horizon is always a null surface, it never changes its slope in a Penrose diagram, while the apparent horizon grows and becomes space-like during the accretion process. At the end of the accretion the apparent horizon becomes stationary again, thus coinciding with the Schwarzschild event horizon. We see that the apparent horizon behaves locally, in contradistinction with the event horizon.}
  \end{figure}
Let's see how the definition of apparent horizon applies to find its location in the Vaidya spacetime. First of all, we need to find the null geodesics congruences. From the metric, we see that a pair of future-oriented, outgoing and ingoing null vectors is given by
\begin{align} \label{eq:null-vectors-Vaidya}
    \tilde{l} &= \partial_v + \frac{1}{2} \bigg( 1 - \frac{2m(v)}{r} \bigg ) \partial_r \\
    N &= \partial_r \ .
\end{align}
Computing the covariant derivative of $N$, we find
\begin{equation}
   ( N^\mu \nabla_\mu N)^\nu  = \Gamma^\nu_{r r} = 0 \  ,
\end{equation}
so it is affinely parameterised. On the other hand, computing the covariant derivative of $\tilde{l}$, we find
\begin{equation}
    \tilde{l}^\mu \nabla_\nu \tilde{l}^\nu = \nabla_v\tilde{l}^nu + \frac{1}{2} f \nabla_r \tilde{l}^\nu \ .
\end{equation}
For $\nu = v$ we have
\begin{equation}
    \Gamma^v_{v k}\tilde{l}^k + \frac{1}{2}f \Gamma^v_{r k} \tilde{l}^k = - \frac{1}{2}\partial_r g_{vv} = \frac{1}{2} \partial_r f \ ,
\end{equation}
while for $\nu = r$
\begin{equation}
    \frac{1}{2}\partial_v f + \Gamma^r_{vv} + \frac{1}{2}f \Gamma^r_{vr} + \frac{1}{2}f (\frac{1}{2}\partial_r f + \Gamma^r_{rv}) = \frac{1}{2}f (\frac{1}{2} \partial_r f) \ .
\end{equation}
From which we see that $\tilde{l}$ is not affinely parameterised, with inaffinity function $\tilde{\kappa} = \frac{1}{2} \partial_r f = \frac{m(v)}{r^2}$. We have seen in section \ref{ssec:hypersurfaces} that we can construct a tangent vector, $l$, with affine parameter $\lambda$, proportional to $\tilde{l}$, with proportionality factor in the form
\begin{equation}
    l = e^{- \gamma(r,v)}\tilde{l} \ .
\end{equation}
Computing its covariant derivative, and imposing that it vanishes, we found that
\begin{equation} \label{eq:gamma-definition}
\dv{\gamma}{v} = \tilde{\kappa}(v) \ .
\end{equation}

We can find the relation between the affine parameter $\lambda$ and $v$ reasoning that, by chain rule,
\begin{equation}
    \tilde{l}(v) = \dv{\lambda}{v}l(\lambda) \ .
\end{equation}
So we have $\dv{\lambda}{v} = e^{ \gamma}$, or
\begin{equation} \label{eq:v-affineparameter}
\dv{\lambda}{v} = e^{\int_{v_0}^v \tilde{\kappa}(v') \dd v'} \ .
\end{equation}

As we are eventually interested in integrating over the null geodesic congruence generated by outgoing null rays, we fix $\lambda = 0$ over a null hypersurface $v = v_0$, so to use the same affine parameter over the congruence.

To locate the apparent horizon, we need to compute the expansion of the outgoing geodesics. The three conditions for the apparent horizon are
\begin{align}
    \theta_{l} = 0 \\
    \theta_N < 0 \\
    \mathscr L_N \theta_l < 0 \ .
\end{align}

The expansion of outgoing geodesics is given by
\begin{equation}
\theta_l = \nabla_\alpha l^\alpha = \nabla_\alpha e^{-\gamma} \tilde l^\alpha = \tilde l^\alpha \partial_\alpha e^{-\gamma} + e^{-\gamma} \nabla_\alpha \tilde l^\alpha = e^{-\gamma} ( \nabla_\alpha \tilde l^\alpha - \tilde l^\alpha \partial_\alpha \gamma ) \ .
\end{equation}

We compute $\nabla_\alpha \tilde l^\alpha = \theta_{\tilde l} $:
\begin{equation}
\theta_{\tilde l} = \nabla_\alpha \tilde l^\alpha = \frac{1}{2}\partial_r f + \Gamma^\mu_{\mu v} + \frac{1}{2} f \Gamma^\mu_{\mu r} = \frac{1}{2}\partial_r f + \frac{1}{r} f =  \frac{1}{r} - \frac{2m}{2r^2} = \frac{2}{r^2}(r-m) \ .
\end{equation}
On the other hand, we have
\begin{equation}
\tilde l^\alpha \partial_\alpha \gamma = \tilde \kappa = \frac{m}{r^2} \ ,
\end{equation}
so that the location of the apparent horizon is 
\begin{equation} \label{eq:AH-Vaidya}
   r_{AH} = 2m(v) \ .
\end{equation}
The defining equation for this hypersurface is $S: r-r_{AH} = 0$, from which we find the gradient of the function which define the hypersurface, $\nabla S = (-2 \dot m(v), 1, 0, 0)$, where the dot denotes a derivative with respect to $v$. From it we compute the contravariant normal vector,
\begin{equation}
    \tilde n_{AH} = g^{\mu \nu} \partial_\nu S \partial_\mu = g^{vr} \partial_v - 2 \dot m(v) g^{vr} \partial_r= \partial_v - 2 \dot m \partial_r \ .
\end{equation}
Since $g^{rr} = f = 0$ on the horizon.

We see that $\norm{n_{AH}}^2 = - 4 \dot m(v)$ on the apparent horizon. This means that the apparent horizon is a space-like hypersurface if the black hole is growing and the matter is classical, it is a time-like hypersurface in the case of an evaporating black hole, and becomes null in the static limit $\dot m \to 0$. In the static limit, we recover the Schwarzschild spacetime, and the apparent and event horizons coincide.

We normalise the vector to get the unit normal vector of the apparent horizon:
\begin{equation} \label{eq:normal-vector-AH}
    n_{AH} =  \frac{1}{\sqrt{4 \dot m}} \partial_v - \sqrt{\dot m} \partial_r \ .
\end{equation}
    
    \subsection{Kodama Miracle} \label{ssec:kodama-miracle}

Having found the analogue of the event horizon for dynamical black holes, we need now to find the analogue of the Killing conservation law. Since the metric is $v$-dependent, there is no time-like Killing vector analogue to the $\partial_t$ vector we used in Schwarzschild. However, it is possible to define a time-like vector field which is not associated to any symmetry of the metric, and yet defines a conserved quantity and a preferred direction for the flow of time in any time-dependent, spherically symmetric spacetime \cite{Kodama79}, \cite{Abreu10}.

The \textit{Kodama vector} is defined as
\begin{equation}
k^\alpha = \epsilon^{\alpha \beta}_2 \nabla_\beta r \ ,
\end{equation}
where $\epsilon^{\alpha \beta}_2 $ is the $(1+1)$-dimensional Levi-Civita symbol in the radial temporal plane, embedded in $(3+1)$ dimensions as
\begin{equation}
\epsilon^{\alpha \beta}_2 =  \begin{pmatrix}
\epsilon^{ij} & 0 \\
0 & 0 \\
\end{pmatrix}
\end{equation}  
where $i,j = 0,1$. In our case of a Vaidya metric, the Kodama vector takes the simple form
\begin{equation}
k = \partial_v \ .
\end{equation}
The Kodama vector is null on the apparent horizon and it is time-like for $r \to \infty$, with $\norm{k}^2 \to -1$. We can say, then, that the Kodama vector in dynamical spacetimes takes the same role as the Killing vector in stationary spacetimes; in particular, it allows us to define a conserved current.

First, we show that the Kodama vector is divergence-free:
\begin{equation}
\nabla_\alpha k^\alpha = \nabla_\alpha ( \epsilon^{\alpha \beta}_2 \nabla_\beta r ) = \epsilon^{\alpha \beta}_2 \nabla_\alpha \nabla_\beta r + \nabla_\alpha \epsilon^{\alpha \beta}_2 \nabla_\beta r = 0 \ .
\end{equation}
The first term vanishes because it is the product of a symmetric tensor and an antisymmetric tensor. The second one vanishes because, since the Levi-Civita tensor is constant, it is the sum of contractions between the Christoffel symbols and the Levi-Civita tensor, which are vanishing again for symmetry. Now, a conserved quantity is given by
\begin{equation}
J_\mu = T_{\mu \nu} k^\nu \ .
\end{equation}
The divergence of $J_\mu$ vanishes because the divergences of the stress-energy tensor and the Kodama vector vanishes separately. We see that the Kodama conservation law is not a consequence of some symmetry, but it comes from its algebraic  properties: this is why it is called \textit{Kodama miracle} \cite{Abreu10}, which holds only in spherically symmetric spacetimes. 

Thanks to this conservation law, we can use the same framework we considered in Schwarzschild: we construct a finite region in the outside of the Vaidya spacetime with the boundary given by four hypersurfaces: two null hypersurfaces at $v = v_0$ and $v = v_1$, the corresponding portion of conformal past infinity $\delta \mathscr I^-$ with $v \in [v_0, v_1]$, and the portion of apparent horizon $\delta \mathscr H_{AH}$ with $v \in [v_0, v_1]$. The difference from the Schwarzschild case is that this is a "deformed" double cone, because the apparent horizon is a space-like hypersurface. We then transform the Kodama conservation law into a flux integral equation, by means of the Gauss-Stokes theorem, finding
\begin{equation}
\begin{split}
    &\int_{\delta \mathscr H_{AH}} T_{\mu \nu} k^\mu n^\nu_{AH} \ \dd v \dd \mathbb S^2 
    + \int_{\Sigma_{v_1}} T_{\mu \nu} k^\mu n^\nu_{v_1} \ \dd r \dd \mathbb S^2 = \\
    &=  \int_{\Sigma_{v_0}}  T_{\mu \nu} k^\mu n^\nu_{v_0} \ \dd r \dd \mathbb S^2  
    + \int_{\delta \mathscr I^-} T_{\mu \nu} k^\mu n^\nu_{\mathscr I^-} \ \dd v \dd \mathbb S^2 \ .
\end{split}
\end{equation}

We can write it more explicitly. In fact, the normal vector to conformal past infinity is $l$, while we already computed the unit normal vector to the apparent horizon, \eqref{eq:normal-vector-AH}. The gradient of the two hypersurfaces at fixed $v$ is $\nabla S = \partial_v$, and the normal vector is $n_{v} = g^{\mu \nu} \partial_\mu S \partial_\nu = \partial_r = N  $. The Kodama current is $J_\mu = T_{\mu v} k^\nu$, and substituting in the conservation law we find
\begin{equation} \label{eq:flux-Vaidya}
\begin{split}
    &\int_{\delta \mathscr H_{AH}} \bigg( \frac{1}{\sqrt{4 \dot m}} T_{vv} - \sqrt{\dot m} T_{vr} \bigg) \ \dd \lambda \dd \mathbb S^2 + \int_{\Sigma_{v_1}} T_{vr} \ \dd r \dd \mathbb S^2 = \\
    &=  \int_{\Sigma_{v_0}} T_{vr} \ \dd r \dd \mathbb S^2  
    + \int_{\delta \mathscr I^-} \bigg( T_{vv}+ \frac{1}{2} f(v,r) \ T_{vr} \bigg) \ \dd v \dd \mathbb S^2 \ .
\end{split}
\end{equation}

We note that the possible confusion between the first integral in the above equation and the integration involved in the first law of apparent horizons \eqref{eq:unified-first-law}, on which we briefly commented at the end of subsection \ref{ssec:killing-conservation-law}, does not arise here. In fact, since the apparent horizon is a space-like surface, the distinction between the normal and tangent directions is clear: the first law of apparent horizons dynamics regards the variations in energy along the horizon, while here we are integrating the component of the conserved current normal to the horizon.

Now, as we did in the Schwarzschild case, we want to evaluate the term on the apparent horizon to connect it to the variation of the horizon area, via the Raychaudhuri equation, and on the conformal past infinity, to connect it to the variation of relative entropy between the Unruh state and its classical perturbation.

    \subsection{Apparent Horizon term} \label{ssec:AH-flux-term}
Once again, we want to transform the surface integral of the stress-energy tensor on the apparent horizon to an integral of a geodesic expansion, in order to find the variation of the horizon area. However, since the apparent horizon is not a null hypersurface, we cannot use the expansion of its null generators. Since it is space-like, there is a congruence of outgoing null geodesics that cross the horizon once, and therefore the idea is to integrate the expansion of the null congruence over the apparent horizon. The expansion of the null congruence takes into account the variation of the degenerate, 2-dimensional metric transverse to the null geodesics, while the surface element on the apparent horizon is a 3-dimensional metric. We can anticipate, then, that the simplification which occurred in the case of Schwarzschild between the logarithmic derivative of the transverse metric and the surface element, which we saw in subsection \ref{ssec:area-term-Schw}, computing equation \eqref{eq:flux-term-area-Schw}, will not be possible, and therefore the area variation will be more complicated. This is to be expected, because while the scalar field perturbs the metric and therefore causes some variation of the apparent horizon, the source matter of the Vaidya metric falls into the black hole, and contributes to the variation of the area as well. Since the induced metric itself depends on the mass function, as we will see in a moment, there is no way \textit{a priori} to split the variation in the horizon area due to the scalar field from that of the matter source, and therefore we will not compute the integral of the second order variation of the expansion, $\delta^2 \theta_l$, but of the full expansion $\theta_l$.
 
As we saw, by definition $\theta_l = 0$ on the horizon, while $\eval{l}_{AH} = e^{-\gamma} \partial_v $, where $\gamma$ is defined in equation \eqref{eq:gamma-definition}. The Raychaudhuri equation then gives the variation of the expansion of the outgoing null geodesics,
\begin{equation}
    \dv{ \theta_l}{\lambda} = - 8 \pi T_{\mu \nu} l^\mu l^\nu = - 8 \pi e^{-2 \gamma} T_{vv} \ .
\end{equation}
The flux term becomes, then,
\begin{equation} \label{eq:flux-element-AH}
   \int_{AH} T_{\mu \nu}n^\mu_{AH}k^\nu \dd \lambda \dd \mathbb S_2 =
    - \frac{1}{8 \pi} \int_{v_0}^{v_1} \frac{1}{\sqrt{4 \dot m}} e^{2 \gamma}\dv{ \theta_l}{\lambda} \ \dd \lambda \dd \mathbb S_2 
    -  \int_{AH}\sqrt{\dot m} T_{v r} \dd \lambda \dd \mathbb S_2 \ .
\end{equation}
We note that it reduces to the static form we computed in the Schwarzschild case in subsection \ref{ssec:area-term-Schw} in the limit $\dot m \to 0$.

Now we have to write explicitly the surface element $\dd \lambda \dd \mathbb S_2$. The parametric equation for the apparent horizon with parameter $\lambda$ is
\begin{align*}
v = v(\lambda) \\
r = 2m(v(\lambda)) \\
\theta = \theta \\
\varphi = \varphi
\end{align*}
From these we can construct the vectors $e^\alpha_a = \pdv{x^\alpha}{y^a}$, where the lower index refers to the coordinates on the surface. The relevant one is 
$$e^\alpha_v = \big( e^{-\gamma}, \ e^{-\gamma} \ 2\dot m(v), \ 0, \ 0 \big ) \ .$$
The induced metric is given by the inner product
\begin{equation}
h_{a b} = e^\mu_a e^\nu_b g_{\mu \nu} \ .
\end{equation}
The only non-trivial product is $e^\alpha_\lambda e^\beta_\lambda g_{\alpha \beta} =e^{-2 \gamma} 4 \dot m(v)$. The induced metric then is
\begin{equation}
h = \begin{pmatrix} e^{-2 \gamma} 4\dot m(v)& 0 & 0 \\ 0 & r_{AH}^2 & 0 \\ 0 & 0 & r^2_{AH} \sin^2 \theta \end{pmatrix} \ .
\end{equation}

The surface element is
\begin{equation} \label{eq:surface-element-AH}
\sqrt h_{AH} = e^{-\gamma} \sqrt{4 \dot m(v)} r^2_{AH}\sin \theta \ ,
\end{equation}
where $r_{AH}$ is a function of $v$.

On the other hand, the expansion of the outgoing geodesics $\theta$ is the logarithmic derivative of the surface element of the area transverse to the null geodesics, and so it is 2-dimensional. The transverse surface has parametric equations $v = v_0$, $r=r$, $\theta = \theta$, $\varphi = \varphi$, so that the induced metric is the two dimensional angular metric:
\begin{equation}
    \sigma = \begin{pmatrix} r_{AH}^2 & 0 \\ 0 & r_{AH}^2 \sin^2 \theta \end{pmatrix} \ ,
\end{equation}
and the surface element is
\begin{equation} \label{eq:surface-element-null-vaidya}
    \sqrt{\sigma} = r_{AH}^2 \sin \theta \ .
\end{equation}
Using the geometric meaning of the expansion \eqref{eq:expansion-area-light-like},
$$
\theta_l = \frac{1}{\sqrt \sigma}\dv{\sqrt \sigma}{\lambda} \ ,
$$
we can rewrite the first integral in the right hand side of the flux element, \eqref{eq:flux-element-AH}. We substitute the expression for the surface element,
\begin{equation}
  - \frac{1}{8 \pi} \int_{v_0}^{v_1} \frac{1}{\sqrt{4 \dot m}} e^{2 \gamma}\dv{ \theta_l}{\lambda} \ \dd \lambda \dd \mathbb S_2 =
   - \frac{1}{8 \pi} \int_{\lambda_0}^{\lambda_1} \int_{S^2} e^\gamma \dv{\theta_l}{\lambda} \ r^2_{AH} \sin \theta \dd \lambda \dd \theta \dd \varphi \ ,
\end{equation}
and of the expansion, to get
\begin{multline}
 - \frac{1}{8 \pi} \int_{\lambda_0}^{\lambda_1} \int_{S^2} e^\gamma \dv{\theta_l}{\lambda} \ r^2_{AH} \sin \theta \dd \lambda \dd \theta \dd \varphi = \\
 = - \frac{1}{8\pi} \int_{\lambda_0}^{\lambda_1} \int_{S^2} e^\gamma \dv{\lambda} \bigg( \frac{1}{\sqrt \sigma} \dv{\sqrt{\sigma}}{\lambda} \bigg) r^2_{AH} \sin \theta \dd \lambda \dd \theta \dd \varphi = 
- \frac{1}{8\pi} \int_{\lambda_0}^{\lambda_1} \int_{S^2} e^\gamma \dv[2]{ \sqrt{\sigma}}{\lambda} \dd \lambda \dd \theta \dd \varphi \ .
\end{multline}
The equation follows since
\begin{equation}
\dv{\lambda} \frac{1}{\sqrt \sigma} \dv{\sqrt{\sigma}}{\lambda} = \frac{1}{\sqrt \sigma} \dv[2]{ \sqrt{\sigma}}{\lambda} - \bigg ( \frac{1}{\sqrt \sigma} \dv{\sqrt{\sigma}}{\lambda} \bigg)^2 = \frac{1}{\sqrt \sigma} \dv[2]{ \sqrt{\sigma}}{\lambda} - \theta_l^2 \ ,
\end{equation}
while $\theta_l = 0$ on the horizon, and using the explicit expression for the determinant of the induced metric, \eqref{eq:surface-element-null-vaidya}. The first integral of the flux on the horizon is therefore in the same form as the integral we computed in Schwarzschild, in \eqref{eq:area-integral-S}. The difference now is that integrating by parts gives a dynamical correction,
\begin{equation}
\begin{split}
&- \frac{1}{8\pi} \int_{\lambda_0}^{\lambda_1} \int_{S^2} e^\gamma \dv[2]{ \sqrt{\sigma}}{\lambda} \dd \lambda \dd \theta \dd \varphi = \\
&= - \frac{1}{8\pi}  \int_{S^2} \bigg[e^\gamma \dv{ \sqrt{\sigma}}{\lambda}\bigg ]^{\lambda_1}_{\lambda_0} \dd \theta \dd \varphi 
+ \frac{1}{8 \pi} \int_{\lambda_0}^{\lambda_1} \int_{S^2} \dot \gamma \dv{ \sqrt{\sigma}}{\lambda} \dd \lambda \dd \theta \dd \varphi = \\
&= - \frac{1}{8\pi} \dv{v}  \int_{S^2} \eval{\sqrt{\sigma}}^{v_1}_{v_0} \dd \theta \dd \varphi 
+ \frac{1}{8\pi}  \int_{\lambda_0}^{\lambda_1} \int_{S^2} \dot \gamma  \dv{\sqrt \sigma}{\lambda} \ \dd \lambda \dd \theta \dd \varphi \ ,
\end{split}
\end{equation}
where we have $\dot \gamma = \dv{\gamma}{v} = e^\gamma \dv{\gamma}{\lambda}$, and we used $e^\gamma \dv{\lambda} = \dv{v}$, passing from $\lambda_{0,1}$ to $v_{0,1}$.
In the static limit the second term is vanishing and we recover the Schwarzschild result.

The first integral, in the last line of the above expression, can be integrated immediately: the integral of the surface element over the 2-sphere gives the difference in area of the horizon between $\lambda_1$ and $\lambda_0$. In the second integral if we multiply by $e^{-\gamma} e^\gamma$ we can use the two facts that $e^\gamma \dv{\lambda} = \dv{v}$ and $e^{-\gamma} \dd \lambda = \dd v$, to switch to an integral in terms of $v$. We get
\begin{multline} \label{eq:area-term-vaidya-area-mass-terms}
- \frac{1}{8\pi} \dv{v}  \int_{S^2} \eval{\sqrt{\sigma}}^{v_1}_{v_0} \dd \theta \dd \varphi 
+ \frac{1}{8\pi}  \int_{\lambda_0}^{\lambda_1} \int_{S^2} \dot \gamma  \dv{\sqrt \sigma}{\lambda} \ \dd \lambda \dd \theta \dd \varphi = \\
- \frac{1}{8 \pi} \bigg( \dv{A(v_1)}{v} - \dv{A(v_0)}{v} \bigg ) + \frac{1}{8 \pi} \int_{v_0}^{v_1} \dot \gamma \dv{\sqrt \sigma}{v} \ \dd v \dd \theta \dd \varphi \ .
\end{multline}
Now, we focus on the remaining integral. We have by definition $\dot \gamma = m/r^2 = (4m)^{-1}$, while the derivative of the determinant of the induced metric is $\dv{\sqrt \sigma}{v} = 8 m \ \dot m \sin \theta$. Therefore, the second integral gives
\begin{equation}
\frac{1}{8 \pi} \int_{v_0}^{v_1} \dot \gamma \dv{\sqrt \sigma}{v} \ \dd v \dd \theta \dd \varphi = 
\frac{1}{4 \pi} \int_{v_0}^{v_1} \int_{S_2} \dot m(v) \ \sin \theta \dd v \dd \theta \dd \varphi \ ,
\end{equation}
which is nothing but the integral of the source stress-energy tensor, \eqref{eq:SET-source-Vaidya}. The integral is now immediate: the integral along $v$ gives the difference in mass of the black hole between $v_1$ and $v_0$, and the angular integral gives a factor of $4 \pi$; the result is
\begin{equation} \label{eq:area-term-vaidya-mass-term}
\frac{1}{8 \pi} \int_{v_0}^{v_1} \dot \gamma \dv{\sqrt \sigma}{v} \ \dd v \dd \theta \dd \varphi = 
\frac{1}{4 \pi} \int_{v_0}^{v_1} \int_{S_2} \dot m(v) \ \sin \theta \dd v \dd \theta \dd \varphi = m(v_1) - m(v_0) \ .
\end{equation} 

Putting all the pieces together, we substitute, in the flux term on the apparent horizon, \eqref{eq:flux-element-AH} the two terms we computed, \eqref{eq:area-term-vaidya-area-mass-terms} and \eqref{eq:area-term-vaidya-mass-term}, keeping the last integral in \eqref{eq:flux-element-AH} unmodified for now, to find that the flux integral on the apparent horizon is
\begin{multline}
 \int_{AH} T_{\mu \nu}n^\mu_{AH}k^\nu \dd v \dd \mathbb S_2 = \\ - \frac{1}{8 \pi} \bigg( \dv{A(v_1)}{v} - \dv{A(v_0)}{v} \bigg ) 
 + m(v_1) - m(v_0) 
 - \int_{AH}\sqrt{\dot m} T_{v r} \dd \lambda \dd \mathbb S_2  \ .
\end{multline}
We can now comment on the last integral in the above equation. If we substitute the surface element of the apparent horizon, \eqref{eq:surface-element-AH}, the integral becomes
\begin{equation}
- \int_{AH}\sqrt{\dot m} T_{v r} \dd \lambda \dd \mathbb S_2 = - \int_{AH} T_{vr} 2 \dot m(v) r^2_{AH} \sin \theta \dd v \dd \theta \dd \varphi \ .
\end{equation}  

The function $w$ that we introduced discussing the first law of apparent horizons, equation \eqref{eq:unified-first-law} in subsection \ref{ssec:AH-first-law}, in the Vaidya case becomes
\begin{equation}
w = -\frac{1}{2} \Tr_{v-r} T = - T_{vr} \ .
\end{equation}
On the other hand, if we consider the volume of the black hole, $ V = \frac{4}{3} \pi r_{AH}^3$, we see that the derivative of the areal volume is
\begin{equation}
\dv{V}{v} = 4 \pi r_{AH}^2 \dot r_{AH} = 8 \pi \dot m(v) r_{AH}^2 \ .
\end{equation}
Therefore, the remaining flux term is the integral of the work term which appears in the first law of apparent horizons, \eqref{eq:unified-first-law}, and we can compute the angular integration to get
\begin{equation} \label{eq:work-term-vaidya}
- \int_{AH} T_{vr} 2 \dot m(v) r^2_{AH} \sin \theta \dd v \dd \theta \dd \varphi = \int_{AH} w \dv{V}{v} \dd v \ .
\end{equation}

Finally, the flux term on the apparent horizon is
\begin{multline}\label{eq:flux-area-Vaidya-result}
 \int_{AH} T_{\mu \nu}n^\mu_{AH}k^\nu \dd v \dd \mathbb S_2 = \\ - \frac{1}{8 \pi} \bigg( \dv{A(v_1)}{v} - \dv{A(v_0)}{v} \bigg ) 
 + m(v_1) - m(v_0) + \int_{AH} w \dv{V}{v} \dd v \ .
\end{multline}

The area term differs from that of Schwarzschild, \eqref{eq:flux-term-area-Schw}, for two contributions. The first comes from the fact that the Vaidya black hole is not in vacuum, but there is matter falling into the black hole in the interval $[v_0, v_1]$, which contributes to the flux term with the mass contribution $m(v_1) - m(v_0)$: the black hole mass changes during the process because of the infalling matter. The second term comes from the nature of apparent horizons, in contradistinction to that of event horizons. In a sense, apparent horizons show some compressibility, since it is possible to do some work on them, while the event horizon acts as an incompressible barrier. This is the last term in the above expression.

    \subsection{Variation of Generalised Entropy}

Having computed the contribution on the apparent horizon, we need to evaluate the flux integral at conformal past infinity, to relate the variation in horizon area to the variation of relative entropy between the Unruh state and its coherent perturbation. This is done just as in the Schwarzschild case. As we said, the relative entropy is computed exactly in the same way, and so we recover the same result
\begin{equation}
\frac{1}{2 \pi} \dv{t}(S_{v_1} - S_{v_0}) = \int_{v_0}^{v_1} \int_{S^2} (\partial_v \tilde \psi)^2 \dd v \dd \Omega_2 \ ,
\end{equation}
with the same symbols as in section \ref{sec:RE-Unruh} on the relative entropy between the Unruh state and a coherent perturbation.

Now, we compute the flux element at past infinity, and see the connection with the relative entropy. As $f \to 1$ at infinity, the flux element is
\begin{equation}
     \int_{\delta \mathscr I^-} ( T_{v v} + \frac{1}{2} \ T_{v r} ) r^2 \dd v \dd \Omega_2 \ .
\end{equation}
The stress-energy tensor is the sum of two contributions, $T_{\mu \nu} = T_{\mu \nu}^\psi + T_{\mu \nu}^{source}$. The flux integral of the source stress-energy tensor is simply $m(v_1) - m(v_0)$. This accounts for the flux at infinity of the source term, which is given by the difference in mass of the black hole: since the source is actually infalling radiation moving along $v$, the energy density which one adds at infinity inevitably falls into the black hole.

Regarding the field stress-energy tensor, we see that, at infinity,
\begin{multline}
    T_{v r} = \partial_v \psi \partial_r \psi 
    - \frac{1}{2} ( g^{r r} (\partial_r \psi)^2 + g^{v v} \partial_v \psi \partial_v \psi + 2 g^{v r} \partial_r \psi \partial_v \psi + \sigma^{a b} \partial_a \psi \partial_b \psi) = \\
    = - \frac{1}{2} (\partial_r \psi)^2 - \frac{1}{2} \sigma^{a b} \partial_a \psi \partial_b \psi = 0 \ .
\end{multline}
Since $g^{v v} = 0$, $\partial_r \psi \sim 1/r^2$, and $\sigma^{a b} \sim 1/r^2$. On the other hand,
\begin{equation}
    T_{v v} = (\partial_v \psi)^2 - \frac{1}{2}g_{v v} (g^{r r} (\partial_r \psi)^2 + 2 g^{v r} \partial_r \psi \partial_v \psi) = 
    (\partial_v \psi)^2 \ ,
\end{equation}
as $g_{v v} \to -1$, and using the same asymptotic behaviour as before. Once again, we can rewrite the integrand in terms of $\tilde \psi = r \psi$ as $r^2 (\partial_v \psi)^2 \simeq (\partial_v \tilde \psi)^2$, with the remaining terms vanishing at infinity. Using the expression for the derivative of the relative entropy we arrive at the result
\begin{equation} \label{eq:flux-RE-Vaidya}
    \int_{\mathscr I^-} ( T_{v v} + \frac{1}{2} \ T_{v r} ) r^2 \dd v \dd \Omega_2 = \frac{1}{2 \pi} \dv{v} (S_{v_1} - S_{v_0}) + m(v_1) - m(v_0) \ . 
\end{equation}

We can now rewrite the conservation law \eqref{eq:flux-Vaidya} substituting all the terms we computed, the flux term on the apparent horizon \eqref{eq:flux-area-Vaidya-result} and the flux term at past infinity \eqref{eq:flux-RE-Vaidya}, arriving at
\begin{multline} \label{eq:result-RE-Vaidya}
\frac{1}{2 \pi} \dv{v} \bigg[ S_{v_1}(\omega | \omega_\psi) + \frac{1}{4} A(v_1) - \bigg( S_{v_0}(\omega | \omega_\psi) + \frac{1}{4} A(v_0) \bigg) \bigg] = \\
 \int_{AH} w \dv{V}{v} \dd v + \int_{\Sigma_{v_1}} T_{vr} \ \dd r \dd \mathbb S^2 - \int_{\Sigma_{v_0}} T_{vr} \ \dd r \dd \mathbb S^2 \ .
\end{multline}
The derivative with respect to $v$ as the derivative of the quantities inside squared brackets, evaluated at $v_1$ (the first two terms) and at $v_0$ (the last two terms, in round brackets).

We see that the contribution $m(v_1) - m(v_0)$ on the apparent horizon is compensated by the flux at infinity of the source stress-energy tensor, and we find again that a variation of the relative entropy term is accompanied by a variation of one-fourth the area of the apparent horizon. Again, then, we can associate an entropy to apparent horizons just as for event horizons, equals to one-fourth its area. The main difference from the Schwarzschild case is the work term, \eqref{eq:work-term-vaidya}, which vanishes in the static limit. We can then introduce the generalised entropy, with the same form as in Schwarzschild, \eqref{eq:generalised-entropy-Schw}, where, now, the area term is related to the apparent horizon. If we consider the limit $v_0 \to - \infty$, the same considerations we did in the Schwarzschild case, in subsection \ref{ssec:generalised-entropy-Schw}, follow: the relative entropy associated to the region $D(\mathscr W \cup \mathscr I^-(v_0)$ goes to zero, since the surface integral in $v$ shrinks to a 2-sphere, $S_{v_0}(\omega|\omega_\psi) \to 0$. At the same time $\dv{A(v_0)}{v} \to 0$ for $v_0 \to -\infty$, since the black hole is stationary at past infinity. Finally, the flux integral on $\Sigma_{v_0}$ vanishes because we chose vanishing initial data on $\mathscr W$. Therefore, we can write \eqref{eq:result-RE-Vaidya}, in the limit $v_0 \to -\infty$, as
\begin{equation}
\frac{1}{2 \pi} \dv{v} S_{gen}(v_1) = \int_{AH} w \dv{V}{v} \dd v + \int_{\Sigma_{v_1}} T_{vr} \ \dd r \dd \mathbb S_2 \ .
\end{equation}
Again, this result coincide with the Schwarzschild case, \eqref{eq:result-RE-Schw-limit}, apart from the work term, the first integral in the right-hand side of the above equation.

\section{Relative Entropy for Cosmological Horizons}

    \subsection{Asymptotically Flat FLRW spacetimes}
    
As a final application we want to show that the same framework can work in the context of cosmological horizons. After the discovery of Hawking radiation, it was shown that it is possible to associate an entropy and a temperature to event horizons which arise in the context of cosmological spacetimes \cite{Gibbons77}. As in the case of black holes, cosmological horizons are causal boundaries which prevent some observer to reach future infinity, and they are realised in certain spacetimes relevant for modern Cosmology, as the de Sitter space. The difference from black hole horizons is that in the case of cosmological horizons we usually think of the observer as trapped \textit{inside} the horizon, while in the case of the black hole we take the point of view of an observer outside the black hole, and this is why the entropy associated to cosmological horizons usually comes with the "wrong" sign.

Cosmological event horizons share the same non-local, teleological nature of their black hole counterparts, and therefore, just as for dynamical black holes, a notion of cosmological apparent horizon has been developed \cite{Faraoni11}. 

Although there is a large class of cosmological spacetimes exhibiting event and apparent horizons, we restrict our attention to a particularly simple class of spacetimes in order to show how the ideas developed for the entropy of black hole horizon can be applied in this different context.

If we assume the Cosmological Principle, at least at large scale, our Universe is described by a homogeneous and isotropic Universe, and the metric is one in the class of Friedmann-Lemaitre-Robertson-Walker metric:
\begin{equation}
\dd s^2 = - \dd t^2 + a(t)^2 \bigg( \frac{1}{1 - kr} \dd r^2 + r^2 \dd \Omega^2 \bigg) \ .
\end{equation}
The metric described a spacetime foliated by 3-surfaces with constant $t$ and constant spatial curvature $k$, which can be positive, negative or zero. The \textit{cosmic time} $t$ is the proper time for particles at rest described by integral lines of the time-like vector $\partial_t$. However, it is often useful to introduce the \textit{conformal time} $\tau$ so that $\dd t = a(t) \dd \tau$. If we use the conformal time, the metric takes the form
\begin{equation}
\dd s^2 = a(t)^2 \bigg( \dd \tau^2 + \frac{1}{1-kr}\dd r^2 + r^2 \dd \Omega_2^2 \bigg) \ .
\end{equation}
By construction the cosmological time has values in an open interval $(\alpha, \beta)$ which may be infinite or not. We are interested in the case in which we can define a conformal null past infinity $\mathscr I^-$ for our spacetime, so we consider the case $\tau \in ( - \infty, 0) $.

We restrict our attention to FLRW spacetime with flat slicing, $k = 0$, in which the cosmological horizon is at infinity. Further, we restrict the possible forms of the scale factor to the case
\begin{equation}
a(\tau) = \frac{\gamma}{\tau} + \mathcal O (\frac{1}{\tau^2}) \qquad \dv{a(\tau)}{\tau} = - \frac{\gamma}{\tau^2} + \mathcal O (\frac{1}{\tau^3}) \ .
\end{equation}

As in the Schwarzschild case, to introduce a conformal boundary we need to map the physical spacetime into an unphysical one, $(\tilde{\mathcal M}, \tilde g)$ which includes the infinity in a finite interval. However, if $k = 0$ we see that, apart the scale factor $a$ in front, the metric coincides with that of Minkowski, so the compactification follows the same procedure: we introduce a new set of variables,
\begin{align*}
    V = \arctan(\tau + r) = \arctan v \\
    U = \arctan(\tau - r) = \arctan u
\end{align*}
so that in these coordinates the metric becomes
\begin{equation}
\dd s^2 = \frac{a^2(\tau)}{4 \cos^2 U \cos^2 V} \bigg( - 4 \dd U \dd V +  \sin^2{(U - V)} \dd \Omega^2_2 \bigg) \ .
\end{equation}

We see that the spatially flat FLRW Universe is mapped into the same portion of Einstein Universe as Minkowski spacetime, and therefore the conformal diagram is again a triangle. However, in this case, the conformal past infinity, corresponds to the cosmological event horizon.

\subsection{Relative Entropy}
On the past cosmological horizon we can define a symplectic form
\begin{equation}
w(\psi_1, \psi_2) = \int (\tilde \psi_2 \partial_v \tilde \psi_1 - \tilde \psi_1 \partial_v \tilde \psi_2) \dd v \dd \Omega_2 \ .
\end{equation}
Then, if we consider the vacuum with respect to translations in the $v$ direction, and a coherent perturbation generated by a classical solution $\psi$ of the Klein-Gordon equation, with initial data on a portion of the cosmological horizon $\mathscr I^-(v_0)$ given by the condition $v < v_0$, the relative entropy takes the form
\begin{equation}
S = 2 \pi \int_{\mathscr I^-(v_0)} (v_0 - v) (\partial_v \tilde \psi)^2 \ \dd v \dd \Omega_2 \ .
\end{equation}
Here we note that this vacuum is not directly related to translations of the cosmic time $t$, but rather of the conformal time $\tau$. Such a state is called \textit{Bunch-Davies vacuum} in the case of de Sitter spacetime, and plays the role of the Unruh state for black holes in cosmological contexts. As in the previous cases, if we translate the boundary of the region in which the theory is defined by an amount $v_0 \to v_0 + t$, the relative entropy varies accordingly
\begin{equation}
\dv{S}{t} = 2 \pi \int_{\mathscr I^-(v_0)} (\partial_v \tilde \psi)^2 \dd v \dd \Omega_2 \ .
\end{equation}

Now we can proceed in re-express the derivative of the relative entropy as a variation in the horizon area. Since in this background the horizon and the conformal past infinity on which the relative entropy is defined coincide, there is no need to introduce a conservation law, but we can directly rewrite the relative entropy term by means of the Raychaudhuri equation for the null congruence of the horizon generators. The normal vector to the horizon is $n = \partial_v$, and the Raychaudhuri equation then states
\begin{equation}
\dv{\theta}{v} = - 8 \pi T_{\mu \nu} n^\mu n^\nu = - 8 \pi (\partial_v \tilde \psi)^2 \ .
\end{equation}
We can substitute the right-hand side in the expression for the relative entropy, and integrating the expansion we find immediately that
\begin{equation}
\dv{S}{t} = \frac{1}{4} \int_{\mathscr I^-} \dv{\theta}{v} \dd v \dd \Omega_2 = - \frac{1}{4} \bigg( \dv{A(v_1)}{v} - \dv{A(v_0)}{v} \bigg)
\end{equation}
So that the variation of generalised entropy in this simple case is a conservation law:
\begin{equation} \label{eq:result-cosmological-horizon-RE}
    \dv{v} \bigg( S + \frac{1}{4} A \bigg) = 0 \ .
\end{equation}
    
\addcontentsline{toc}{chapter}{Conclusions}
 
\chapter*{Conclusions}

After more than 40 years from its discovery, black hole thermodynamics is still an active field of research. Because of its many connections with different branches of modern physics, from astrophysics and thermodynamics, to the nature of space and time, black hole thermodynamics plays the role of a central mystery, at the crossing point of many of the questions physicists are trying to answer.

The main question of this thesis was to take into consideration the effects of the propagation of a scalar field on a curved background on the entropy balance of a dynamical black hole. The problem has been tackled directly, computing the expression for the entropy of the scalar field and connecting it to the variation in area of the horizon, following closely the original ideas of Bekenstein and Hawking and the modern approach of Hollands and Ishibashi. In particular, we showed that the approach of computing the effects of a scalar field propagating in a stationary black hole background can be applied effectively to dynamical black holes, where a notion of entropy associated to the area of the apparent horizon naturally arises. Therefore, just as the geometric approach to black hole mechanics, the quantum approach can be generalised to apparent horizons, not only for what regards the temperature of the black hole but also its entropy.

Although Vaidya black holes are idealised models, this is the first step toward a characterisation of the thermodynamics of realistic black holes. The kind of black holes studied and observed in astrophysics are fully dynamical, axisymmetric, and in a Universe which is not necessary asymptotically flat. Therefore, one should find a description of black hole entropy based on local observables, as the apparent horizon, and without the assumptions of spherical symmetry and asymptotic flatness. It is important to understand how it can affect observations of real black holes, although it seems unlikely that the temperature of astrophysical black holes will be observed any time soon, because black hole thermodynamics could be our best hint on a quantum theory of gravity, .

In this work we followed Hawking's direct approach to black hole thermodynamics employing the modern framework of Quantum Field Theory on Curved Spacetime. The KMS characterisation of thermal states, the Tomita-Takesaki theory on algebras' modular automorphisms, and the algebraic approach to the quantization of fields in curved spacetime allowed us to derive our results in a rigorous framework, and this in turn helped to clarify the physical picture. In particular, we chose a simple background metric, and we were able to select a suitable quantum state which closely mimic the usual Minkowski vacuum. We motivate the Unruh state both on physical and mathematical grounds, and we use it in the discussion of dynamical black holes as well as stationary black holes. From the Unruh state we constructed a coherent perturbation, defined through a classical solution of the free, massless scalar field equation, and we computed the relative entropy between them by means of the Araki formula \eqref{eq:araki-formula}. We then showed that a variation of relative entropy is always accompanied by a variation of one-fourth the area of the horizon, in the static case the event horizon, in the dynamical case the apparent horizon, and that the same applies for cosmological event horizons. This method showed that the notions of black hole entropy and generalised entropy naturally arise, when one computes the effects of the propagation of a quantum field on a background with a horizon, without the necessity to appeal to the temperature of the horizon or to the analogy between the laws of thermodynamics and the geometric laws of black hole mechanics.

From this starting point, one can proceed in many different directions. To start, one could move to more general dynamical black holes, applying the same setup to the formalism for apparent horizons in spherical symmetry developed by Nielsen and Visser, \cite{Nielsen05}. The next step would be to consider axisymmetric black holes, renouncing to spherical symmetry. This way one could account for the thermodynamic properties of fully dynamical, realistic models of black holes, described by their local properties.

On the other hand, one could proceed to generalise the matter content of the theory, instead of its geometry. In this work we considered only the "most classical" quantum states, coherent states, because the modular operator takes a particularly simple form. Although in general it is difficult to find an explicit expression for the relative modular operator, there are known examples in flat spacetime in which it has been explicitly computed: for example, for free fermions \cite{Longo17}, conformal theories in 2 dimensions \cite{Hollands19}, or different equilibrium states in perturbative QFT \cite{Drago17}. One could try then to generalise these computations to curved spacetime, maybe deducing them from their properties at infinity in asymptotically flat spacetimes, as we did here, to evaluated the relative entropy for states different from coherent perturbations of the vacuum. This way one could check if the characterisation of black hole entropy via the relative entropy is dependent on the matter model considered, as entanglement entropy, or if it is a universal feature. The latter case would give strong support on the idea that the relative entropy can correctly account for the black hole entropy, since the Bekenstein-Hawking formula is a universal quantity, as shown in \cite{FredenhagenHaag89}.

Finally, one could try to go beyond the semiclassical regime and try to account for the quantum effects of gravity itself. The method we applied in this work could not account for the statistical origin of black hole entropy: we could show that a generalised entropy measure naturally arises, and from this we can argue that, for consistency with the second law of thermodynamics and with the Hawking temperature of black hole horizons, the area of the black hole is a true measure of the entropy carried by the black hole. However, we are not able to give a Boltzmann interpretation of the entropy, as the counting of the gravitational degrees of freedom which describes the black hole from a microscopic point of view. Such a result has been found in different contexts, in candidates theories of quantum gravity, as string theory \cite{StromingerVafa96} and loop quantum gravity \cite{Rovelli96}.

However, in \cite{HollandsIshibashi19}, assuming a key technical result, that is, that the decay properties of the scalar field we briefly discussed in section \ref{ssec:quantization-in-Schwarzschild} holds for linear gravitational waves too, it was shown that it is possible to find analogue formulas to \eqref{eq:result-RE-Schw} for coherent perturbations in linearized gravity. One could ask, then, if the same result holds in perturbative quantum gravity. Perturbative quantum gravity has been developed as an Effective Field Theory for a nonrenormalizable, Yang-Mills-type theory, with diffeomorphism invariance as a gauge group, in \cite{Hollands07} and \cite{Brunetti13}. In such a framework, the gravitational degrees of freedom could be taken into account directly in the relative entropy measure, rather than entering the computation in a coarse-grained, geometric way through the Raychaudhuri equation. If an entropy-area law would be found, it could point to the relative entropy between gravitational quantum states as the origin of the black hole entropy.

Black hole entropy is still probably the most mysterious property of black holes. The quest for understanding its microscopic origin has led to major results in different areas of physics, and trying to understand how quantum fields affects General Relativity has led to the development of an entirely new field of study, Quantum Field Theory in Curved Spacetimes. Still, the fact that it is possible to find a entropy-area law in a plethora of contexts and methods, and that it is not only a property of black holes, but a seemingly universal feature shared by cosmological horizons and even apparent horizons, hints to a profound connection between the quantum theory, gravity, and information.

 \clearpage 
 \singlespacing
 \addcontentsline{toc}{chapter}{Bibliography}
 
 \printbibliography 

\end{document}